\documentclass[12pt]{iopart}

\usepackage{amsfonts, amssymb, hyperref}
\def\diag{{\rm diag}}
\def\Mat{{\rm Mat}}
\def\G{{\rm G}}

\def\T{{\rm T}}
\def\SU{{\rm SU}}
\def\su{\mathfrak{su}}
\def\C{{\mathbb C}}
\def\R{{\mathbb R}}
\def\SO{{\rm SO}}
\def\OO{{\rm O}}
\def\Spin{{\rm Spin}}
\def\tr{{\rm tr}}
\def\rank{{\rm rank}}
\def\sign{{\rm sign}}
\newtheorem{theorem}{Theorem}
\newtheorem{lemma}{Lemma}

\def\Zero{{\rm O}}

\begin{document}

\title[Classification of all constant solutions of $\SU(2)$ Yang--Mills equations in $\R^{p,q}$]{Classification of all constant solutions of $\SU(2)$ Yang--Mills equations with arbitrary current in pseudo-Euclidean space $\R^{p,q}$}

\author{Dmitry Shirokov$^1$ $^2$}

\address{$^1$ HSE University, Myasnitskaya str. 20, 101000 Moscow, Russia}
\vspace{10pt}
\address{$^2$ Institute for Information Transmission Problems of Russian Academy of Sciences, Bolshoy Karetny per. 19, 127051 Moscow, Russia}
\ead{dm.shirokov@gmail.com, dshirokov@hse.ru, shirokov@iitp.ru}
\vspace{10pt}
\begin{indented}
\item[]December 2022
\end{indented}

\begin{abstract}
We present a classification and an explicit form of all constant solutions of the Yang--Mills equations with $\SU(2)$ gauge symmetry for an arbitrary constant non-Abelian current in pseudo-Euclidean space $\R^{p,q}$ of arbitrary finite dimension $n=p+q$. Using hyperbolic singular value decomposition and two-sheeted covering of orthogonal group by spin group, we solve the nontrivial system for constant solutions of the Yang--Mills equations of $3n$ cubic equations with $3n$ unknowns and $3n$ parameters in the general case. We present a new symmetry of this system of equations. All solutions in terms of the potential, strength, and invariant of the Yang--Mills field are presented. Nonconstant solutions of the Yang--Mills equations can be considered in the form of series of perturbation theory using all constant solutions as a zeroth approximation. 
\end{abstract}
\ams{70S15}

%
\vspace{2pc}
\noindent{\it Keywords}: Yang--Mills equations, hyperbolic singular value decomposition, cubic equations, constant solutions,  pseudo-Euclidean space, $\SU(2)$
%

\tableofcontents
\title[Classification of all constant solutions of $\SU(2)$ Yang--Mills equations in $\R^{p,q}$]{}
\vspace{-40pt}
%
%

\section{Introduction.}
\label{sect1}

The law of elementary particles physics is given by quantum gauge theories \cite{Fad}. We need exact solutions of classical Yang--Mills equations to describe the vacuum structure of the theory and to fully understand a quantum gauge theory. The well-known classes of solutions of the Yang--Mills equations are described in detail in classical papers \cite{WYa, tH, Pol, Bel, Wit, ADHM, deA} and various reviews \cite{Actor, Zhdanov}, etc. In this paper, we present all constant solutions of the Yang--Mills equations with $\SU(2)$ gauge symmetry for an arbitrary constant non-Abelian current in pseudo-Euclidean space $\R^{p,q}$ of arbitrary finite dimension $n=p+q$. We present a new symmetry of this system of equations (see Lemma \ref{thSym0}), which can have important physical consequences (in the continuous case too). The case of arbitrary Euclidean space $\R^n$, which is much simpler, was considered in the previous paper \cite{Shirokov1}. The Yang--Mills equations with $\SU(2)$ gauge symmetry describe electroweak interactions. The results of this paper are new and can be used to solve some problems in particle physics, in particular, to describe physical vacuum \cite{Gr, Ja}. Note the papers on constant solutions of the Yang-- Mills equations \cite{C1, C2, C3, C4, Sch, Sch2}. In these papers, the instability of these solutions is discussed, as well as their importance for describing physical vacuum\footnote{See \cite{C3}: ``... Such fluctuations will be unstable, and could give rise to a disordered vacuum. Hence, if we look at the physical ground state on a large space and time scale, the ground state may appear as an ocean of turbulence...'' and \cite{C2}: ``... Motivated by the perturbative calculation of the effective Lagrangian, speculations have been made that the vacuum is characterized by a nonzero expectation value of the operator $F_{\mu\nu}F^{\mu\nu}$. If the vacuum is assumed to be characterized by a constant magnetic field or by some modification of this configuration, interesting physical results can be
obtained, including a bound on the MIT bag constant...''}.

In this paper, we use the method of hyperbolic singular value decomposition (HSVD) to solve a nontrivial system of a special form (for constant solutions of the Yang--Mills equations) of $3n$ cubic equations with $3n$ unknowns and $3n$ parameters in the general case. The method of ordinary singular value decomposition (SVD, \cite{SVD1, SVD2}) is rather standard and is widely used in different applications (see applications in the Yang--Mills context in \cite{Simonov}). Whereas the method of hyperbolic singular value decomposition (HSVD) is  not standard. It was proposed in 1991 \cite{Bojan, Bojan2} and is used in signal and image processing, engineering, and computer science. The new version of the HSVD without the use of hyperexchange matrices is presented in our previous paper \cite{Shirokovhsvd}. This new version of HSVD naturally incudes the ordinary SVD. We use the HSVD in this paper for a physical problem.

Note that the cases where either $p$ or $q$ is equal to 1 are most important since they corresponds to the $n$-dimensional Minkowski space. However, we consider the case of arbitrary $p$ and $q$ for the sake of completeness. All possible constant solutions of the $\SU(2)$ Yang--Mills equations are presented in Tables 2, 3, and 4 depending on $p$ and $q$ and other parameters (hyperbolic singular values) of the matrix of current. Note that we need another technique to solve the same problem in the case of the Lie group $\SU(3)$, which is important for describing strong interactions. Our method, which works for the case of $\SU(2)$, essentially uses the two-sheeted covering of the orthogonal group $\SO(3)$ by the spin group $\Spin(3)\simeq \SU(2)$. We synchronize gauge transformation and pseudo-orthogonal transformation of coordinates such that the matrix of potential of the Yang--Mills field has the canonical form with respect to the HSVD. We prove that the matrix of non-Abelian current should have some modified canonical form. Choosing a specific system of coordinates and a specific gauge fixing for each constant current, we obtain all constant solutions of the $\SU(2)$ Yang--Mills equations. The methods of this paper can be modified and used to solve the same problems on curved manifolds.

The results of the this paper are consistent with the results \cite{Sch, Sch2} for the zero current and an arbitrary compact Lie algebra. In particular, in \cite{Sch2} it is proved that in the case of the zero current $J=0$, the strength of the Yang--Mills field is zero $F=0$ for all constant potentials $A$ satisfying the Yang--Mills equations in Euclidean and Lorentzian cases. We verified this fact for all constant solutions of the $\SU(2)$ Yang--Mills equations presented in this paper. Also we give an example with the explicit form of solutions with nonzero strength $F\neq 0$ and the zero current $J=0$ in all other cases $p\geq 2$ and $q\geq 2$ (see Table 3: the solutions (\ref{Ad000a}) and (\ref{Fd000a}) for $p\geq 2$ and $q\geq 2$; the solutions (\ref{Ad000b}) and (\ref{Fd000b}) for $p\geq 3$ and $q\geq 3$). The advantage of the current work is that an arbitrary (nonzero) current is considered.


\section{The main ideas.}
\label{sect2}

Let us consider pseudo-Euclidean space $\R^{p,q}$ of arbitrary finite dimension $n=p+q$, $p, q\geq 1$, $n\geq 2$. We denote Cartesian coordinates by $x^\mu$, $\mu=1, \ldots, n$ and partial derivatives by $\partial_\mu=\partial/{\partial x^\mu}$. The metric tensor of $\R^{p,q}$ is given by the diagonal matrix
\begin{eqnarray}
\eta=\diag(\underbrace{1,\ldots,1}_p, \underbrace{-1, \ldots, -1}_q)=(\eta_{\mu\nu})=(\eta^{\mu\nu}),\qquad p+q=n.\label{eta}
\end{eqnarray}
We can raise or lower indices of components of tensor fields with the aid of the metric tensor. For example, $A_\nu=\eta_{\nu\mu}A^\mu$.

Let us consider the Lie group $\SU(2)=\{S\in\Mat(2,\C):\,  S^\dagger S =I_2,\, \det S=1\}$ and the corresponding Lie algebra
$\su(2)=\{S\in\Mat(2,\C):\, S^\dagger=-S,\, \tr S=0\}$.

Consider the $\SU(2)$ Yang--Mills equations
\begin{eqnarray}
&&\partial_\mu A_\nu-\partial_\nu A_\mu-[A_\mu, A_\nu]=:F_{\mu\nu},\label{YM1}\\
&&\partial_\mu F^{\mu\nu} -[A_\mu, F^{\mu\nu}]=J^\nu,\label{YM2}
\end{eqnarray}
where $A^\mu: \R^{p,q} \to \su(2)$ is the potential of the Yang--Mills field, $J^\nu: \R^{p, q}\to \su(2)$ is the non-Abelian current. The tensor field $F_{\mu\nu}=-F_{\nu\mu}: \R^{p,q} \to \su(2)$ with the definition (\ref{YM1}) is the strength of the Yang--Mills field.

Note that the Yang--Mills equations (\ref{YM2}) can be obtained in the conventional way on the basis of the variational principle. Consider the action $S=\int L dx$ for the Lagrangian 
\begin{eqnarray}
L=-\frac{1}{4}\tr(F^2),\qquad F^2=F_{\mu\nu}F^{\mu\nu},\label{lagr}
\end{eqnarray}
where $F_{\mu\nu}$ are the components of the curvature $2$-form with respect to the connection $A_\mu$, i.e., they are coupled by definition by (\ref{YM1}). By varying the action, we obtain (\ref{YM2}) with a zero current $J^\nu=0$. The current $J^\nu\neq 0$ in (\ref{YM2}) appears when terms related to other (scalar or spinor) fields are added to Lagrangian (\ref{lagr}).

We can substitute (\ref{YM1}) into (\ref{YM2}) and obtain
\begin{eqnarray}
\partial_\mu(\partial^\mu A^\nu-\partial^\nu A^\mu- [A^\mu,A^\nu])- [A_\mu,\partial^\mu A^\nu-\partial^\nu A^\mu-[A^\mu,A^\nu]]=J^\nu.\label{YM4}
\end{eqnarray}
The current $J^\nu$ satisfies the following non-Abelian conservation law
\begin{eqnarray}
\partial _\nu J^\nu-[A_\nu, J^\nu]=0.
\end{eqnarray}
The system (\ref{YM1}) - (\ref{YM2}) is gauge invariant with respect to the following transformation
\begin{eqnarray}
&&A^\mu \to  S^{-1}A^\mu S-S^{-1}\partial_\mu S,\qquad F_{\mu\nu}\to S^{-1}F_{\mu\nu}S,\label{gauge}\\
&&J^\mu \to S^{-1} J^\mu S,\qquad S:\R^{p,q}\to\SU(2).\nonumber
\end{eqnarray}
We get the following system of algebraic equations for the constant solutions of the $\SU(2)$ Yang--Mills equation from (\ref{YM4}):
\begin{eqnarray}
[A_\mu,[A^\mu, A^\nu]]=J^\nu,\qquad \nu=1, \ldots, n,\label{YMconst}
\end{eqnarray}
where $A^\nu$ and $J^\nu$ are components of tensor fields with values in the Lie algebra $\su(2)$. One suggests that $A^\mu$ are unknown and $J^\nu$ is known.

We take the basis $\tau^a=\frac{1}{2i}\sigma^a$, $a=1, 2,3$ of $\su(2)$, where $\sigma^a$ are the Pauli matrices with $[\tau^a, \tau^b]=\epsilon^{ab}_{\,\,\,\,\,c}\tau^c$, where $\epsilon^{ab}_{\,\,\,\,\,c}$ is the antisymmetric Levi-Civita symbol, and represent the potential and the current in the form
\begin{eqnarray}
A^\mu=A^\mu_{\,\,a} \tau^a,\qquad J^\mu=J^\mu_{\,\,a} \tau^a,\qquad A^\mu_{\,\, a}, J^\mu_{\,\, a}\in\R.\nonumber
\end{eqnarray}
We obtain the following system of $3n$ equations ($k=1, 2, 3$, $\nu=1, 2, \ldots, n$) for $3n$ unknowns $A^\nu_{\,\,k}$ and $3n$ parameters $J^\nu_{\,\, k}$
\begin{eqnarray}
&&A_{\mu c} A^\mu_{\,\,a} A^\nu_{\,\,b}\epsilon^{ab}_{\,\,\,\,\,d}\epsilon^{cd}_{\,\,\,\,\,k}=J^\nu_{\,\,k},\qquad \nu=1, \ldots, n,\qquad k=1, 2, 3.\label{eq}
\end{eqnarray}
In this paper, we present a general solution of the nontrivial system of cubic equations (\ref{eq}) in the general case. We can consider (\ref{eq}) as the system of equations for the components of two real matrices $A_{n \times 3}=(A^\nu_{\,\, k})$ and $J_{n \times 3}=(J^\nu_{\,\, k})$. The case of arbitrary Euclidean space has been considered in \cite{Shirokov1}. Now we are interested in the case of arbitrary pseudo-Euclidean space $\R^{p,q}$, $p\geq 1$, $q\geq 1$.

Covariantly constant solutions of the Yang--Mills equations are discussed in \cite{YMM, PFE, YMSh}, constant solutions of the Yang--Mills--Proca equations are discussed in \cite{ProcaYM} and in \cite{YMP, hforms} using Clifford algebra formalism.

Let us formulate the HSVD from \cite{Shirokovhsvd} for the case of real matrices $A\in\Mat_{n\times N}(\R)$. Denote any zero block of the matrix by $\Zero$.

\begin{theorem} \cite{Shirokovhsvd}\label{thHSVD}
Assume (\ref{eta}), $p+q=n$. For an arbitrary matrix $A\in\Mat_{n\times N}(\R)$, there exist
matrices $R\in\OO(N)$ and $L\in\OO(p,q)$
such that
\begin{eqnarray}
L^\T A R=\Sigma^A,\qquad
\Sigma^A=\left(
      \begin{array}{cccc}
        X_{x} & \Zero & \Zero & \Zero \\
        \Zero & \Zero & I_d & \Zero \\
        \Zero & \Zero & \Zero & \Zero \\
        \hline
        \Zero & Y_y & \Zero & \Zero \\
        \Zero & \Zero & I_d & \Zero \\
        \Zero & \Zero & \Zero & \Zero\\
      \end{array}
    \right)\!\!\!\!\!\!\!\!\!\begin{array}{c}
      \left.
      \begin{array}{c}
          \\
         \\
          \\
      \end{array}
    \right\}p \\
      \left.
      \begin{array}{c}
          \\
         \\
          \\
      \end{array}
    \right\}q
    \end{array}\in\Mat_{n\times N}(\R),\label{DD2}
\end{eqnarray}
where the first block of the matrix $\Sigma^A$ has $p$ rows and the second block has $q$ rows, $X_{x}$ and $Y_y$ are diagonal matrices of corresponding dimensions $x$ and $y$ with all positive uniquely determined diagonal elements (up to a permutation), $I_d$ is the identity matrix of dimension $d$.

Moreover, choosing $R$, one can swap columns of the matrix $\Sigma^A$. Choosing $L$, one can swap rows in individual blocks but not across blocks. Thus we can always arrange diagonal elements of the matrices $X_{x}$ and $Y_y$ in decreasing order.

Here we have
$$d=\rank (A)-\rank (A^\T \eta A),\qquad x+y=\rank (A^\T \eta A),$$
$x$ is the number of positive eigenvalues of the matrix $A^\T \eta A$, $y$ is the number of negative eigenvalues of the matrix $A^\T \eta A$.
\end{theorem}

Let us call $\Sigma^A$ (\ref{DD2}), where all diagonal elements of the matrices $X_{x}$ and $Y_y$ are positive and in decreasing order, the \emph{canonical form} of the matrix $A\in\Mat_{n\times N}(\R)$. The canonical form is uniquely determined for any matrix $A\in\Mat_{n\times N}(\R)$, the corresponding matrices $L$ and $R$ are not uniquely determined. One can find the algorithm for computing the HSVD in \cite{Shirokovhsvd}. For arbitrary matrix $A\in\Mat_{n\times N}(\R)$, we can always find the matrices $L$, $R$, and $\Sigma^A$ from the formulation of Theorem \ref{thHSVD}.

\begin{lemma}\label{lemma1} The system of equations (\ref{eq}) is invariant under the transformation
$$A \to \acute{A}=AP,\qquad J \to \acute{J}=J P,\qquad P\in\SO(3)$$
and under the transformation
$$A \to \hat{A}=QA,\qquad J \to \acute{J}=Q J,\qquad Q\in\OO(p,q).$$
\end{lemma}
\begin{proof} The proof is similar to the proof of Lemma 1 in \cite{Shirokov1}. The system (\ref{YMconst}) is invariant under the transformation
\begin{eqnarray}
\acute A_\mu = S^{-1}A_\mu S,\qquad \acute J^{\nu} = S^{-1}J^{\nu}S,\qquad S\in\G=\SU(2),\nonumber
\end{eqnarray}
because of the gauge invariance (\ref{gauge}) and the fact that an element $S\in\SU(2)$ does not depend on $x\in\R^{p,q}$ now.

Let us use the theorem on two-sheeted covering of the orthogonal group $\SO(3)$ by the spin group $\Spin(3)\simeq\SU(2)$. For an arbitrary matrix $P=(p^a_b)\in\SO(3)$, there exist two matrices $\pm S\in\SU(2)$ such that
$$S^{-1}\tau^a S=p^a_b \tau^b.$$
We conclude that the system (\ref{eq}) is invariant under the transformation
\begin{eqnarray}
&&\acute{A^\mu}=S^{-1}A^\mu_{\,\,a} \tau^a S=A^\mu_{\,\,a} S^{-1}\tau^a S=A^\mu_{\,\,a} p^a_b \tau^b=\acute{A^\mu_{\,\,b}} \tau^b,\qquad \acute{A^\mu_{\,\,b}}=A^\mu_{\,\,a} p^a_b,\nonumber\\
&&\acute{J^\mu}=S^{-1}J^\mu_{\,\,a} \tau^a S=J^\mu_{\,\,a} S^{-1}\tau^a S=J^\mu_{\,\,a} p^a_b \tau^b=\acute{J^\mu_{\,\,b}} \tau^b,\qquad \acute{J^\mu_{\,\,b}}=J^\mu_{\,\,a} p^a_b.\nonumber
\end{eqnarray}

The Yang--Mills equations are invariant under the pseudo-orthogonal transformations of coordinates. Namely, let us consider the transformation $x^\mu \to \hat{x^\mu}=q^\mu_\nu x^\nu$, where $Q=(q^\mu_\nu)\in\OO(p,q).$ The system (\ref{eq}) is invariant under the transformation
\begin{eqnarray}
&&\hat{A^\nu}=q^\nu_\mu A^\mu=q^\nu_\mu A^\mu_{\,\,a} \tau^a=\hat{A^\nu_{\,\,a}} \tau^a,\qquad \hat{A^\nu_{\,\,a}}=q^\nu_\mu A^\mu_{\,\,a},\nonumber\\
&&\hat{J^\nu}=q^\nu_\mu J^\mu=q^\nu_\mu J^\mu_{\,\,a} \tau^a=\hat{J^\nu_{\,\,a}} \tau^a,\qquad \hat{J^\nu_{\,\,a}}=q^\nu_\mu J^\mu_{\,\,a}.\nonumber
\end{eqnarray}
The lemma is proved. $\blacksquare$
\end{proof}

Combining gauge and pseudo-orthogonal transformations, we conclude that the system (\ref{eq}) is invariant under the transformation
\begin{eqnarray}
&&A^\nu_b \to  \acute{\hat{A^\nu_{\,\,b}}}=q^\nu_\mu A^\mu_{\,\,a} p^a_b,\quad A_{n\times 3}\to \acute{\hat{A}}_{n \times 3}=Q_{n\times n}A_{n\times 3}P_{3\times3},\label{tr2}\\
&&J^\nu_b \to  \acute{\hat{J^\nu_{\,\,b}}}=q^\nu_\mu J^\mu_{\,\,a} p^a_b,\quad J_{n\times 3}\to \acute{\hat{J}}_{n \times 3}=Q_{n\times n}J_{n\times 3}P_{3\times3}\nonumber
\end{eqnarray}
for any $P\in\SO(3)$ and $Q\in\OO(p, q)$.

\begin{theorem}\label{th1} Assume that $A=(A^\nu_{\,\,k})$, $J=(J^\nu_{\,\, k})$ satisfy the system of $3n$ cubic equations (\ref{eq}). Then:

1) There exist matrices $P\in\SO(3)$ and $Q\in\OO(p,q)$ such that the matrix $QAP$ is in the canonical form (with parameters $x_A$, $y_A$, $d_A$)\footnote{Note that this matrix and subsequent matrices have 3 columns. This means that some blocks of matrices have size zero.
}
\begin{eqnarray*}
\Sigma^A=\acute{\hat{A}}_{n \times 3}=\left(
      \begin{array}{cccc}
        X_{x_A} & \Zero & \Zero & \Zero \\
        \Zero & \Zero & I_{d_A} & \Zero \\
        \Zero & \Zero & \Zero & \Zero \\
        \hline
        \Zero & Y_{y_A} & \Zero & \Zero \\
        \Zero & \Zero & I_{d_A} & \Zero \\
        \Zero & \Zero & \Zero & \Zero\\
      \end{array}
    \right)\!\!\!\!\!\!\!\!\!\begin{array}{c}
      \left.
      \begin{array}{c}
          \\
         \\
          \\
      \end{array}
    \right\}p \\
      \left.
      \begin{array}{c}
          \\
         \\
          \\
      \end{array}
    \right\}q
    \end{array}.
\end{eqnarray*}
For all such matrices $P$ and $Q$, the matrix $QJP$ has the following form
\begin{eqnarray}
\Phi^J:=\acute{\hat{J}}_{n \times 3}=\left(
      \begin{array}{cccc}
        Z_{x_A} & \Zero & \Zero & \Zero \\
        \Zero & \Zero & \alpha I_{d_A} & \Zero \\
        \Zero & \Zero & \Zero & \Zero \\
        \hline
        \Zero & W_{y_A} & \Zero & \Zero \\
        \Zero & \Zero & \alpha I_{d_A} & \Zero \\
        \Zero & \Zero & \Zero & \Zero\\
      \end{array}
    \right)\!\!\!\!\!\!\!\!\!\begin{array}{c}
      \left.
      \begin{array}{c}
          \\
         \\
          \\
      \end{array}
    \right\}p \\
      \left.
      \begin{array}{c}
          \\
         \\
          \\
      \end{array}
    \right\}q
    \end{array},\label{xx4}
\end{eqnarray}
where elements of the diagonal matrices $Z$ and $W$ are real numbers (can be zero), $\alpha\in\R$ (can be zero).

2) For the parameters of the matrices $A$ and $J$, we have:
\begin{eqnarray}
&&x_J \leq x_A,\nonumber\\
&&y_J \leq y_A,\nonumber\\
&&d_J=d_A > 0 \qquad\mbox{or}\qquad d_J=0,\, d_A\geq 0.\nonumber
\end{eqnarray}

3) There exist matrices $P\in\SO(3)$ and $Q\in\OO(p,q)$ such that the matrix $QJP$ is in the canonical form (with parameters $x_J$, $y_J$, $d_J$)
\begin{eqnarray*}
\Sigma^J=\left(
      \begin{array}{cccc}
        X_{x_J} & \Zero & \Zero & \Zero \\
        \Zero & \Zero & I_{d_J} &\Zero \\
        \Zero & \Zero & \Zero & \Zero \\
        \hline
        \Zero & Y_{y_J} & \Zero & \Zero \\
        \Zero & \Zero & I_{d_J} & \Zero \\
        \Zero & \Zero &\Zero & \Zero\\
      \end{array}
    \right)\!\!\!\!\!\!\!\!\!\begin{array}{c}
      \left.
      \begin{array}{c}
          \\
         \\
          \\
      \end{array}
    \right\}p \\
      \left.
      \begin{array}{c}
          \\
         \\
          \\
      \end{array}
    \right\}q
    \end{array},
\end{eqnarray*}
and $QAP$ has the following form
\begin{eqnarray}\Psi^A:=\left(
      \begin{array}{ccccccc}
        K_{x_J} & \Zero & \Zero & \Zero & \Zero & \Zero & \Zero \\
        \Zero & \Zero & \beta I_{d_J} & \Zero & \Zero & \Zero & \Zero \\
        \Zero & \Zero & \Zero & L_{x_A-x_J} & \Zero & \Zero & \Zero \\
        \Zero & \Zero & \Zero & \Zero & \Zero & \Zero & I_{d_A-d_J} \\
        \Zero & \Zero & \Zero & \Zero & \Zero & \Zero & \Zero \\
        \hline
        \Zero & M_{y_J} & \Zero & \Zero & \Zero & \Zero & \Zero \\
        \Zero & \Zero & \beta I_{d_J} & \Zero & \Zero & \Zero & \Zero \\
        \Zero & \Zero & \Zero & \Zero & N_{y_A-y_J} & \Zero & \Zero \\
        \Zero & \Zero & \Zero & \Zero & \Zero & \Zero & I_{d_A-d_J}  \\
        \Zero & \Zero & \Zero & \Zero & \Zero & \Zero & \Zero \\
      \end{array}
    \right)\!\!\!\!\!\!\!\!\!\begin{array}{c}
      \left.
      \begin{array}{c}
          \\
         \\
          \\
          \\
          \\
      \end{array}
    \right\}p \\
      \left.
      \begin{array}{c}
          \\
         \\
          \\
          \\
          \\
      \end{array}
    \right\}q
    \end{array},\label{xx5}
\end{eqnarray}
where $\beta\in \R\setminus\{0\}$; elements of the diagonal matrices $K$, $L$, $M$, $N$ are arbitrary nonzero real numbers.
\end{theorem}

\begin{proof} 1) Assume that the system (\ref{eq}) has some solution $A^\mu_{\,\,a}$, $J^\mu_{\,\,a}$. Let us synchronize gauge transformation and pseudo-orthogonal transformation such that $A=(A^\mu_{\,\, a})$ is in the canonical form (\ref{DD2}). Namely, we take $P\in\SO(3)$ and $Q\in\OO(p,q)$ such that $QAP=\Sigma^A$. Note that we can always find the matrix $P\in\SO(3)$ from the special orthogonal group in the HSVD. If it has the determinant $-1$, then we can change the sign of all elements of the matrices $P$ and $Q$.

Note that the matrix $\Sigma^A$ has at most one nonzero entry in each row. We must take $a=c$ in (\ref{eq}) and $b=k\neq a=c$ for Levi-Civita symbol to obtain nonzero summands. The product of two Levi-Civita symbols in (\ref{eq}) equals $-1$. We obtain the following expression
$$-\acute{\hat{A}}^\nu_{\,\, k}\sum_{a=1; \neq k}^3 \, \sum_{\mu=1}^n \eta^{\mu\mu} (\acute{\hat{A}}^\mu_{\,\, a})^2.$$
Note that for each fixed $k$, we have the sum of $\pm$ squares of all elements of the matrix $\acute{\hat{A}}$ but not in the $k$-th column. Summands from the identity blocks $I_{d}$ are reduced because of different signs. Thus we have no more than two summands in this sum for each pair $\nu$ and $k$.

We see that if $\acute{\hat{A}}^\nu_{\,\, k}=0$, then $\acute{\hat{J}}^\nu_{\,\, k}=0$. If we have two elements $\acute{\hat{A}}^{\nu_i}_{\,\,k}=1$ for some $i=1, 2$, then the corresponding elements $\acute{\hat{J}}^{\nu_i}_{\,\, k}$, $i=1, 2$ must be the same (not necessarily equal to $1$). That is why the identical diagonal blocks $\alpha I_d$, $\alpha\in\R$ are allowed for the matrix $QJP$ instead of the blocks $I_d$. Finally, we obtain the special form (\ref{xx4}) of the matrix $QJP$ and the system of no more than three different equations
\begin{eqnarray}
-\acute{\hat{A}}^\nu_{\,\, k}\sum_{a=1; \neq k}^3 \, \sum_{\mu=1}^n \eta^{\mu\mu} (\acute{\hat{A}}^\mu_{\,\, a})^2=\acute{\hat{J}}^\nu_{\,\, k}.\nonumber
\end{eqnarray}

2) - 3) We can change the matrices $Q\in\OO(p,q)$ and $P\in\SO(3)$ such that the new matrix $QJP$ will be in the canonical form because:
\begin{enumerate}
  \item we can change the order of columns of the matrix $QJP$ by multiplying $P$ by the  $\pm$ permutation matrix (see Theorem \ref{thHSVD} and \cite{Shirokovhsvd} for details), which belongs to the group $\SO(3)$;
  \item we can change the order of rows in each of the two blocks of the matrix $QJP$ by multiplying $Q$ by the matrix of the form
$$
\left(
      \begin{array}{c|c}
        S_1 & \Zero \\ \hline
        \Zero & S_2  \\
      \end{array}
    \right)\in\OO(p,q),
$$
where $S_1$ and $S_2$ are permutation matrices of the dimensions $p$ and $q$ respectively (see Theorem \ref{thHSVD} and \cite{Shirokovhsvd});
  \item we can change signs of some elements of the matrices $Z$ and $W$ by multiplying the matrix $Q\in\OO(p,q)$ by the diagonal matrix with $\pm 1$ on the diagonal, which is also belongs to the group $\OO(p,q)$;
  \item the blocks $\alpha I_d$ (in the cases $\alpha \neq 0$ and $\alpha\neq 1$) can be reduced to the blocks $I_d$ by multiplying the matrix $Q\in\OO(p,q)$ by the matrix
$$\left(
      \begin{array}{ccc|ccc}
        I_x & \Zero & \Zero & \Zero & \Zero &\Zero \\
        \Zero & \frac{\alpha+\frac{1}{\alpha}}{2} I_d & \Zero & \Zero & \frac{\frac{1}{\alpha}-\alpha}{2} I_d &\Zero \\
        \Zero & \Zero & I_{p-x-d} & \Zero & \Zero &\Zero \\
        \hline
        \Zero & \Zero & \Zero & I_y & \Zero &\Zero \\
        \Zero & \frac{\frac{1}{\alpha}-\alpha}{2} I_d & \Zero & \Zero & \frac{\alpha+\frac{1}{\alpha}}{2} I_d &\Zero \\
        \Zero & \Zero & \Zero & \Zero & \Zero &I_{q-y-d}\\
      \end{array}
    \right
 )\in\OO(p,q).$$
\end{enumerate}
Using the same transformations (i) - (iv) for the matrix $A$, we obtain the specific form (\ref{xx5}) of the matrix $QAP$. Namely, the order of columns will change respectively, the order of rows in each of two blocks of the matrix $QAP$ will change respectively, some elements of the matrices $X$ and $Y$ will become negative. In the cases $\alpha\neq 0$ and $\alpha\neq 1$, the blocks $I_{d_A}$ will be multiplied by $\beta:=\frac{1}{\alpha}$. $\blacksquare$
\end{proof}

\textbf{Remark 1.} Suppose we have the known matrix $J_{n\times 3}=(J^\nu_{\,\,k})$ and want to obtain all solutions $A_{n\times 3}=(A^\nu_{\,\,k})$ of the system (\ref{eq}). Let us denote the canonical form of the known matrix $J$ for some $Q\in\OO(p,q)$, $P\in\SO(3)$ by
$$QJP=\Sigma^J.$$
We use the statements 2) and 3) of Theorem \ref{th1} and Lemma \ref{lemma1}. The system (\ref{eq}) takes the following form under the transformation (\ref{tr2}) of no more than 3 independent equations for (each column) $k=1, 2, 3$ and the unique corresponding $\nu=\nu(k)$:
\begin{eqnarray}
-A^\nu_{\,\, k}\sum_{a=1; \neq k}^3 \, \sum_{\mu=1}^n \eta^{\mu\mu} (A^\mu_{\,\, a})^2=J^\nu_{\,\, k}.\label{eqaa}
\end{eqnarray}
Below we present a general solution of the system (\ref{eqaa}) in different cases.

Finally, we obtain all solutions $\Psi^A$ of the system (\ref{eq}) but in the other system of coordinates depending on $Q\in\OO(p,q)$ and with gauge fixing depending on $P\in\SO(3)$. The matrix
$$A=Q^{-1} \Psi^A P^{-1}$$
will be solution of the system (\ref{eq}) in the original system of coordinates and with the original gauge fixing.

\textbf{Remark 2.} Note that $Q^{-1} Q_1^{-1} \Psi^A P_1^{-1} P^{-1}$ for all $Q_1\in\OO(p,q)$ and $P_1\in\SO(3)$ such that $Q_1 \Sigma^J P_1=\Sigma^J$ will be also solutions of the system (\ref{eq}) in the original system of coordinates and with the original gauge fixing because of Lemma \ref{lemma1}.

Let us give one example. If $J=0$, then all hyperbolic singular values of this matrix equal zero and we can take identity matrices $Q=I_n$, $P=I_3$ for its HSVD. We solve the system (\ref{eqaa}) for $J=0$ and obtain all solutions $\Psi^A$ of this system. We have $Q_1 \Sigma^J P_1=\Sigma^J$ for $\Sigma^J=0$ and any $Q_1\in\OO(p,q)$, $P_1\in\SO(3)$. Therefore, the matrices $Q_1 \Psi^A P_1$ for all $Q_1\in\OO(p, q)$ and $P_1\in\SO(3)$ are solutions of the system (\ref{eq}) because of Lemma \ref{lemma1} in this case.

\textbf{Remark 3.} In the case of constant potential of the Yang--Mills field, we have the following expression for the strength
\begin{eqnarray}
F^{\mu\nu}=- [A^\mu, A^\nu]=-  [A^\mu_{\,\,a} \tau^a, A^\nu_{\,\,b} \tau^b]=-  A^\mu_{\,\,a} A^\nu_{\,\,b} \epsilon^{ab}_{\,\,\,\,\, c}\tau^c=F^{\mu\nu}_{\,\,\,\,\, c}\tau^c.\label{str}
\end{eqnarray}
For each solution $A$, we present all nonzero components of the strength (we write expressions for only one of the two components $F^{\mu\nu}=-F^{\nu\mu}$). Also we calculate the invariant $F^2=F_{\mu\nu}F^{\mu\nu}$, which is important from a physical point of view and is present in the Lagrangian (\ref{lagr}) of the Yang--Mills field. We have
$$
F^2=-\frac{1}{2}\sum_{\mu<\nu} \eta^{\mu\mu}\eta^{\nu\nu} (F^{\mu\nu}_a)^2 I_2.
$$
From the equations (\ref{eqaa}), it follows that we have nonzero strength $F\neq 0$ only for the solutions $A$ with $d_A+x_A+y_A\geq 2$, and we have nonzero invariant $F^2\neq 0$ only for the solutions $A$ with $x_A+y_A\geq 2$.

Let us consider the cases $n=2$ and $n\geq 3$ separately below.

\section{The results for the case \texorpdfstring{$n=2$}{n=2}: \texorpdfstring{$\R^{1,1}$}{R11}.}
\label{sect3}

Let us consider the case of pseudo-Euclidean space of dimension $n=2$ ($p=q=1$). Note that two dimensional Yang--Mills theory is discussed in many papers (\cite{GN} and others). This case is much simpler than the case $n\geq 3$. We discuss the case $n=2$ in details for the sake of completeness.

The parameters $x_J, y_J, d_J$ of the matrix $J$ can take the values $0, 1$ such that $x_J+y_J+2d_J \leq 2$.

\subsection{The case \texorpdfstring{$d_J=0$}{dJ0}.}

Let us consider the nondegenerate case $d_J=0$. We will see below that the problem actually is reduced to solving the following system
  \begin{eqnarray}
b_1(b_2)^2=j_1,\qquad b_2(b_1)^2=j_2,\label{casen22}
\end{eqnarray}
where $b_1=\pm a_1$, $b_2= \pm a_2$ for some potentials $a_1$, $a_2$. The general solution of the system (\ref{casen22}) is discussed in \cite{Shirokov1}.

\begin{lemma} \label{thn2}\cite{Shirokov1} The system of equations (\ref{casen22}) has the following general solution:
\begin{enumerate}
  \item in the case $j_1=j_2=0$, has the solutions
$(b_1, 0)$, $(0, b_2)$ for all $b_1, b_2\in \R$;
  \item in the cases $j_1=0$, $j_2\neq 0$; $j_1\neq 0$, $j_2=0$, has no solutions;
  \item in the case $j_1\neq 0$, $j_2\neq 0$,
   has the unique solution
  $$b_1=\sqrt[3]{\frac{j_2^2}{j_1}},\qquad b_2=\sqrt[3]{\frac{j_1^2}{j_2}}.$$
\end{enumerate}
\end{lemma}

\subsubsection{The case \texorpdfstring{$d_J=0$}{dJ=0}, \texorpdfstring{$\rank(J)=2$}{rankJ=2}.}

If $d_J=0$, $x_J=y_J=1$, then we have
$$\Sigma^J=\left(
                 \begin{array}{ccc}
                   j_1 & 0 & 0 \\ \hline
                   0 & j_2 & 0 \\
                 \end{array}
               \right),\qquad j_1, j_2 \neq 0.
$$
Then $x_A=y_A=1$, $d_A=0$ and we are looking for a solution in the form
$$
\Psi^A=\left(
                 \begin{array}{ccc}
                   a_1 & 0 & 0 \\ \hline
                   0 & a_2 & 0 \\
                 \end{array}
               \right),\qquad a_1, a_2 \neq 0.
$$
We obtain the system
$$a_1(a_2)^2=j_1,\qquad - a_2(a_1)^2=j_2.$$
Using Lemma \ref{thn2}, we obtain the solution
\begin{eqnarray}
\Psi^A=\left(
                 \begin{array}{ccc}
                   a_1 & 0 & 0 \\ \hline
                   0 & a_2 & 0 \\
                 \end{array}
               \right),\qquad a_1=\sqrt[3]{\frac{j_2^2}{j_1}},\qquad a_2=-\sqrt[3]{\frac{j_1^2}{j_2}}.\label{An2}
\end{eqnarray}
We get the following nonzero components of the strength
\begin{eqnarray}
F^{12}=-F^{21}=\sqrt[3]{j_1 j_2}\tau^3,\label{Fn2}
\end{eqnarray}
using specific coordinates and gauge fixing, where $j_1$ and $j_2$ are hyperbolic singular values of the matrix $J=(J^{\nu}_{\,\,k})$. We obtain the following expression for the invariant
\begin{eqnarray}
F^2=F_{\mu\nu}F^{\mu\nu}=\frac{1}{2}\sqrt[3]{(j_1 j_2)^2} I_2\neq 0.\label{F2n2}
\end{eqnarray}

\subsubsection{The case \texorpdfstring{$d_J=0$}{dJ=0}, \texorpdfstring{$\rank(J)=1$}{rankJ=1}.}

If $d_J=0$, $x_J=1$, $y_J=0$, then we have
$$\Sigma^J=\left(
                 \begin{array}{ccc}
                   j_1 & 0 & 0 \\ \hline
                   0 & 0 & 0 \\
                 \end{array}
               \right),\qquad j_1\neq 0.
$$
We have $d_A=0$, $x_A=1$, $y_A=0$ or $d_A=0$, $x_A=1$, $y_A=1$, i.e.
$$
\Psi^A=\left(
                 \begin{array}{ccc}
                   a_1 & 0 & 0 \\ \hline
                   0 & 0 & 0 \\
                 \end{array}
               \right)\qquad \mbox{or}\qquad
  \left(
                 \begin{array}{ccc}
                   a_1 & 0 & 0 \\ \hline
                   0 & a_2 & 0 \\
                 \end{array}
               \right),\qquad a_1, a_2 \neq 0.
$$
We obtain $-a_1 0 =j_1$ or $-a_1 (a_2)^2=j_1$, $a_2 (a_1)^2=0$, and a contradiction. There is no solution in this case.

If $d_J=0$, $x_J=0$, $y_J=1$, then we have
$$\Sigma^J=\left(
                 \begin{array}{ccc}
                   0 & 0 & 0 \\ \hline
                   j_1 & 0 & 0 \\
                 \end{array}
               \right),\qquad j_1\neq 0.
$$
We have $d_A=0$, $x_A=0$, $y_A=1$ or $d_A=0$, $x_A=1$, $y_A=1$, i.e.
$$
\Psi^A=\left(
                 \begin{array}{ccc}
                   0 & 0 & 0 \\ \hline
                   a_1 & 0 & 0 \\
                 \end{array}
               \right)\qquad\mbox{or}\qquad
  \left(
                 \begin{array}{ccc}
                   0 & a_2 & 0 \\ \hline
                   a_1 & 0 & 0 \\
                 \end{array}
               \right),\qquad a_1, a_2 \neq 0.
$$
We obtain $a_1 0 =j_1$ or $a_1 (a_2)^2=j_1$, $-a_2 (a_1)^2=0$, and a contradiction. There is no solution in this case too.

\subsubsection{The case \texorpdfstring{$d_J=0$}{dJ=0}, \texorpdfstring{$\rank(J)=0$}{rankJ=0}.}

If $d_J=0$, $x_J=y_J=0$, then
$$
\Sigma^J=\left(
                 \begin{array}{ccc}
                   0 & 0 & 0 \\ \hline
                   0 & 0 & 0 \\
                 \end{array}
               \right).
$$
We have 1) $d_A=x_A=y_A=0$, or 2) $d_A=0$, $x_A=1$, $y_A=0$, or 3) $d_A=0$, $x_A=0$, $y_A=1$, or 4) $d_A=0$, $x_A=y_A=1$, or 5) $d_A=1$, $x_A=y_A=0$. In the first case, we have zero solution $A=0$. In the second case, we have the solutions
\begin{eqnarray}
 \Psi^A=\left(
                 \begin{array}{ccc}
                   a_1 & 0 & 0 \\ \hline
                   0 & 0 & 0 \\
                 \end{array}
               \right), \qquad a_1\in\R\setminus \{0\}.\label{An201}
\end{eqnarray}
In the third case, we have the solutions
\begin{eqnarray}
 \Psi^A=\left(
                 \begin{array}{ccc}
                   0 & 0 & 0 \\ \hline
                   a_1 & 0 & 0 \\
                 \end{array}
               \right), \qquad a_1\in\R\setminus \{0\}.\label{An202}
\end{eqnarray}
In the fourth case, we are looking for a solution in the form
$$
 \Psi^A=\left(
                 \begin{array}{ccc}
                   a_1 & 0 & 0 \\ \hline
                   0 & a_2 & 0 \\
                 \end{array}
               \right),\qquad a_1, a_2\neq 0,
$$
get the system $a_1(a_2)^2=0$, $-a_2(a_1)^2=0$, and a contradiction. There is no solution in this case. In the fifth case, we obtain the solution
\begin{eqnarray}
\Psi^A=\left(
                 \begin{array}{ccc}
                   1 & 0 & 0 \\ \hline
                   1 & 0 & 0 \\
                 \end{array}
               \right).\label{An203}
\end{eqnarray}
For all these potentials $\Psi^A$, we have zero strength $F=0$. This also follows from the results of \cite{Sch}: in Euclidean and Lorentzian cases, $F=0$ for $J=0$. This is not so in the case of other signatures, see the example below (see Table 3: the solutions (\ref{Ad000a}) and (\ref{Fd000a}) for $p\geq 2$ and $q\geq 2$; the solutions (\ref{Ad000b}) and (\ref{Fd000b}) for $p\geq 3$ and $q\geq 3$).

\subsection{The case \texorpdfstring{$d_J=1$}{dJ=1}.}

If $d_J=1$, $x_J=y_J=0$, then we have
$$\Sigma^J=\left(
                 \begin{array}{ccc}
                   1 & 0 & 0 \\ \hline
                   1 & 0 & 0 \\
                 \end{array}
               \right).
$$
Then we have $d_A=1$, $x_A=y_A=0$, and we are looking for a solution in the form
\begin{eqnarray}
\Psi^A=\left(
                 \begin{array}{ccc}
                   \beta & 0 & 0 \\ \hline
                   \beta & 0 & 0 \\
                 \end{array}
               \right),\qquad \beta\neq 0,\nonumber
\end{eqnarray}
and obtain equation $\beta \, 0 = 1$, i.e. there is no solution.

\subsection{Summary for the case \texorpdfstring{$\R^{1,1}$}{R11}.}

The results are summarized in Table 1. In the case $n=2$, $p=q=1$, the $\SU(2)$ Yang--Mills equations have nonzero strength $F\neq 0$ only in the case
$$d_J=0, \qquad x_J=y_J=1,$$
i.e. $\rank (J)=\rank (J^\T \eta J)=2$, and the number of positive eigenvalues and the number of negative eigenvalues of the matrix $J^\T \eta J$ is the same and equals 1.

\begin{table}[ht]\caption{Constant solutions of the $\SU(2)$ Yang--Mills equations in the case $\R^{1,1}$.}
\begin{center}
\begin{tabular}{|c|c|c|c|c|c|c|c|c|}
\hline
  $d_J$ & $x_J$ & $y_J$ & $d_A$ & $x_A$ & $y_A$ & $A$ & $F$ & $F^2$  \\ \hline 
  0 & 1 & 1 & 0 & 1 & 1 & 1:\, see (\ref{An2}) & 1:\, see (\ref{Fn2}) & 1:\, see (\ref{F2n2})  \\ \hline
  0 & 1 & 0 &  &  &  & $\varnothing$ & $\varnothing$ & $\varnothing$ \\ \hline
  0 & 0 & 1 &  &  &  & $\varnothing$ & $\varnothing$ & $\varnothing$  \\ \hline
  0 & 0 & 0 & 0 & 0 & 0 & $A=0$ & $F=0$ & $F^2=0$ \\
   &  &  & 0 & 1 & 0 & $\infty$:\, see (\ref{An201}) & $F=0$ & $F^2=0$  \\
   &  &  & 0 & 0 & 1 & $\infty$:\, see (\ref{An202}) & $F=0$ & $F^2=0$ \\
   &  &  & 1 & 0 & 0 & 1:\, see (\ref{An203}) & $F=0$ & $F^2=0$ \\ \hline
 1 & 0 & 0 &  &  &  & $\varnothing$ & $\varnothing$ & $\varnothing$ \\  \hline
\end{tabular}
\end{center}
\end{table}

\section{The results for the case \texorpdfstring{$\R^{p,q}$}{Rpq}, \texorpdfstring{$p+q=n\geq 3$}{p+q=n<3}, \texorpdfstring{$p\geq 1$}{p<1}, \texorpdfstring{$q\geq 1$}{q<1}.}
\label{sect4}

Let us consider the case $n=p+q\geq 3$, $p\geq 1$, $q\geq 1$ for the known $J_{n\times 3}=(J^\nu_{\,\,k})$ and the unknown $A_{n\times 3}=(A^\nu_{\,\,k})$. Denote the elements of the diagonal blocks $X_{x_J}$ and $Y_{y_J}$ of the matrix $\Sigma^J$ by $j_{1}, j_{2}, j_{3}$ and the elements of the diagonal blocks $K$, $L$, $M$, $N$ of the matrix $\Psi^A$ (see Theorem \ref{th1}) by $a_{1}, a_{2}, a_{3}$.

\subsection{The case \texorpdfstring{$d_J=0$}{dJ=0}.}

In the case $d_J=0$, the problem  actually is reduced to solving the following two systems of three cubic equations:
\begin{eqnarray}
b_1((b_2)^2+(b_3)^2)&=&j_1,\nonumber\\
b_2((b_1)^2+(b_3)^2)&=&j_2,\label{sp4}\\
b_3((b_1)^2+(b_2)^2)&=&j_3\nonumber
\end{eqnarray}
and
\begin{eqnarray}
b_1((b_2)^2-(b_3)^2)&=&j_1,\nonumber\\
b_2((b_1)^2-(b_3)^2)&=&j_2,\label{sp44}\\
b_3((b_1)^2+(b_2)^2)&=&j_3.\nonumber
\end{eqnarray}
The systems (\ref{sp4}), (\ref{sp44}) have the following symmetry. Suppose that $(b_1, b_2, b_3)$ is a solution of (\ref{sp4}) or (\ref{sp44}) for the known $j_1, j_2, j_3$. If we change the sign of some $j_k$, $k=1, 2, 3$, then we must change the sign of the corresponding $b_k$, $k=1, 2, 3$. Without loss of generality, we can assume that all expressions $j_k$, $k=1, 2, 3$, in (\ref{sp4}) and (\ref{sp44}) are nonnegative. In (\ref{sp4}), the expressions $b_k$, $k=1, 2, 3$ will be nonnegative too. In (\ref{sp44}), the expression $b_3$ will be nonnegative, the expressions $b_1$ and $b_2$ will be arbitrary real numbers.

The general solution of the system (\ref{sp4}) and its symmetries are discussed in \cite{Shirokov1}, where we obtain the same system in the case of arbitrary Euclidean space $\R^n$. We remind these statements (Lemmas \ref{thSym} and \ref{thGenSol}) here without proof for the convenience of the reader. We present the general solution of the system (\ref{sp44}) and its symmetries in Lemmas \ref{thSym0} and \ref{thSh2}.

\begin{lemma}\label{thSym}\cite{Shirokov1}
If the system (\ref{sp4}) has the solution $(b_1, b_2, b_3)$ with $b_1\neq 0$, $b_2\neq 0$, $b_3\neq 0$, then this system has also the solution $(\frac{K}{b_1},\frac{K}{b_2},\frac{K}{b_3})$, where $K=(b_1 b_2 b_3)^{\frac{2}{3}}.$
\end{lemma}

Thus the expression
$$K=b_1 b_1'=b_2 b_2'=b_3 b_3'=(b_1 b_2 b_3)^{\frac{2}{3}}=(b_1' b_2' b_3')^{\frac{2}{3}}$$
is a conserved quantity for a pair of solutions $(b_1, b_2, b_3)$ and $(b_1', b_2', b_3')$ of the system (\ref{sp4}), where all numbers $b_1$, $b_2$, $b_3$, $b_1'$, $b_2'$, $b_3'$ are nonzero.

\begin{lemma}\label{thGenSol}\cite{Shirokov1} The system of equations (\ref{sp4}) with nonnegative parameters $j_1\geq 0$, $j_2\geq 0$, $j_3\geq 0$ has the following general solution:
\begin{enumerate}
  \item in the case $j_1=j_2=j_3=0$, has the solutions
$(b_1, 0, 0)$, $(0, b_2, 0)$, and $(0, 0, b_3)$ for all $b_1, b_2, b_3\in \R$;
  \item in the cases $j_1=j_2=0$, $j_ 3\neq 0$ (or similar cases with circular permutation), has no solutions;
  \item in the case $j_ 1\neq0$, $j_ 2\neq 0$, $j_ 3=0$ (or similar cases with circular permutation),
   has the following unique solution
  $$b_1=\sqrt[3]{\frac{j_2^2}{j_1}},\qquad b_ 2=\sqrt[3]{\frac{j_1^2}{j_ 2}},\qquad b_ 3=0;$$

 \item in the case $j_1=j_2=j_3\neq 0$, has the following unique solution
 $$b_ 1=b_ 2=b_ 3=\sqrt[3]{\frac{j_1}{2}};$$

 \item in the case of not all the same $j_1, j_2, j_3>0$ (and we take positive for simplicity), has the following two solutions
     $$(b_ {1+}, b_ {2+}, b_ {3+}),\qquad (b_ {1-}, b_ {2-}, b_ {3-})$$
     with the following expression for $K$ from Lemma \ref{thSym}
$$K:=b_{1+}b_{1-}=b_{2+}b_{2-}=b_{3+}b_{3-}=(b_{1+}b_{2+}b_{3+})^{\frac{2}{3}}=(b_{1-}b_{2-}b_{3-})^{\frac{2}{3}}:$$
       \begin{enumerate}
        \item in the case $j_1=j_2>j_3>0$ (or similar cases with circular permutation):
\begin{eqnarray}
&&b_{1\pm}=b_{2\pm}=\sqrt[3]{\frac{j_3}{2z_{\pm}}},\quad b_{3\pm}=z_{\pm} b_{1\pm},\quad z_{\pm}=\frac{j_1\pm\sqrt{j_1^2-j_3^2}}{j_3}.\nonumber\\
&&\mbox{Moreover,}\quad z_+ z_-=1,\quad K=(\frac{j_3}{2})^\frac{2}{3}.\nonumber
\end{eqnarray}
        \item in the case $j_3>j_1=j_2>0$ (or similar cases with circular permutation):
\begin{eqnarray}
&&b_{1\pm}=\frac{1}{w_\pm}b_{3},\quad b_{2\pm}=w_{\pm}b_{3},\quad b_{3\pm}=b_{3}=\sqrt[3]{\frac{j_1}{s}},\nonumber\\
&&w_{\pm}=\frac{s\pm\sqrt{s^2-4}}{2},\quad s=\frac{j_3+\sqrt{j_3^2+8j_1^2}}{2j_1}.\nonumber\\
&&\mbox{Moreover,}\quad w_+ w_-=1,\quad b_{1\pm}=b_{2\mp},\quad K=(\frac{j_1}{s})^\frac{2}{3}.\nonumber
\end{eqnarray}
        \item in the case of all different $j_1, j_2, j_3>0$:
\begin{eqnarray}
&&b_{1\pm}=\sqrt[3]{\frac{j_3}{t_0 y_\pm z_\pm}},\quad b_{2\pm}=y_{\pm} b_{1\pm},\quad b_{3\pm}=z_{\pm} b_{1\pm},\nonumber\\
&&z_{\pm}=\sqrt{\frac{y_\pm(j_1-j_2 y_\pm)}{j_ 2-j_1 y_\pm}},\quad y_{\pm}=\frac{t_0\pm\sqrt{t_0^2-4}}{2},\nonumber
\end{eqnarray}
where $t_0>2$ is the solution (it always exists, moreover, it is bigger than $\frac{j_2}{j_1}+\frac{j_ 1}{j_2}$) of the cubic equation $$j_1 j_2 t^3-(j_1^2+j_2^2+j_3^2)t^2+4j_3^2=0.$$
\begin{eqnarray}
\mbox{Moreover,}\quad y_+ y_-=1,\quad z_+ z_-=1,\quad K=(\frac{j_3}{t_0})^\frac{2}{3}.\nonumber
\end{eqnarray}
We can use the explicit Viete or Cardano formulas for $t_0$:
$$t_0=\Omega+2 \Omega\cos(\frac{1}{3}\arccos(1-\frac{2\beta}{\Omega^3})),$$
$$\Omega:=\frac{\alpha+\beta}{3},\qquad \alpha:=A+\frac{1}{A}>2,\qquad \beta:=\frac{B^2}{A},\qquad A:=\frac{j_2}{j_1},\qquad B:=\frac{j_3}{j_1},$$
$$t_0=\Omega+L+ \frac{\Omega^2}{L}\qquad L:=\sqrt[3]{\Omega^3-2\beta+2\sqrt{\beta(\beta-\Omega^3)}}.$$
\end{enumerate}
\end{enumerate}
\end{lemma}

\begin{lemma}\label{thSym0}
If the system (\ref{sp44}) has the solution $(b_1, b_2, b_3)$ with $b_1\neq 0$, $b_2\neq 0$, $b_3\neq 0$, then this system has also the solution $(\frac{-K}{b_1},\frac{-K}{b_2},\frac{K}{b_3})$, where $K=(b_1 b_2 b_3)^{\frac{2}{3}}>0.$
\end{lemma}

\begin{proof}
Let us substitute $(\frac{-K}{b_1},\frac{-K}{b_2},\frac{K}{b_3})$ into the first equation. We have
$$j_1=\frac{-K}{b_1}(\frac{K^2}{b_2^2}-\frac{K^2}{b_3^2})=\frac{-K^3(b_3^2-b_2^2)}{b_1b_2^2b_3^2}.$$
Using $j_1=b_1(b_2^2-b_3^2)$, we get $K=(b_1 b_2 b_3)^{\frac{2}{3}}$.
We can do the same with the second equation. For the third equation, we have
$$j_3=\frac{K}{b_3}(\frac{K^2}{b_1^2}+\frac{K^2}{b_2^2})=\frac{K^3(b_1^2+b_2^2)}{b_3 b_1^2 b_2^2}.$$
Using $j_3=b_3(b_1^2+b_2^2)$, we obtain $K=(b_1 b_2 b_3)^{\frac{2}{3}}$ again. $\blacksquare$
\end{proof}

Thus the expression
$$K=-b_1 b_1'=-b_2 b_2'=b_3 b_3'=(b_1 b_2 b_3)^{\frac{2}{3}}=(b_1' b_2' b_3')^{\frac{2}{3}}$$
is a conserved quantity for the pair of solutions $(b_1, b_2, b_3)$ and $(b_1', b_2', b_3')$ of the system (\ref{sp44}), where all numbers $b_1$, $b_2$, $b_3$, $b_1'$, $b_2'$, $b_3'$ are nonzero.

\begin{lemma}\label{thSh2} The system of equations (\ref{sp44}) with nonnegative parameters $j_1\geq 0$, $j_2\geq 0$, $j_3\geq 0$ has the following general solution:
\begin{enumerate}
  \item in the case $j_1=j_2=j_3=0$, has the solutions
$(b_1, 0, 0)$, $(0, b_2, 0)$, and $(0, 0, b_3)$ for all $b_1, b_2, b_3\in \R$;
  \item in the case $j_1=j_2=0$, $j_3\neq 0$, has the following 4 solutions:
$$
\mbox{the pair of solutions}\,\,(\sqrt[3]{\frac{j_3}{2}}, \sqrt[3]{\frac{j_3}{2}},  \sqrt[3]{\frac{j_3}{2}}),\,\, (-\sqrt[3]{\frac{j_3}{2}}, -\sqrt[3]{\frac{j_3}{2}}, \sqrt[3]{\frac{j_3}{2}})\,\mbox{with $K=(\frac{j_3}{2})^{\frac{2}{3}}$;}
$$
$$
\mbox{the pair of solutions}\,\,(-\sqrt[3]{\frac{j_3}{2}}, \sqrt[3]{\frac{j_3}{2}}, \sqrt[3]{\frac{j_3}{2}}),\,\, (\sqrt[3]{\frac{j_3}{2}}, -\sqrt[3]{\frac{j_3}{2}}, \sqrt[3]{\frac{j_3}{2}})\,\mbox{with $K=(\frac{j_3}{2})^{\frac{2}{3}}$;}
$$

\item in the cases $j_1=j_3=0$, $j_2\neq 0$ and $j_2=j_3=0$, $j_1\neq 0$, has no solutions;

\item in the case $j_1\neq 0$, $j_2\neq 0$, $j_3 =0$, has the following unique solution
  $$b_1=\sqrt[3]{\frac{j_2^2}{j_1}}, \qquad b_2=\sqrt[3]{\frac{j_1^2}{j_2}}, \qquad b_3=0.$$

\item in the case $j_1=0$, $j_2\neq 0$, $j_3 \neq0$ (and analogously in the case $j_1\neq0$, $j_2=0$, $j_3\neq 0$),
   has the following 1 or 5 solutions: one solution
  $$b_1=0,\qquad b_2=-\sqrt[3]{\frac{j_3^2}{j_2}},\qquad b_3=\sqrt[3]{\frac{j_2^2}{j_3}},$$
  and 4 additional solutions in the subcase $j_3>j_2$:
$$b_1=\pm\sqrt{\frac{j_3+j_2}{2b_3}}=\pm\frac{\sqrt{j_3+j_2}}{\sqrt[6]{4(j_3-j_2)}},\qquad b_3=b_2=\sqrt[3]{\frac{j_3-j_2}{2}};$$
$$b_1=\pm\sqrt{\frac{j_3-j_2}{2b_3}}=\pm\frac{\sqrt{j_3-j_2}}{\sqrt[6]{4(j_3+j_2)}},\qquad b_3=-b_2=\sqrt[3]{\frac{j_3+j_2}{2}}.$$
  Each of two pair of solutions (in each pair, we have one $b_1>0$ and one $b_1'=\frac{-K}{b_1}<0$) has the same $$K=\sqrt[3]{\frac{j_3^2-j_2^2}{4}}.$$

 \item in the case of all nonzero $ j_1>0$, $ j_2>0$, $ j_3>0$, has the following 2, 4, or 6 solutions:
     \begin{enumerate}
        \item in the case $j_1= j_2<\frac{j_3}{2\sqrt{2}}$, has 6 solutions ($b_{1\pm}, b_{2\pm}, b_{3\pm}$), ($c_{1\pm}^+, c_{2\pm}^+, c_{3\pm}^+$), ($c_{1\pm}^-, c_{2\pm}^-, c_{3\pm}^-$):
\begin{eqnarray}
&&\!\!\!\!\!\!\!\!\!\!\!\!\!\!\!\!\!\!\!\!b_{1\pm}=b_{2\pm}=\sqrt[3]{\frac{ j_3}{2z_{\pm}}},\qquad b_{3\pm}=z_{\pm} b_{1\pm},\qquad z_{\pm}=\frac{- j_1\pm\sqrt{ j_1^2+ j_3^2}}{ j_3},\label{solB}\\
&&\mbox{Moreover, $z_+ z_-=-1$ and $K=(\frac{ j_3}{2})^\frac{2}{3}$.}\nonumber
\end{eqnarray}

\begin{eqnarray}
&&\!\!\!\!\!\!\!\!\!\!\!\!\!\!\!\!\!\!\!\!c_{1\pm}^\pm=\frac{1}{w_{\pm}^\pm}c_{3}^\pm,\qquad c_{2\pm}^\pm=-w_{\pm}^\pm c_{3}^\pm,\qquad c_{3}^\pm=c_{3\pm}^\pm=\sqrt[3]{\frac{ j_1}{s^\pm}},\label{solC}\\
&&w_{\pm}^\pm=\frac{s^{\pm}\pm\sqrt{(s^{\pm})^2+4}}{2},\qquad s^{\pm}=\frac{ j_3\pm\sqrt{ j_3^2-8 j_1^2}}{2 j_1}.\nonumber\\
&&\mbox{Moreover, $w_+^\pm w_-^\pm=-1$, $s^+ s^-=2$, $c_{1\pm}^\pm=c_{2\mp}^\pm$, and $K=(\frac{ j_1}{s^\pm})^\frac{2}{3}$.}\nonumber
\end{eqnarray}

\item in the case $ j_1= j_2=\frac{j_3}{2\sqrt{2}}$, has 4 solutions ($b_{1\pm}, b_{2\pm}, b_{3\pm}$) (\ref{solB}) and ($c_{1\pm}, c_{2\pm}, c_{3\pm}$) (as a particular class of solutions (\ref{solC})):
\begin{eqnarray}
&&\!\!\!\!\!\!\!\!\!\!\!\!\!\!\!\!\!\!\!\!b_{1\pm}=b_{2\pm}=\sqrt[3]{\frac{ j_3}{2z_{\pm}}},\qquad b_{3\pm}=z_{\pm} b_{1\pm},\qquad z_{\pm}=\frac{1}{\sqrt{2}},\, -\sqrt{2},\label{solB2}\\
&&\mbox{Moreover, $z_+ z_-=-1$ and $K=(\frac{ j_3}{2})^\frac{2}{3}$;}\nonumber\\
&&\!\!\!\!\!\!\!\!\!\!\!\!\!\!\!\!\!\!\!\!c_{1\pm}=\frac{1}{w_{\pm}}c_{3},\quad c_{2\pm}=-w_{\pm} c_{3},\quad c_{3\pm}=c_{3}=\sqrt[3]{\frac{ j_1}{\sqrt{2}}},\quad w_{\pm}=\frac{\sqrt{2}\pm\sqrt{6}}{2}.\nonumber\\
&&\mbox{Moreover, $w_+w_-=-1$, $c_{1\pm}=c_{2\mp}$, and $K=(\frac{j_1}{\sqrt{2}})^\frac{2}{3}$.}\nonumber \\
&&\mbox{Note that  $\sqrt[3]{j_3}=b_{3-}\neq b_{3+}=c_{3\pm}=\sqrt[3]{\frac{j_3}{4}}$.}\nonumber
\end{eqnarray}

\item in the case $ j_1= j_2>\frac{j_3}{2\sqrt{2}}$, has 2 solutions ($b_{1\pm}, b_{2\pm}, b_{3\pm}$) (\ref{solB}).

\item in the case $ j_1\neq j_2$, $j_3^{\frac{2}{3}}>j_1^{\frac{2}{3}}+j_2^{\frac{2}{3}}$, has 6 solutions $(d_{1\pm k}, d_{2\pm k}, d_{3\pm k})$, $k=1, 2, 3$:
\begin{eqnarray}
&&\!\!\!\!\!\!\!\!\!\!\!\!\!\!\!\!\!\!\!\!d_{1\pm k}=\sqrt[3]{\frac{j_3}{t_k y_{\pm k} z_{\pm k}}},\qquad d_{2\pm k}=y_{\pm k} d_{1\pm k},\qquad d_{3\pm k}=z_{\pm k} d_{1\pm k},\label{solD}\\
&&z_{\pm k}=\lambda_{\pm k}\sqrt{\frac{y_{\pm k}( j_2 y_{\pm k}- j_1)}{ j_2-y_{\pm k} j_1}},\qquad y_{\pm k}=\frac{t_k\pm\sqrt{t_k^2-4}}{2},\nonumber
\end{eqnarray}
where $\lambda_{\pm k}=\sign(\frac{ j_3 y_{\pm k} (1-y_{\pm k})}{ j_2- j_1y_{\pm k}})$, and $t_k$, $k=1, 2, 3$, are solutions of the cubic equation
$$ j_1  j_2 t^3+( j_3^2- j_2^2- j_1^2)t^2-4 j_3^2=0.$$

Note that in the case $j_1\neq j_2$, this cubic equation always has one solution $2<t_1<\frac{ j_2}{ j_1}+\frac{ j_1}{ j_2}$ and also, in the case $j_3^{\frac{2}{3}}>j_1^{\frac{2}{3}}+j_2^{\frac{2}{3}}$, two solutions $t_2, t_3<-2$.

\begin{eqnarray}
\mbox{Moreover, $y_{+k} y_{-k}=1$,\quad $\lambda_{+k} \lambda_{-k}=-1$,\quad $z_{+k} z_{-k}=-1$,}\nonumber\\
\mbox{and $K_k=(\frac{ j_3}{t_k})^\frac{2}{3}$ for each $k=1, 2, 3$.}\nonumber
\end{eqnarray}

We can use the explicit Viete or Cardano formulas for $t_k$:
$$t_k=\Omega+2 \Omega\cos(\frac{1}{3}\arccos(1+\frac{2\beta}{\Omega^3})+\varphi_k),\qquad \varphi_k=\frac{2\pi}{3}, -\frac{2\pi}{3}, 0,$$
$$\Omega:=\frac{\alpha-\beta}{3},\qquad \alpha:=A+\frac{1}{A}>2,\qquad \beta:=\frac{B^2}{A},\qquad A:=\frac{j_2}{j_1},\qquad B:=\frac{j_3}{j_1},$$
$$t_k=\Omega+L+ \frac{\Omega^2}{L}\qquad L:=\sqrt[3]{\Omega^3+2\beta+2\sqrt{\beta(\beta+\Omega^3)}}.$$
Note that $j_3^{\frac{2}{3}}>j_1^{\frac{2}{3}}+j_2^{\frac{2}{3}}$ is equivalent to $2\beta^\frac{1}{3}<\beta-\alpha$. So, we have $\Omega<0$.
\item in the case $ j_1\neq j_2$, $j_3^{\frac{2}{3}}=j_1^{\frac{2}{3}}+j_2^{\frac{2}{3}}$, has 4 solutions: $(d_{1\pm k}, d_{2\pm k}, d_{3\pm k})$, $k=1, 2$, (\ref{solD}) with
    \begin{eqnarray}
    t_1=\sqrt[3]{\frac{j_2}{j_1}}+\sqrt[3]{\frac{j_1}{j_2}}>2,\qquad t_2=-2(\sqrt[3]{\frac{j_2}{j_1}}+\sqrt[3]{\frac{j_1}{j_2}})<-4.\nonumber
    \end{eqnarray}

\item in the case $ j_1\neq j_2$, $j_3^{\frac{2}{3}}<j_1^{\frac{2}{3}}+j_2^{\frac{2}{3}}$, has 2 solutions: $(d_{1\pm 1}, d_{2\pm 1}, d_{3\pm 1})$ (\ref{solD}) with $t_1>2$.
\end{enumerate}
\end{enumerate}
\end{lemma}

\begin{proof}  The detailed proof is given in Appendix A. $\blacksquare$ \end{proof}

\subsubsection{The case \texorpdfstring{$d_J=0$}{dJ=0}, \texorpdfstring{$\rank (J)=3$}{rank J=3}.}

We have the following four subcases 1) $x_J=3$, $y_J=0$; 2) $x_J=0$, $y_J=3$; 3) $x_J=1$, $y_J=2$; 4) $x_J=2$, $y_J=1$. It is important which of the three diagonal elements of the matrix $J$ are from the first $p$ rows and which of them are not, and this depends on $x_J$ and $y_J$. In different cases, we obtain different systems of equations, because of the signs before the summands $(a_k)^2$, $k=1, 2, 3$.

We remind that we consider only positive elements $j_1$, $j_2$, $j_3$ (they are called hyperbolic singular values of the matrix $J$) for simplicity.

In the case $d_J=0$, $x_J=3$, $y_J=0$, we have
$$\Sigma^J=\left(
                 \begin{array}{ccc}
                  j_1 & 0 & 0 \\
                  0 & j_2 & 0 \\
                  0 & 0 & j_3 \\
                  \Zero & \Zero & \Zero \\ \hline
                  \Zero & \Zero & \Zero \\
                 \end{array}
               \right),\qquad j_1, j_2, j_3\neq 0.
$$
In this case, we have $d_A=0$, $x_A=3$, $y_A=0$, and we are looking for a solution in the form
$$
\Psi^A=\left(
                 \begin{array}{ccc}
                  a_1 & 0 & 0 \\
                  0 & a_2 & 0 \\
                  0 & 0 & a_3 \\
                  \Zero & \Zero & \Zero \\ \hline
                  \Zero & \Zero & \Zero \\
                 \end{array}
               \right),\qquad a_1, a_2, a_3\neq 0.
$$
We obtain the following system of three cubic equations
\begin{eqnarray}
-a_1((a_2)^2+(a_3)^2)&=&j_1,\nonumber\\
-a_2((a_1)^2+(a_3)^2)&=&j_2,\nonumber\\
-a_3((a_1)^2+(a_2)^2)&=&j_3.\nonumber
\end{eqnarray}
This system is reduced to the system (\ref{sp4}) using the change of variables $b_k=-a_k$, $k=1, 2, 3$.
We have one or two solutions. This depends on the numbers $j_1$, $j_2$ and $j_3$.
If $j:=j_1=j_2=j_3\neq 0$, then we have the unique solution
\begin{eqnarray}
\Psi^A=\left(
                 \begin{array}{ccc}
                  a & 0 & 0 \\
                  0 & a & 0 \\
                  0 & 0 & a \\
                  \Zero & \Zero & \Zero \\ \hline
                  \Zero & \Zero & \Zero \\
                 \end{array}
               \right),\qquad a=-\sqrt[3]{\frac{j}{2}},\label{Ad030}
\end{eqnarray}
with
\begin{eqnarray}
F^{12}=-\sqrt[3]{\frac{j^2}{4}}\tau^3,\quad F^{23}=-\sqrt[3]{\frac{j^2}{4}}\tau^1,\quad F^{31}=-\sqrt[3]{\frac{j^2}{4}}\tau^2,\label{Fd030}\\
F^2=\frac{-3}{2}\sqrt[3]{\frac{j^4}{16}}I_2\neq 0\label{Fqd030}.
\end{eqnarray}
If $j_1$, $j_2$, and $j_3$ are not all the same, then we have the following two different solutions
\begin{eqnarray}
\Psi^A=\left(
                 \begin{array}{ccc}
                  -b_{1\pm} & 0 & 0 \\
                  0 & -b_{2\pm} & 0 \\
                  0 & 0 & -b_{3\pm} \\
                  \Zero & \Zero & \Zero \\ \hline
                  \Zero & \Zero & \Zero \\
                 \end{array}
               \right),\label{Ad030n}
\end{eqnarray}
where $b_{k\pm}$, $k=1, 2, 3$ are from Case (v) of Lemma \ref{thGenSol}, with
\begin{eqnarray}
F^{12}_{\pm}=-b_{1\pm}b_{2\pm}\tau^3,\quad F^{23}_{\pm}=-b_{2\pm}b_{3\pm}\tau^1,\quad F^{31}_{\pm}=-b_{3\pm}b_{1\pm}\tau^2,\label{Fd030n}\\
F^2_{\pm}=\frac{-1}{2}((b_{1\pm}b_{2\pm})^2+(b_{2\pm}b_{3\pm})^2+(b_{3\pm}b_{1\pm})^2) I_2\neq 0.\label{Fqd030n}
\end{eqnarray}
In the next lemma, we give the explicit form of (\ref{Fqd030n}). You can find the proof in \cite{Shirokov1}.

\begin{lemma}\label{lemF2}\cite{Shirokov1} In the case of not all the same $j_1$, $j_2$, $j_3$, (\ref{Fqd030n}) takes the form:
\begin{enumerate}
  \item in the case $j_1=j_2>j_3>0$ (or similar cases with circular permutation):
\begin{eqnarray}
F^2_{\pm}=\frac{-K^2 (1+2z_{\pm}^2)}{2 z_{\pm}^{\frac{4}{3}}}I_2,\qquad F^2_+\neq F^2_-,\label{F21}
\end{eqnarray}
where
$$z_{\pm}=\frac{j_1\pm\sqrt{j_1^2-j_3^2}}{j_3},\qquad K=(\frac{j_3}{2})^{\frac{2}{3}}.$$
  \item in the case $j_3>j_1=j_2>0$ (or similar cases with circular permutation):
\begin{eqnarray}
F^2_{\pm}=\frac{-K^2(s^2-1)}{2}I_2,\qquad  F^2_+=F^2_-,\label{F22}
\end{eqnarray}
where
$$s=\frac{j_3+\sqrt{j_3^2+8j_1^2}}{2j_1}>2,\qquad K=(\frac{j_1}{s})^{\frac{2}{3}}.$$
 \item in the case of all different $j_1$, $j_2$, $j_3 >0$:
\begin{eqnarray}
F^2_{\pm}=\frac{-K^2 (y_{\pm}^2+z_{\pm}^2+y_{\pm}^2 z_{\pm}^2)}{2(y_{\pm} z_{\pm})^{\frac{4}{3}}}I_2,\qquad  F^2_+\neq F^2_-,\label{F23}
\end{eqnarray}
where $$K=(\frac{j_3}{t_0})^{\frac{2}{3}},$$
and $y_{\pm}$, $z_\pm$, $t_0$ are from Case (v) - (c) of Lemma \ref{thGenSol}.
\end{enumerate}
In all cases of Lemma, the expression $K$ is the invariant for each pair of solutions (see Lemmas \ref{thSym} and \ref{thGenSol}).
\end{lemma}

In the case $d_J=0$, $x_J=0$, $y_J=3$, we have
$$\Sigma^J=\left(
                 \begin{array}{ccc}
                  \Zero & \Zero & \Zero \\ \hline
                 j_1 & 0 & 0 \\
                  0 & j_2 & 0 \\
                  0 & 0 & j_3 \\
                \Zero & \Zero & \Zero \\
                 \end{array}
               \right),\qquad j_1, j_2, j_3\neq 0.
$$
In this case, we have $d_A=0$, $x_A=0$, $y_A=3$, and we are looking for a solution in the form
$$
\Psi^A=\left(
                 \begin{array}{ccc}
                  \Zero & \Zero & \Zero \\ \hline
                  a_1 & 0 & 0 \\
                  0 & a_2 & 0 \\
                  0 & 0 & a_3 \\
                  \Zero & \Zero & \Zero \\
                 \end{array}
               \right),\qquad a_1, a_2, a_3\neq 0.
$$
We obtain the following system of three cubic equations
\begin{eqnarray}
a_1((a_2)^2+(a_3)^2)&=&j_1,\nonumber\\
a_2((a_1)^2+(a_3)^2)&=&j_2,\nonumber\\
a_3((a_1)^2+(a_2)^2)&=&j_3.\nonumber
\end{eqnarray}
Using the change of variables $b_k=a_k$, $k=1, 2, 3$, we obtain the system (\ref{sp4}). We have one or two solutions. This depends on the numbers $j_1$, $j_2$ and $j_3$.
If $j:=j_1=j_2=j_3\neq 0$, then we have the unique solution
\begin{eqnarray}
\Psi^A=\left(
                 \begin{array}{ccc}
                  \Zero & \Zero & \Zero \\ \hline
                  a & 0 & 0 \\
                  0 & a & 0 \\
                  0 & 0 & a \\
                  \Zero & \Zero & \Zero \\
                 \end{array}
               \right),\qquad a=\sqrt[3]{\frac{j}{2}},\label{Ad003}
\end{eqnarray}
with
\begin{eqnarray}
F^{p+1\,p+2}=-\sqrt[3]{\frac{j^2}{4}}\tau^3,\, F^{p+2\, p+3}=-\sqrt[3]{\frac{j^2}{4}}\tau^1,\, F^{p+3\, p+1}=-\sqrt[3]{\frac{j^2}{4}}\tau^2,\label{Fd003}\\
F^2=\frac{-3}{2}\sqrt[3]{\frac{j^4}{16}}I_2\neq 0\label{Fqd003}.
\end{eqnarray}
If $j_1$, $j_2$, and $j_3$ are not all the same, then we have the following two different solutions
\begin{eqnarray}
\Psi^A=\left(
                 \begin{array}{ccc}
                  \Zero & \Zero & \Zero \\ \hline
                  b_{1\pm} & 0 & 0 \\
                  0 & b_{2\pm} & 0 \\
                  0 & 0 & b_{3\pm} \\
                  \Zero & \Zero & \Zero \\
                 \end{array}
               \right),\label{Ad003n}
\end{eqnarray}
where $b_{k\pm}$, $k=1, 2, 3$ are from Case (v) of Lemma \ref{thGenSol}, with
\begin{eqnarray}
F^{p+1\, p+2}_{\pm}=-b_{1\pm}b_{2\pm}\tau^3,\, F^{p+2\, p+3}_{\pm}=-b_{2\pm}b_{3\pm}\tau^1,\, F^{p+3\,p+1}_{\pm}=-b_{3\pm}b_{1\pm}\tau^2,\label{Fd003n}\\
F^2=\frac{-1}{2}((b_{1\pm}b_{2\pm})^2+(b_{2\pm}b_{3\pm})^2+(b_{3\pm}b_{1\pm})^2) I_2\neq 0.\label{Fqd003n}
\end{eqnarray}
The explicit expressions for $F^2$ (\ref{Fqd003n}) are (\ref{F21}), (\ref{F22}), (\ref{F23}) respectively for 3 subcases of Case (v) in Lemma \ref{thGenSol}.

In the case $d_J=0$, $x_J=2$, $y_J=1$, we have
$$\Sigma^J=\left(
                 \begin{array}{ccc}
                  j_1 & 0 & 0 \\
                  0 & j_2 & 0 \\
                  \Zero & \Zero & \Zero \\ \hline
                      0 & 0 & j_3 \\
                  \Zero & \Zero & \Zero \\
                 \end{array}
               \right),\qquad j_1, j_2, j_3\neq 0.
$$
In this case, we have $d_A=0$, $x_A=2$, $y_A=1$, and we are looking for a solution in the form
$$
\Psi^A=\left(                 \begin{array}{ccc}
                  a_1 & 0 & 0 \\
                  0 & a_2 & 0 \\
                  \Zero & \Zero & \Zero \\ \hline
                  0 & 0 & a_3 \\
                  \Zero & \Zero & \Zero \\
                 \end{array}
               \right),\qquad a_1, a_2, a_3\neq 0.
$$
We obtain the following system of three cubic equations
\begin{eqnarray}
-a_1((a_2)^2-(a_3)^2)&=&j_1,\nonumber\\
-a_2((a_1)^2-(a_3)^2)&=&j_2,\nonumber\\
-a_3((a_1)^2+(a_2)^2)&=&j_3.\nonumber
\end{eqnarray}
This system is reduced to the system (\ref{sp44}) using the change of variables $b_k=-a_k$, $k=1, 2, 3$.
We have 2, 4, or 6 solutions:
\begin{eqnarray}
\Psi^A=\left(                 \begin{array}{ccc}
                  -b_1 & 0 & 0 \\
                  0 & -b_2 & 0 \\
                  \Zero & \Zero & \Zero \\ \hline
                  0 & 0 & -b_3 \\
                  \Zero & \Zero & \Zero \\
                 \end{array}
               \right),\label{Ad0212}
\end{eqnarray}
where $b_{k}$, $k=1, 2, 3$ are from Case (vi) of Lemma \ref{thSh2}.

In all these cases, we have
\begin{eqnarray}
F^{12}=-b_1 b_2 \tau^3,\qquad F^{2\, p+1}=-b_2b_3\tau^1,\qquad F^{ p+1\, 1}=-b_3b_1 \tau^2,\label{Fd021}\\
F^2=\frac{-1}{2}((b_1b_2)^2-(b_2b_3)^2-(b_3b_1)^2)I_2.\label{Fqd021}
\end{eqnarray}
We give the explicit expressions for $F^2$ in different cases in Lemma \ref{lemmaF222}.

\begin{lemma}\label{lemmaF222} In the case of all nonzero $j_1>0$, $j_2>0$, $j_3>0$, the expression (\ref{Fqd021}) takes the form:
\begin{enumerate}
        \item in the case $j_1= j_2<\frac{j_3}{2\sqrt{2}}$
\begin{eqnarray}
\!\!\!\!\!\!\!\!\!\!\!\!\!\!\!\!\!\!\!\!\!\!\!\!\!\!\!\!\!\!\!\!\!\!\mbox{for $b_\pm$:}\quad F^2_{\pm}=\frac{K^2(2z_{\pm}^2-1)}{2 z_\pm^{\frac{4}{3}}}I_2\neq 0,\qquad F^2_+\neq F^2_-,\label{F2ww1}\\
\mbox{where $z_{\pm}=\frac{- j_1\pm\sqrt{ j_1^2+ j_3^2}}{ j_3}$,\quad $K=(\frac{ j_3}{2})^\frac{2}{3}$;}\nonumber\\
\!\!\!\!\!\!\!\!\!\!\!\!\!\!\!\!\!\!\!\!\!\!\!\!\!\!\!\!\!\!\!\!\!\!\mbox{for $c^\pm_\pm$:}\quad (F_{\pm}^{\pm})^2=\frac{K^2(1+(s^\pm)^2)}{2}I_2>0,\quad (F_+^+)^2=(F_-^+)^2\neq(F_+^-)^2=(F_-^-)^2,\nonumber\\
\mbox{where $s^{\pm}=\frac{ j_3\pm\sqrt{ j_3^2-8 j_1^2}}{2 j_1}$,\quad $K=(\frac{ j_1}{s^\pm})^\frac{2}{3}$;}\nonumber
\end{eqnarray}
The expression $F^2$ for $c^\pm_\pm$ coincides with $F^2$ for $b_\pm$ only in the case:
$$s^*:=\sqrt{13+\sqrt{193-6^\frac{4}{3}}+\sqrt{386+6^\frac{4}{3}+\frac{5362}{193-6^\frac{4}{3}}}}\approx 7.39438,$$
$$B^*:=(\frac{j_3}{j_1})^*=\frac{(s^*)^2+2}{s^*}\approx 7.66486>2\sqrt{2},$$
$$z_+^*:=\frac{-1+\sqrt{1+(B^*)^2}}{B^*}\approx 0.878009.$$
This means that if $B=B^*$, then $(F^+_{\pm})^2$ for $(s^+)^*=\frac{B^*+\sqrt{(B^*)^2-8}}{2}$ coincides with $F_+^2$ for $z_+^*=\frac{-1+\sqrt{1+(B^*)^2}}{B^*}$.
\item in the case $ j_1= j_2=\frac{j_3}{2\sqrt{2}}$,
\begin{eqnarray}
\mbox{for $b_\pm$:}\quad F^2_{+}=0,\quad F^2_-=\frac{3j_3^{\frac{4}{3}}}{8}I_2>0,\label{F2ww3}\\
\mbox{for $c_\pm$:}\quad (F_{\pm})^2=\frac{3}{2}(\frac{j_1}{\sqrt{2}})^\frac{4}{3}I_2=\frac{3 j_3^{\frac{4}{3}}}{2^{\frac{11}{3}}}I_2> 0,\nonumber
\end{eqnarray}
\item in the case $ j_1= j_2>\frac{j_3}{2\sqrt{2}}$,
\begin{eqnarray}
\!\!\!\!\!\!\!\!\!\!\!\!\!\!\!\!\!\!\!\!\!\!\!\!\!\!\!\!\!\!\!\!\!\!\mbox{for $b_\pm$:}\quad F^2_{\pm}=\frac{K^2(2z_{\pm}^2-1)}{2 z_\pm^{\frac{4}{3}}}I_2\neq 0,\qquad F^2_+\neq F^2_-,\label{F2ww5}\\
\mbox{where $z_{\pm}=\frac{- j_1\pm\sqrt{ j_1^2+ j_3^2}}{ j_3}$,\quad $K=(\frac{ j_3}{2})^\frac{2}{3}$;}\nonumber
\end{eqnarray}
\item in the case $ j_1\neq j_2$, $j_3^{\frac{2}{3}}>j_1^{\frac{2}{3}}+j_2^{\frac{2}{3}}$,
\begin{eqnarray}
\!\!\!\!\!\!\!\!\!\!\!\!\!\!\!\!\!\!\!\!\!\!\!\!\!\!\!\!\!\!\!\!\!\!\mbox{for $d_{\pm k}$:}\quad F^2_{\pm k}=\frac{K_k^2(y_{\pm k}^2z^2_{\pm k}-y_{\pm k}^2+z_{\pm k}^2)}{2(y_{\pm k}z_{\pm k })^{\frac{4}{3}}}I_2\neq 0,\quad k=1, 2, 3,\quad F_{+k}^2 \neq F_{-k}^2,\label{F2ww6}\\
\!\!\!\!\!\!\!\!\!\!\!\!\!\!\!\!\!\!\!\!\!\!\!\!\!\!\!\!\!\!\!\!\!\!\mbox{where $K_k=(\frac{j_3}{t_k})^{\frac{2}{3}}$, and $y_{\pm k}$, $z_{\pm k}$, $t_k$ are from Case (vi) - (d) of Lemma \ref{thSh2}.}\nonumber
\end{eqnarray}
\item in the case $ j_1\neq j_2$, $j_3^{\frac{2}{3}}=j_1^{\frac{2}{3}}+j_2^{\frac{2}{3}}$,
\begin{eqnarray}
F^2_{+1}=0,\quad F^2_{-1}=\frac{j_1^\frac{4}{3}+j_1^\frac{2}{3}j_2^\frac{2}{3}+j_2^\frac{4}{3}}{2}I_2> 0,\label{F2ww7}\\
F^2_{\pm 2}=\frac{K_2^2(y_{\pm 2}^2z^2_{\pm 2}-y_{\pm 2}^2+z_{\pm 2}^2)}{2(y_{\pm 2}z_{\pm 2 })^{\frac{4}{3}}}I_2\neq 0,\qquad t_2=-2(\sqrt[3]{\frac{j_2}{j_1}}+\sqrt[3]{\frac{j_1}{j_2}}),\nonumber\\
F^2_{+2}\neq F^2_{-2},\qquad \mbox{$F^2_{\pm 2}$ do not coincide with $F^2_{-1}$}.\nonumber
\end{eqnarray}
\item in the case $ j_1\neq j_2$, $j_3^{\frac{2}{3}}<j_1^{\frac{2}{3}}+j_2^{\frac{2}{3}}$,
\begin{eqnarray}
F^2_{\pm 1}=\frac{K_1^2(y_{\pm 1}^2z^2_{\pm 1}-y_{\pm 1}^2+z_{\pm 1}^2)}{2(y_{\pm 1}z_{\pm 1 })^{\frac{4}{3}}}I_2\neq 0,\qquad F^2_{+1}\neq F^2_{-1}.\label{F2ww8}
\end{eqnarray}
\end{enumerate}
In all cases of Lemma, the expression $K$ is the invariant for each pair of solutions (see Lemmas \ref{thSym0} and \ref{thSh2}).
\end{lemma}
\begin{proof}
We give the proof of this Lemma in Appendix B. $\blacksquare$
\end{proof}

In the case $d_J=0$, $x_J=1$, $y_J=2$, we have
$$\Sigma^J=\left(
                 \begin{array}{ccc}
                  j_1 & 0 & 0 \\
                  \Zero & \Zero & \Zero \\ \hline
             0 & j_2 & 0 \\
                      0 & 0 & j_3 \\
                  \Zero & \Zero & \Zero \\
                 \end{array}
               \right),\qquad j_1, j_2, j_3\neq 0.
$$
In this case, we have $d_A=0$, $x_A=1$, $y_A=2$, and we are looking for a solution in the form
$$
\Psi^A=\left(                 \begin{array}{ccc}
                  a_1 & 0 & 0 \\
                  \Zero & \Zero & \Zero \\ \hline
                  0 & a_2 & 0 \\
                  0 & 0 & a_3 \\
                  \Zero & \Zero & \Zero \\
                 \end{array}
               \right),\qquad a_1, a_2, a_3\neq 0.
$$
We obtain the following system of three cubic equations
\begin{eqnarray}
-a_1(-(a_2)^2-(a_3)^2)&=&j_1,\nonumber\\
-a_2((a_1)^2-(a_3)^2)&=&j_2,\nonumber\\
-a_3((a_1)^2-(a_2)^2)&=&j_3.\nonumber
\end{eqnarray}
This system is reduced again to the system (\ref{sp44}) using the change of variables $b_1=a_2$, $b_2=a_3$, $b_3=a_1$ and the change $j_1 \to j_3$, $j_2\to j_1$, $j_3\to j_2$. We have 2, 4, or 6 solutions:
\begin{eqnarray}
\Psi^A=\left(                 \begin{array}{ccc}
                  b_3 & 0 & 0 \\
                  \Zero & \Zero & \Zero \\ \hline
                  0 & b_1 & 0 \\
                  0 & 0 & b_2 \\
                  \Zero & \Zero & \Zero \\
                 \end{array}
               \right),\label{Ad0122}
\end{eqnarray}
where $b_k$, $k=1, 2, 3$ are from Lemma \ref{thSh2} with the change $j_1 \to j_3$, $j_2\to j_1$, $j_3\to j_2$.

In all these cases, we have
\begin{eqnarray}
F^{1\, p+1}=-b_3 b_1 \tau^3,\qquad F^{p+1\, p+2}=-b_1 b_2\tau^1,\qquad F^{p+2\, 1}=-b_2 b_3 \tau^2,\label{Fd012}\\
F^2=\frac{-1}{2}(-(b_3 b_1)^2+(b_1 b_2)^2-(b_2 b_3)^2)I_2.\label{Fqd012}
\end{eqnarray}

The explicit form of $F^2$ (\ref{Fqd012}) is given in Lemma \ref{lemmaF222} with the change $j_1 \to j_3$, $j_2\to j_1$, $j_3\to j_2$. For the convenience of the reader, we present it here:
\begin{enumerate}
        \item in the case $j_3= j_1<\frac{j_2}{2\sqrt{2}}$,
\begin{eqnarray}
\!\!\!\!\!\!\!\!\!\!\!\!\!\!\!\!\!\!\!\!\!\!\!\!\!\!\!\!\!\!\!\!\!\!\mbox{for $b_\pm$:}\quad F^2_{\pm}=\frac{K^2(2z_{\pm}^2-1)}{2 z_\pm^{\frac{4}{3}}}I_2\neq 0,\qquad F^2_+\neq F^2_-,\label{F2www1}\\
\mbox{where $z_{\pm}=\frac{- j_3\pm\sqrt{ j_3^2+ j_2^2}}{ j_2}$,\quad $K=(\frac{ j_2}{2})^\frac{2}{3}$;}\nonumber\\
\!\!\!\!\!\!\!\!\!\!\!\!\!\!\!\!\!\!\!\!\!\!\!\!\!\!\!\!\!\!\!\!\!\!\mbox{for $c^\pm_\pm$:}\quad (F_{\pm}^{\pm})^2=\frac{K^2(1+(s^\pm)^2)}{2}I_2>0,\quad (F_+^+)^2=(F_-^+)^2\neq(F_+^-)^2=(F_-^-)^2,\nonumber\\
\mbox{where $s^{\pm}=\frac{ j_2\pm\sqrt{ j_2^2-8 j_3^2}}{2 j_3}$,\quad $K=(\frac{ j_3}{s^\pm})^\frac{2}{3}$.}\nonumber
\end{eqnarray}
The expression $F^2$ for $c^\pm_\pm$ coincides with $F^2$ for $b_\pm$ only in the case:
$$s^*:=\sqrt{13+\sqrt{193-6^\frac{4}{3}}+\sqrt{386+6^\frac{4}{3}+\frac{5362}{193-6^\frac{4}{3}}}}\approx 7.39438,$$
$$B^*:=(\frac{j_2}{j_3})^*=\frac{(s^*)^2+2}{s^*}\approx 7.66486>2\sqrt{2},$$
$$z_+^*:=\frac{-1+\sqrt{1+(B^*)^2}}{B^*}\approx 0.878009.$$
This means that if $B=B^*$, then $(F^+_{\pm})^2$ for $(s^+)^*=\frac{B^*+\sqrt{(B^*)^2-8}}{2}$ coincides with $F_+^2$ for $z_+^*=\frac{-1+\sqrt{1+(B^*)^2}}{B^*}$.
\item in the case $ j_3= j_1=\frac{j_2}{2\sqrt{2}}$,
\begin{eqnarray}
\mbox{for $b_\pm$:}\quad F^2_{+}=0,\quad F^2_-=\frac{3j_2^{\frac{4}{3}}}{8}I_2>0,\label{F2www2}\\
\mbox{for $c_\pm$:}\quad (F_{\pm})^2=\frac{3}{2}(\frac{j_3}{\sqrt{2}})^\frac{4}{3}I_2=\frac{3 j_2^{\frac{4}{3}}}{2^{\frac{11}{3}}}I_2> 0,\nonumber
\end{eqnarray}
\item in the case $ j_3= j_1>\frac{j_2}{2\sqrt{2}}$,
\begin{eqnarray}
\!\!\!\!\!\!\!\!\!\!\!\!\!\!\!\!\!\!\!\!\!\!\!\!\!\!\!\!\!\!\!\!\!\!\mbox{for $b_\pm$:}\quad F^2_{\pm}=\frac{K^2(2z_{\pm}^2-1)}{2 z_\pm^{\frac{4}{3}}}I_2\neq 0,\qquad F^2_+\neq F^2_-,\label{F2www3}\\
\mbox{where $z_{\pm}=\frac{- j_3\pm\sqrt{ j_3^2+ j_2^2}}{ j_2}$,\quad $K=(\frac{ j_2}{2})^\frac{2}{3}$;}\nonumber
\end{eqnarray}
\item in the case $ j_3\neq j_1$, $j_2^{\frac{2}{3}}>j_3^{\frac{2}{3}}+j_1^{\frac{2}{3}}$,
\begin{eqnarray}
\!\!\!\!\!\!\!\!\!\!\!\!\!\!\!\!\!\!\!\!\!\!\!\!\!\!\!\!\!\!\!\!\!\!\mbox{for $d_{\pm k}$:}\quad F^2_{\pm k}=\frac{K_k^2(y_{\pm k}^2z^2_{\pm k}-y_{\pm k}^2+z_{\pm k}^2)}{2(y_{\pm k}z_{\pm k })^{\frac{4}{3}}}I_2\neq 0,\quad k=1, 2, 3,\quad F_{+k}^2 \neq F_{-k}^2,\label{F2www4}\\
\!\!\!\!\!\!\!\!\!\!\!\!\!\!\!\!\!\!\!\!\!\!\!\!\!\!\!\!\!\!\!\!\!\!\mbox{where $K_k=(\frac{j_2}{t_k})^{\frac{2}{3}}$, and $y_{\pm k}$, $z_{\pm k}$, $t_k$ are from Case (vi) - (d) of Lemma \ref{thSh2}.}\nonumber
\end{eqnarray}
\item in the case $ j_3\neq j_1$, $j_2^{\frac{2}{3}}=j_3^{\frac{2}{3}}+j_1^{\frac{2}{3}}$,
\begin{eqnarray}
F^2_{+1}=0,\quad F^2_{-1}=\frac{j_3^\frac{4}{3}+j_3^\frac{2}{3}j_1^\frac{2}{3}+j_1^\frac{4}{3}}{2}I_2> 0,\label{F2www5}\\
F^2_{\pm 2}=\frac{K_2^2(y_{\pm 2}^2z^2_{\pm 2}-y_{\pm 2}^2+z_{\pm 2}^2)}{2(y_{\pm 2}z_{\pm 2 })^{\frac{4}{3}}}I_2\neq 0,\qquad t_2=-2(\sqrt[3]{\frac{j_1}{j_3}}+\sqrt[3]{\frac{j_3}{j_1}}),\nonumber\\
F^2_{+2}\neq F^2_{-2},\qquad \mbox{$F^2_{\pm 2}$ do not coincide with $F^2_{-1}$}.\nonumber
\end{eqnarray}
\item in the case $ j_3\neq j_1$, $j_2^{\frac{2}{3}}<j_3^{\frac{2}{3}}+j_1^{\frac{2}{3}}$,
\begin{eqnarray}
F^2_{\pm 1}=\frac{K_1^2(y_{\pm 1}^2z^2_{\pm 1}-y_{\pm 1}^2+z_{\pm 1}^2)}{2(y_{\pm 1}z_{\pm 1 })^{\frac{4}{3}}}I_2\neq 0,\qquad F^2_{+1}\neq F^2_{-1}.\label{F2www6}
\end{eqnarray}
\end{enumerate}
In all cases of Lemma, the expression $K$ is the invariant for each pair of solutions (see Lemmas \ref{thSym0} and \ref{thSh2}).

\subsubsection{The case \texorpdfstring{$d_J=0$}{dJ=0}, \texorpdfstring{$\rank (J)=2$}{rank J=2}.}

In the case $d_J=0$, $x_J=2$, $y_J=0$, we have
$$\Sigma^J=\left(
                 \begin{array}{ccc}
                  j_1 & 0 & 0 \\
                  0 & j_2 & 0 \\
                  \Zero & \Zero & \Zero \\ \hline
                  \Zero & \Zero & \Zero \\
                 \end{array}
               \right),\qquad j_1, j_2\neq 0.
$$
We have the following possible cases: 1) $d_A=0$, $x_A=2$, $y_1=0$; 2) $d_A=1$, $x_A=2$, $y_A=0$; 3) $d_A=0$, $x_A=3$, $y_A=0$; 4) $d_A=0$, $x_A=2$, $y_A=1$, and we are looking for a solution in the form
$$
\Psi^A=\left(
                 \begin{array}{ccc}
                  a_1 & 0 & 0 \\
                  0 & a_2 & 0 \\
                  \Zero & \Zero & \Zero \\ \hline
                  \Zero & \Zero & \Zero \\
                 \end{array}
               \right),\,
               \left(
                 \begin{array}{ccc}
                  a_1 & 0 & 0 \\
                  0 & a_2 & 0 \\
                  0 & 0& 1\\
                  \Zero & \Zero & \Zero \\ \hline
                  0 & 0 & 1\\
                  \Zero & \Zero & \Zero \\
                 \end{array}
               \right),\,
               \left(
                 \begin{array}{ccc}
                  a_1 & 0 & 0 \\
                  0 & a_2 & 0 \\
                  0 & 0 & a_3 \\
                  \Zero & \Zero & \Zero \\ \hline
                  \Zero & \Zero & \Zero \\
                 \end{array}
               \right),\,
               \left(
                 \begin{array}{ccc}
                  a_1 & 0 & 0 \\
                  0 & a_2 & 0 \\
                  \Zero & \Zero & \Zero \\ \hline
                  0 & 0 & a_3\\
                  \Zero & \Zero & \Zero \\
                 \end{array}
               \right),
$$
where $a_1, a_2, a_3\neq 0$, respectively. In the first case, we have
$$-a_1 a_2^2=j_1,\qquad -a_2 a_1^2=j_2,$$
and obtain the unique solution
\begin{eqnarray}
\Psi^A=\left(
                 \begin{array}{ccc}
                  a_1 & 0 & 0 \\
                  0 & a_2 & 0 \\
                  \Zero & \Zero & \Zero \\ \hline
                  \Zero & \Zero & \Zero \\
                 \end{array}
               \right),\qquad
               a_1=-\sqrt[3]{\frac{j_2^2}{j_1}},\quad a_2=-\sqrt[3]{\frac{j_1^2}{j_2}}, \label{Ad020}\\
F^{12}=-\sqrt[3]{j_1 j_2}\tau^3,\label{Fd020}\\
F^2=\frac{-1}{2}\sqrt[3]{(j_1 j_2)^2}I_2\neq 0\label{Fqd020}.
\end{eqnarray}
In the second case, one of the equations is $-(a_1^2+a_2^2)=0$. In the third and fourth cases, one of the equations is
$-a_3(a_1^2+a_2^2)=0$. Thus we have no solutions in these cases.

In the case $d_J=0$, $x_J=0$, $y_J=2$, we have
$$\Sigma^J=\left(
                 \begin{array}{ccc}
                  \Zero & \Zero & \Zero \\ \hline
                  j_1 & 0 & 0 \\
                  0 & j_2 & 0 \\
                 \Zero & \Zero & \Zero \\
                 \end{array}
               \right),\qquad j_1, j_2\neq 0.
$$
We have the following possible cases: 1) $d_A=0$, $x_A=0$, $y_1=2$; 2) $d_A=1$, $x_A=0$, $y_A=2$; 3) $d_A=0$, $x_A=0$, $y_A=3$; 4) $d_A=0$, $x_A=1$, $y_A=2$, and we are looking for a solution in the form
$$
\Psi^A=\left(
                 \begin{array}{ccc}
                  \Zero & \Zero & \Zero \\ \hline
                  a_1 & 0 & 0 \\
                  0 & a_2 & 0 \\
                 \Zero & \Zero & \Zero \\
                 \end{array}
               \right),\,
               \left(
                 \begin{array}{ccc}
                  0 & 0& 1\\
                  \Zero & \Zero & \Zero \\ \hline
                  a_1 & 0 & 0 \\
                  0 & a_2 & 0 \\
                  0 & 0 & 1\\
                  \Zero & \Zero & \Zero \\
                 \end{array}
               \right),\,
               \left(
                 \begin{array}{ccc}
                  \Zero & \Zero & \Zero \\ \hline
                  a_1 & 0 & 0 \\
                  0 & a_2 & 0 \\
                  0 & 0 & a_3 \\
                 \Zero & \Zero & \Zero \\
                 \end{array}
               \right),\,
               \left(
                 \begin{array}{ccc}
                0 & 0 & a_3\\
                  \Zero & \Zero & \Zero \\ \hline
                 a_1 & 0 & 0 \\
                  0 & a_2 & 0 \\
                  \Zero & \Zero & \Zero \\
                 \end{array}
               \right),
$$
where $a_1, a_2, a_3\neq 0$, respectively. In the first case, we have
$$a_1 a_2^2=j_1,\qquad a_2 a_1^2=j_2,$$
and obtain the unique solution
\begin{eqnarray}
\Psi^A=\left(
                 \begin{array}{ccc}
                  \Zero & \Zero & \Zero \\ \hline
                  a_1 & 0 & 0 \\
                  0 & a_2 & 0 \\
                 \Zero & \Zero & \Zero \\
                 \end{array}
               \right),\qquad
a_1=\sqrt[3]{\frac{j_2^2}{j_1}},\quad a_2=\sqrt[3]{\frac{j_1^2}{j_2}}, \label{Ad002}\\
F^{12}=-\sqrt[3]{j_1 j_2}\tau^3,\label{Fd002}\\
F^2=\frac{-1}{2}\sqrt[3]{(j_1 j_2)^2}I_2\neq 0.\label{Fqd002}
\end{eqnarray}
In the second case, one of the equations is $(a_1^2+a_2^2)=0$. In the third and fourth cases, one of the equations is
$a_3(a_1^2+a_2^2)=0$. Thus we have no solutions in these cases.

In the case $d_J=0$, $x_J=1$, $y_J=1$, we have
$$\Sigma^J=\left(
                 \begin{array}{ccc}
                  j_1 & 0 & 0 \\
                  \Zero & \Zero & \Zero \\ \hline
                 0 & j_2 & 0 \\
                  \Zero & \Zero & \Zero \\
                 \end{array}
               \right),\qquad j_1, j_2\neq 0.
$$
We have the following possible cases: 1) $d_A=0$, $x_A=1$, $y_1=1$; 2) $d_A=1$, $x_A=1$, $y_A=1$; 3) $d_A=0$, $x_A=2$, $y_A=1$; 4) $d_A=0$, $x_A=1$, $y_A=2$, and we are looking for a solution in the form
$$
\Psi^A=\left(
                 \begin{array}{ccc}
                  a_1 & 0 & 0 \\
                  \Zero & \Zero & \Zero \\ \hline
                  0 & a_2 & 0 \\
                  \Zero & \Zero & \Zero \\
                 \end{array}
               \right),\,
               \left(
                 \begin{array}{ccc}
                  a_1 & 0 & 0 \\
                  0 & 0& 1\\
                  \Zero & \Zero & \Zero \\ \hline
                 0 & a_2 & 0 \\
                  0 & 0 & 1\\
                  \Zero & \Zero & \Zero \\
                 \end{array}
               \right),\,
               \left(
                 \begin{array}{ccc}
                  a_1 & 0 & 0 \\
                  0 & 0 & a_3 \\
                  \Zero & \Zero & \Zero \\ \hline
                  0 & a_2 & 0 \\
                  \Zero & \Zero & \Zero \\
                 \end{array}
               \right),\,
               \left(
                 \begin{array}{ccc}
                  a_1 & 0 & 0 \\
                  \Zero & \Zero & \Zero \\ \hline
                 0 & a_2 & 0 \\
                  0 & 0 & a_3\\
                  \Zero & \Zero & \Zero \\
                 \end{array}
               \right),
$$
where $a_1, a_2, a_3\neq 0$, respectively. In the first case, we have
$$a_1 a_2^2=j_1,\qquad -a_2 a_1^2=j_2,$$
and obtain the solution
\begin{eqnarray}
\Psi^A=\left(
                 \begin{array}{ccc}
                  a_1 & 0 & 0 \\
                  \Zero & \Zero & \Zero \\ \hline
                  0 & a_2 & 0 \\
                  \Zero & \Zero & \Zero \\
                 \end{array}
               \right),\qquad a_1=\sqrt[3]{\frac{j_2^2}{j_1}},\quad a_2=-\sqrt[3]{\frac{j_1^2}{j_2}}, \label{Ad011}\\
F^{12}=\sqrt[3]{j_1 j_2}\tau^3,\label{Fd011}\\
F^2=\frac{1}{2}\sqrt[3]{(j_1 j_2)^2}I_2\neq 0.\label{Fqd011}
\end{eqnarray}
In the second case, we have
$$a_1 a_2^2=j_1,\qquad -a_2 a_1^2=j_2,\qquad -(a_1^2-a_2^2)=0.$$
In the subcase $j_1=j_2$ (we consider only positive $j_1$, $j_2$), we obtain the following solution
\begin{eqnarray}
\Psi^A=\left(
                 \begin{array}{ccc}
                  a_1 & 0 & 0 \\
                  0 & 0& 1\\
                  \Zero & \Zero & \Zero \\ \hline
                 0 & a_2 & 0 \\
                  0 & 0 & 1\\
                  \Zero & \Zero & \Zero \\
                 \end{array}
               \right),\qquad a_1=\sqrt[3]{j_1},\quad a_2=-\sqrt[3]{j_2},\label{Ad011a}\\
F^{1\,p+1}=\sqrt[3]{j_1j_2}\tau^3,\label{Fd011a}\\
F^{12}=F^{1\,p+2}=\sqrt[3]{j_1}\tau^2,\quad F^{p+1\, 2}=F^{p+1\,p+2}=\sqrt[3]{j_2}\tau^1,\nonumber\\
F^2=\frac{1}{2}\sqrt[3]{(j_1j_2)^2}I_2\neq 0.\label{Fqd011a}
\end{eqnarray}
In the third case, we have
$$-a_1 (-a_2^2+a_3^2)=j_1,\qquad -a_2 (a_1^2+a_3^2)=j_2,\qquad -a_3(a_1^2-a_2^2)=0.$$
In the subcase $j_2>j_1$, we have the following 4 solutions (see Case (v) of Lemma \ref{thSh2})
\begin{eqnarray}
\Psi^A=\left(
                 \begin{array}{ccc}
                  a_1 & 0 & 0 \\
                  0 & 0 & a_3 \\
                  \Zero & \Zero & \Zero \\ \hline
                  0 & a_2 & 0 \\
                  \Zero & \Zero & \Zero \\
                 \end{array}
               \right),\label{Ad011b}\\
a_3=\mp\sqrt{\frac{j_2+j_1}{-2a_2}}=\mp\frac{\sqrt{j_2+j_1}}{\sqrt[6]{4(j_2-j_1)}},\qquad a_2=a_1=-\sqrt[3]{\frac{j_2-j_1}{2}};\nonumber\\
a_3=\mp\sqrt{\frac{j_2-j_1}{-2a_2}}=\mp\frac{\sqrt{j_2-j_1}}{\sqrt[6]{4(j_2+j_1)}},\qquad a_2=-a_1=-\sqrt[3]{\frac{j_2+j_1}{2}},\nonumber\\
F^{1\,p+1}=-a_1a_2\tau^3,\, F^{p+1\, 2}=-a_2 a_3\tau^1,\, F^{21}=-a_3 a_1\tau^2,\label{Fd011b}\\
F^2=\frac{a_1^4}{2}=\frac{1}{2}(\frac{j_2-j_1}{2})^{\frac{4}{3}}I_2\quad \mbox{or} \quad \frac{1}{2}(\frac{j_2+j_1}{2})^{\frac{4}{3}}I_2 \neq 0.\label{Fqd011b}
\end{eqnarray}
In the fourth case, we have
$$a_1(a_2^2+a_3^2)=j_1,\qquad -a_2 (a_1^2-a_3^2)=j_2,\qquad -a_3(a_1^2-a_2^2)=0.$$
In the subcase $j_1>j_2$, we have the following 4 solutions (see Case (v) of Lemma \ref{thSh2})
\begin{eqnarray}
\left(
                 \begin{array}{ccc}
                  a_1 & 0 & 0 \\
                  \Zero & \Zero & \Zero \\ \hline
                 0 & a_2 & 0 \\
                  0 & 0 & a_3\\
                  \Zero & \Zero & \Zero \\
                 \end{array}
               \right)\label{Ad011d}\\
a_3=\pm\sqrt{\frac{j_1+j_2}{2a_1}}=\pm\frac{\sqrt{j_1+j_2}}{\sqrt[6]{4(j_1-j_2)}},\qquad a_1=a_2=\sqrt[3]{\frac{j_1-j_2}{2}};\nonumber\\
a_3=\pm\sqrt{\frac{j_1-j_2}{2a_1}}=\pm\frac{\sqrt{j_1-j_2}}{\sqrt[6]{4(j_1+j_2)}},\qquad a_1=-a_2=\sqrt[3]{\frac{j_1+j_2}{2}},\nonumber\\
F^{1\,p+1}=-a_1a_2\tau^3,\, F^{p+1\, p+2}=-a_2 a_3\tau^1,\, F^{p+2\, 1}=-a_3 a_1\tau^2,\label{Fd011d}\\
F^2=\frac{a_1^4}{2}=\frac{1}{2}(\frac{j_1-j_2}{2})^{\frac{4}{3}}I_2\quad \mbox{or} \quad \frac{1}{2}(\frac{j_1+j_2}{2})^{\frac{4}{3}}I_2 \neq 0.\label{Fqd011d}
\end{eqnarray}

\subsubsection{The case \texorpdfstring{$d_J=0$}{dJ=0}, \texorpdfstring{$\rank (J)=1$}{rank J=1}.}

In the case $d_J=0$, $x_J=1$, $y_J=0$, we have
$$\Sigma^J=\left(
                 \begin{array}{ccc}
                  j_1 & 0 & 0 \\
                  \Zero & \Zero & \Zero \\ \hline
                  \Zero & \Zero & \Zero \\
                 \end{array}
               \right),\qquad j_1\neq 0.
$$

We have the following 10 possible cases: $(d_A, x_A, y_A)=(0, 1, 0)$, $(1, 1, 0)$, $(0, 2, 0)$, $(0, 1, 1)$, $(2, 1, 0)$, $(1, 2, 0)$, $(1, 1, 1)$, $(0, 3, 0)$, $(0, 2, 1)$, $(0, 1, 2)$, and we are looking for a solution in the following form $\Psi^A$:
$$
\left(
                 \begin{array}{ccc}
                  a_1 & 0 & 0 \\
                  \Zero & \Zero & \Zero \\ \hline
                  \Zero & \Zero & \Zero \\
                 \end{array}
               \right),
               \left(
                 \begin{array}{ccc}
                  a_1 & 0 & 0 \\
                  0 & 1& 0\\
                  \Zero & \Zero & \Zero \\ \hline
                  0 & 1 & 0\\
                  \Zero & \Zero & \Zero \\
                 \end{array}
               \right),
               \left(
                 \begin{array}{ccc}
                  a_1 & 0 & 0 \\
                  0 & a_2 & 0 \\
                  \Zero & \Zero & \Zero \\ \hline
                  \Zero & \Zero & \Zero \\
                 \end{array}
               \right),
               \left(
                 \begin{array}{ccc}
                  a_1 & 0 & 0 \\
                  \Zero & \Zero & \Zero \\ \hline
                  0 & a_2 & 0 \\
                  \Zero & \Zero & \Zero \\
                 \end{array}
               \right),
                  \left(
                 \begin{array}{ccc}
                  a_1 & 0 & 0 \\
                  0 & 1 & 0\\
                  0 & 0 & 1\\
                  \Zero & \Zero & \Zero \\ \hline
                  0 & 1 & 0 \\
                  0 & 0 & 1\\
                  \Zero & \Zero & \Zero \\
                 \end{array}
               \right),
$$
$$\!\!
\left(        \begin{array}{ccc}
                  a_1 & 0 & 0 \\
                  0 & a_2 & 0\\
                  0 & 0 & 1\\
                  \Zero & \Zero & \Zero \\ \hline
                  0 & 0 & 1\\
                \Zero & \Zero & \Zero \\
                 \end{array}
               \right),
               \left(
                 \begin{array}{ccc}
                  a_1 & 0 & 0 \\
                  0 & 0& 1\\
                  \Zero & \Zero & \Zero \\ \hline
                  0 & a_2 & 0\\
                  0 & 0 & 1\\
                  \Zero & \Zero & \Zero \\
                 \end{array}
               \right),
               \left(
                 \begin{array}{ccc}
                  a_1 & 0 & 0 \\
                  0 & a_2 & 0 \\
                  0 & 0 & a_3\\
                  \Zero & \Zero & \Zero \\ \hline
                  \Zero & \Zero & \Zero \\
                 \end{array}
               \right),
               \left(
                 \begin{array}{ccc}
                  a_1 & 0 & 0 \\
                  0 & a_2 & 0\\
                  \Zero & \Zero & \Zero \\ \hline
                  0 & 0 & a_3 \\
                  \Zero & \Zero & \Zero \\
                 \end{array}
               \right),
                \left(
                 \begin{array}{ccc}
                  a_1 & 0 & 0 \\
                  \Zero & \Zero & \Zero \\ \hline
                  0 & a_2 & 0\\
                  0 & 0 & a_3 \\
                  \Zero & \Zero & \Zero \\
                 \end{array}
               \right),
$$
where $a_1, a_2, a_3 \neq 0$, respectively. In the first nine cases, we have no solutions. In the tenth case, we have
$$-a_1(-a_2^2-a_3^2)=j_1,\qquad -a_2(a_1^2-a_3^2)=0,\qquad -a_3(a_1^2-a_2^2)=0,$$
and obtain the following 4 solutions (see Case (ii) of Lemma \ref{thSh2}):
\begin{eqnarray}
\!\!\!\!\!\!\!\!\!\!\!\!\!\!\!\!\!\!\!\!\Psi^A=\left(
                 \begin{array}{ccc}
                  a_1 & 0 & 0 \\
                  \Zero & \Zero & \Zero \\ \hline
                  0 & a_2 & 0\\
                  0 & 0 & a_3 \\
                  \Zero & \Zero & \Zero \\
                 \end{array}
               \right),\quad
a_1=\sqrt[3]{\frac{j_1}{2}},\, a_2=\pm a_1,\, a_3= \pm a_1,\quad\mbox{and}\quad a_3=\mp a_1\label{Ad010}
\end{eqnarray}
with
\begin{eqnarray}
F^{1\, p+1}=-a_1a_2\tau^3,\quad F^{p+1\, p+2}=-a_2a_3\tau^1,\quad F^{p+2\, 1}=-a_3a_1\tau^2,\label{Fd010}\\
F^2=\frac{a_1^4}{2}=\frac{1}{2}(\frac{j_1}{2})^{\frac{4}{3}}I_2\neq 0\label{Fqd010}.
\end{eqnarray}

In the case $d_J=0$, $x_J=0$, $y_J=1$,  the situation is similar. We have
$$\Sigma^J=\left(
                 \begin{array}{ccc}
                  \Zero & \Zero & \Zero \\ \hline
                  j_1 & 0 & 0 \\
                  \Zero & \Zero & \Zero \\
                 \end{array}
               \right),\qquad j_1\neq 0.
$$

We have 10 possible cases for $d_A$, $x_A$, and $y_A$, but we have solutions only in the case $d_A=0$, $x_A=2$, $y_A=1$.
We get the system
$$-a_1(a_2^2+a_3^2)=j_1,\quad -a_2(a_3^2-a_1^2)=0,\quad -a_3(a_2^2-a_1^2)=0,$$
and 4 solutions
\begin{eqnarray}
\!\!\!\!\!\!\!\!\!\!\!\!\!\!\!\!\!\!\!\!\!\!\!\!\!\!\!\!\Psi^A=\left(
                 \begin{array}{ccc}
                 0 & a_2 & 0\\
                 0 & 0 & a_3\\
                  \Zero & \Zero & \Zero \\ \hline
                  a_1 & 0 & 0 \\
                  \Zero & \Zero & \Zero \\
                 \end{array}
               \right),\quad a_1=-\sqrt[3]{\frac{j_1}{2}},\, a_2=\pm a_1,\, a_3= \pm a_1,\quad\mbox{and}\quad a_3=\mp a_1\label{Ad001}
\end{eqnarray}
with
\begin{eqnarray}
F^{1\, p+1}=-a_1a_2\tau^3,\quad F^{p+1\, p+2}=-a_2a_3\tau^1,\quad F^{p+2\, 1}=-a_3a_1\tau^2,\label{Fd001}\\
F^2=\frac{a_1^4}{2}=\frac{1}{2}(\frac{j_1}{2})^{\frac{4}{3}}I_2\neq 0\label{Fqd001}.
\end{eqnarray}

\subsubsection{The case \texorpdfstring{$d_J=0$}{dJ=0}, \texorpdfstring{$\rank (J)=0$}{rank J=0}.}

Let us consider the case of zero matrix $\Sigma^J=\Zero$. We must verify all 20 cases for $(d_A, x_A, y_A)$, $ 0\leq d_A, x_A, y_A \leq 3$, $d_A+x_A+y_A \leq 3$. We have no solutions in 14 cases, and we have solutions in 6 cases. Below we present these 6 cases.

In the case $d_A=x_A=y_A=0$, we have zero matrix $\Psi^A=\Zero$. So, we have trivial solution $A=0$, $F=0$, $F^2=0$.

In the case $d_A=0$, $x_A=1$, $y_A=0$, we have the solutions
\begin{eqnarray}
\Psi^A=\left(
                 \begin{array}{ccc}
                 a_1 & 0 & 0\\
                  \Zero & \Zero & \Zero \\ \hline
                  \Zero & \Zero & \Zero \\
                 \end{array}
               \right),\qquad a_1\in\R\setminus\{0\}.\label{Ad000c}
\end{eqnarray}

In the case $d_A=0$, $x_A=0$, $y_A=1$, we have the solutions
\begin{eqnarray}
\Psi^A=\left(
                 \begin{array}{ccc}
                  \Zero & \Zero & \Zero \\ \hline
                  a_1 & 0 & 0\\
                  \Zero & \Zero & \Zero \\
                 \end{array}
               \right),\quad a_1\in\R\setminus\{0\}.\label{Ad000d}
\end{eqnarray}
In these two cases, we have $F=0$, $F^2=0$.

In the case $d_A=1$, $x_A=y_A=0$, we have the solution
\begin{eqnarray}
\Psi^A=\left(
                 \begin{array}{ccc}
                 1 & 0 & 0\\
                  \Zero & \Zero & \Zero \\ \hline
                  1 & 0 & 0 \\
                  \Zero & \Zero & \Zero \\
                 \end{array}
               \right)\label{Ad000}
\end{eqnarray}
with $F=0$, $F^2=0$.

In the case $d_A=2$, $x_A=y_A=0$, we have the solution
\begin{eqnarray}
\Psi^A=\left(
                 \begin{array}{ccc}
                 1 & 0 & 0\\
                 0 & 1 & 0\\
                  \Zero & \Zero & \Zero \\ \hline
                  1 & 0 & 0 \\
                  0 & 1 & 0\\
                  \Zero & \Zero & \Zero \\
                 \end{array}
               \right)\label{Ad000a}
\end{eqnarray}
with
\begin{eqnarray}
F^{12}=F^{1\,p+2}=F^{p+1\,2}=F^{p+1\,p+2}=-1 \tau^3\label{Fd000a}
\end{eqnarray}
and $F^2=0$.

In the case $d_A=3$, $x_A=y_A=0$, we have the solution
\begin{eqnarray}
\Psi^A=\left(
                 \begin{array}{ccc}
                 1 & 0 & 0\\
                 0 & 1 & 0\\
                 0 & 0 & 1\\
                  \Zero & \Zero & \Zero \\ \hline
                  1 & 0 & 0 \\
                  0 & 1 & 0 \\
                  0 & 0 & 1\\
                  \Zero & \Zero & \Zero \\
                 \end{array}
               \right)\label{Ad000b}
\end{eqnarray}
with
\begin{eqnarray}
F^{12}=F^{1\,p+2}=F^{p+1\,2}=F^{p+1\,p+2}=-1 \tau^3,\label{Fd000b}\\
F^{23}=F^{2\,p+3}=F^{p+2\,3}=F^{p+2\,p+3}=-1 \tau^1,\nonumber\\
F^{31}=F^{3\,p+1}=F^{p+3\,1}=F^{p+3\,p+1}=-1 \tau^2,\nonumber
\end{eqnarray}
and $F^2=0$.

\subsection{The case \texorpdfstring{$d_J=1$}{dJ=1}.}

Let us consider the case $d_J=1$. We have $d_A=1$. The problem actually is reduced to solving the system of two cubic equations (\ref{casen22}).

If 1) $d_J=1$, $x_J=2$, $y_J=0$, or 2) $d_J=1$, $x_J=0$, $y_J=2$, or 3) $d_J=1$, $x_J=1$, $y_J=1$, then we have
$$\Sigma^J=\left(
                 \begin{array}{ccc}
                  j_1 & 0 & 0 \\
                  0 & j_2 & 0 \\
                  0 & 0 & 1 \\
                  \Zero & \Zero & \Zero \\ \hline
                  0 & 0 & 1 \\
                  \Zero & \Zero & \Zero \\
                 \end{array}
               \right),\quad\mbox{or}\,
\left(
                 \begin{array}{ccc}
                  0 & 0 & 1 \\
                  \Zero & \Zero & \Zero \\ \hline
                  j_1 & 0 & 0 \\
                  0 & j_2 & 0 \\
                  0 & 0 & 1 \\
                  \Zero & \Zero & \Zero \\
                 \end{array}
               \right),\quad\mbox{or}\,
\left(
                 \begin{array}{ccc}
                  j_1 & 0 & 0 \\
                  0 & 0 & 1 \\
                  \Zero & \Zero & \Zero \\ \hline
                  0 & j_2 & 0 \\
                  0 & 0 & 1 \\
                  \Zero & \Zero & \Zero \\
                 \end{array}
               \right),\qquad j_1, j_2\neq 0.
$$
We have 1) $d_A=1$, $x_A=2$, $y_A=0$, or 2) $d_A=1$, $x_A=0$, $y_A=2$, or 3) $d_A=1$, $x_A=y_A=1$ (respectively), and we are looking for a solution in the form
$$
\Psi^A=\left(
                 \begin{array}{ccc}
                  a_1 & 0 & 0 \\
                  0 & a_2 & 0 \\
                  0 & 0 & \beta \\
                  \Zero & \Zero & \Zero \\ \hline
                  0 & 0 & \beta \\
                  \Zero & \Zero & \Zero \\
                 \end{array}
               \right),\quad\mbox{or}\,
\left(
                 \begin{array}{ccc}
                  0 & 0 & \beta \\
                  \Zero & \Zero & \Zero \\ \hline
                  a_1 & 0 & 0 \\
                  0 & a_2 & 0 \\
                  0 & 0 & \beta \\
                  \Zero & \Zero & \Zero \\
                 \end{array}
               \right),\quad\mbox{or}\,
\left(
                 \begin{array}{ccc}
                  a_1 & 0 & 0 \\
                  0 & 0 & \beta \\
                  \Zero & \Zero & \Zero \\ \hline
                  0 & a_2 & 0 \\
                  0 & 0 & \beta \\
                  \Zero & \Zero & \Zero \\
                 \end{array}
               \right),\quad a_1, a_2, \beta\neq 0.
$$
respectively. In the first case, we obtain
$$-a_1 a_2^2=j_1,\quad -a_2 a_1^2=j_2,\quad -\beta(a_1^2+a_2^2)=1,$$
and the solution
\begin{eqnarray}
\!\!\!\!\!\!\!\!\!\!\Psi^A=\left(
                 \begin{array}{ccc}
                  a_1 & 0 & 0 \\
                  0 & a_2 & 0 \\
                  0 & 0 & \beta \\
                  \Zero & \Zero & \Zero \\ \hline
                  0 & 0 & \beta \\
                  \Zero & \Zero & \Zero \\
                 \end{array}
               \right),\quad a_1=-\sqrt[3]{\frac{j_2^2}{j_1}},\quad a_2=-\sqrt[3]{\frac{j_1^2}{j_2}},\quad \beta=\frac{-1}{a_1^2+a_2^2},\label{Ad11}\\
F^{12}=-a_1a_2\tau^3=-\sqrt[3]{j_1 j_2}\tau^3,\label{Fd11}\\
F^{13}=F^{1\, p+1}=a_1\beta \tau^2,\quad F^{23}=F^{2\,p+1}=-a_2\beta\tau^1,\nonumber\\
F^2=\frac{-1}{2}(a_1 a_2)^2 I_2=\frac{-1}{2}\sqrt[3]{(j_1j_2)^2} I_2\neq 0.\label{Fqd11}
\end{eqnarray}
In the second case, we obtain
$$a_1 a_2^2=j_1,\quad a_2 a_1^2=j_2,\quad \beta(a_1^2+a_2^2)=1,$$
and the solution
\begin{eqnarray}
\Psi^A=\left(
                 \begin{array}{ccc}
                  0 & 0 & \beta \\
                  \Zero & \Zero & \Zero \\ \hline
                  a_1 & 0 & 0 \\
                  0 & a_2 & 0 \\
                  0 & 0 & \beta \\
                  \Zero & \Zero & \Zero \\
                 \end{array}
               \right),\quad a_1=\sqrt[3]{\frac{j_2^2}{j_1}},\quad a_2=\sqrt[3]{\frac{j_1^2}{j_2}},\quad \beta=\frac{1}{a_1^2+a_2^2},\label{Ad12}\\
F^{12}=-a_1a_2\tau^3=-\sqrt[3]{j_1 j_2}\tau^3,\label{Fd12}\\
F^{1\,p+1}=F^{p+3\, p+1}=-a_1\beta \tau^2,\quad F^{1\,p+2}=F^{p+3\,p+2}=a_2\beta\tau^1,\nonumber\\
F^2=\frac{-1}{2}(a_1 a_2)^2 I_2=\frac{-1}{2}\sqrt[3]{(j_1j_2)^2} I_2\neq 0.\label{Fqd12}
\end{eqnarray}
In the third case, we obtain
$$a_1 a_2^2=j_1,\quad -a_2 a_1^2=j_2,\quad -\beta(a_1^2-a_2^2)=1.$$
This system has a solution only in the case $j_1\neq j_2$ (we consider only positive $j_1$, $j_2$):
\begin{eqnarray}
\Psi^A=\left(
                 \begin{array}{ccc}
                  a_1 & 0 & 0 \\
                  0 & 0 & \beta \\
                  \Zero & \Zero & \Zero \\ \hline
                  0 & a_2 & 0 \\
                  0 & 0 & \beta \\
                  \Zero & \Zero & \Zero \\
                 \end{array}
               \right),\quad a_1=\sqrt[3]{\frac{j_2^2}{j_1}},\quad a_2=-\sqrt[3]{\frac{j_1^2}{j_2}},\quad \beta=\frac{-1}{a_1^2-a_2^2},\label{Ad13}\\
F^{1\, p+1}=-a_1a_2\tau^3=\sqrt[3]{j_1 j_2}\tau^3,\label{Fd13}\\
F^{12}=F^{1\, p+2}=a_1\beta \tau^2,\quad F^{p+1\,2}=F^{p+1\,p+2}=-a_2\beta\tau^1,\nonumber\\
F^2=\frac{1}{2}(a_1 a_2)^2 I_2=\frac{1}{2}\sqrt[3]{(j_1j_2)^2} I_2\neq 0.\label{Fqd13}
\end{eqnarray}

If $d_J=1$, $x_J=1$, $y_J=0$, then we have
$$\Sigma^J=\left(
                 \begin{array}{ccc}
                  j_1 & 0 & 0 \\
                  0 & 1 & 0 \\
                  \Zero & \Zero & \Zero \\ \hline
                  0 & 1 & 0 \\
                  \Zero & \Zero & \Zero \\
                 \end{array}
               \right),\qquad j_1\neq 0.
$$
We have 1) $d_A=1$, $x_A=1$, $y_A=1$, or 2) $d_A=1$, $x_A=2$, $y_A=0$, or 3) $d_A=1$, $x_A=1$, $y_A=0$, and we are looking for a solution in the form
$$
\Psi^A=\left(
                 \begin{array}{ccc}
                  a_1 & 0 & 0 \\
                  0 & \beta & 0 \\
                  \Zero & \Zero & \Zero \\ \hline
                  0 & \beta & 0 \\
                  0 & 0 & a_2\\
                  \Zero & \Zero & \Zero \\
                 \end{array}
               \right),\quad\mbox{or}\quad
\left(
                 \begin{array}{ccc}
                  a_1 & 0 & 0\\
                  0 & \beta & 0\\
                  0 & 0 & a_2 \\
                  \Zero & \Zero & \Zero \\ \hline
                  0 & \beta & 0 \\
                  \Zero & \Zero & \Zero \\
                 \end{array}
               \right),\quad\mbox{or}\quad
\left(
                 \begin{array}{ccc}
                  a_1 & 0 & 0 \\
                  0 & \beta & 0 \\
                  \Zero & \Zero & \Zero \\ \hline
                  0 & \beta & 0 \\
                  \Zero & \Zero & \Zero \\
                 \end{array}
               \right),\quad a_1, a_2,\beta\neq 0,
$$
respectively. In the first case, one of the equations is $-a_2 a_1^2=0$. In the second case, one of the equations is
$-a_2 a_1^2=0$. In the third case, one of the equations is $-a_1 0= j_1$. In all these three cases, we have no solutions.

If $d_J=1$, $x_J=0$, $y_J=1$, then we have
$$\Sigma^J=\left(
                 \begin{array}{ccc}
                  0 & 1 & 0 \\
                  \Zero & \Zero & \Zero \\ \hline
                  j_1 & 0 & 0\\
                  0 & 1 & 0 \\
                  \Zero & \Zero & \Zero \\
                 \end{array}
               \right),\qquad j_1\neq 0.
$$
We have 1) $d_A=1$, $x_A=1$, $y_A=1$, or 2) $d_A=1$, $x_A=0$, $y_A=2$, or 3) $d_A=1$, $x_A=0$, $y_A=1$, and we are looking for a solution in the form
$$
\Psi^A=\left(
                 \begin{array}{ccc}
                  0 & \beta & 0 \\
                  0 & 0 & a_2\\
                  \Zero & \Zero & \Zero \\ \hline
                  a_1 & 0 & 0\\
                  0 & \beta & 0 \\
                  \Zero & \Zero & \Zero \\
                 \end{array}
               \right),\quad\mbox{or}\quad
\left(
                 \begin{array}{ccc}
                 0 & \beta & 0 \\
                  \Zero & \Zero & \Zero \\ \hline
                 a_1 & 0 & 0\\
                  0 & \beta & 0\\
                  0 & 0 & a_2 \\
                  \Zero & \Zero & \Zero \\
                 \end{array}
               \right),\quad\mbox{or}\quad
\left(
                 \begin{array}{ccc}
                  0 & \beta & 0 \\
                  \Zero & \Zero & \Zero \\ \hline
                  a_1 & 0 & 0 \\
                  0 & \beta & 0 \\
                  \Zero & \Zero & \Zero \\
                 \end{array}
               \right),\quad a_1, a_2,\beta\neq 0
$$
respectively. In the first case, one of the equations is $a_2 a_1^2=0$. In the second case, one of the equations is
$a_2 a_1^2=0$. In the third case, one of the equations is $-a_1 0= j_1$. In all these three cases, we have no solutions.

If $d_J=1$, $x_J=0$, $y_J=0$, then we have
$$\Sigma^J=\left(
                 \begin{array}{ccc}
                  1 & 0 & 0 \\
                  \Zero & \Zero & \Zero \\ \hline
                  1 & 0 & 0\\
                  \Zero & \Zero & \Zero \\
                 \end{array}
               \right).
$$
We have 1) $d_A=1$, $x_A=2$, $y_A=0$, or 2) $d_A=1$, $x_A=0$, $y_A=2$, or 3) $d_A=1$, $x_A=1$, $y_A=1$, or 4) $d_A=1$, $x_A=1$, $y_A=0$, or 5) $d_A=1$, $x_A=0$, $y_A=1$, or 6) $d_A=1$, $x_A=0$, $y_A=0$, and we are looking for a solution in the form
$$
\Psi^A=\left(
                 \begin{array}{ccc}
                 \beta & 0 & 0\\
                  0 & a_1 & 0 \\
                  0 & 0 & a_2\\
                  \Zero & \Zero & \Zero \\ \hline
                  \beta & 0 & 0\\
                  \Zero & \Zero & \Zero \\
                 \end{array}
               \right),\quad\mbox{or}\quad
\left(
                 \begin{array}{ccc}
                 \beta & 0 & 0 \\
                  \Zero & \Zero & \Zero \\ \hline
                 \beta & 0 & 0\\
                  0 & a_1 & 0\\
                  0 & 0 & a_2 \\
                  \Zero & \Zero & \Zero \\
                 \end{array}
               \right),\quad\mbox{or}\quad
\left(
                 \begin{array}{ccc}
                  \beta & 0 & 0 \\
                  0 & a_1 & 0\\
                  \Zero & \Zero & \Zero \\ \hline
                  \beta & 0 & 0 \\
                  0 & 0 & a_2 \\
                  \Zero & \Zero & \Zero \\
                 \end{array}
               \right),
$$
$$
\mbox{or}\quad\left(
                 \begin{array}{ccc}
                 \beta & 0 & 0\\
                  0 & a_1 & 0 \\
                  \Zero & \Zero & \Zero \\ \hline
                  \beta & 0 & 0\\
                  \Zero & \Zero & \Zero \\
                 \end{array}
               \right),\quad\mbox{or}\quad
\left(
                 \begin{array}{ccc}
                 \beta & 0 & 0 \\
                  \Zero & \Zero & \Zero \\ \hline
                 \beta & 0 & 0\\
                  0 & a_1 & 0\\
                  \Zero & \Zero & \Zero \\
                 \end{array}
               \right),\quad\mbox{or}\quad
\left(
                 \begin{array}{ccc}
                  \beta & 0 & 0 \\
                  \Zero & \Zero & \Zero \\ \hline
                  \beta & 0 & 0 \\
                  \Zero & \Zero & \Zero \\
                 \end{array}
               \right),\quad a_1, a_2,\beta\neq 0,
$$
respectively. In the first case, one of the equations is $-a_1 a_2^2=0$. In the second case, one of the equations is
$a_1 a_2^2=0$. In the third case, one of the equations is $a_1 a_2^2= 0$. In all these three cases, we have no solutions.
In the fourth case, we have
$$-\beta a_1^2=1,\qquad -a_1 0=0,$$
and the solution is
\begin{eqnarray}
\Psi^A=\left(
                 \begin{array}{ccc}
                 \beta & 0 & 0\\
                  0 & a_1 & 0 \\
                  \Zero & \Zero & \Zero \\ \hline
                  \beta & 0 & 0\\
                  \Zero & \Zero & \Zero \\
                 \end{array}
               \right),\qquad a_1\in\R\setminus\{0\},\qquad \beta=\frac{-1}{a_1^2},\label{Ad14}\\
F^{12}=F^{p+1\, 2}=-\beta a_1 \tau^3=\frac{1}{a_1}\tau^3.\label{Fd14}
\end{eqnarray}
In the fifth case, we obtain
$$\beta a_1^2=1,\qquad -a_1 0=0,$$
and the solution is
\begin{eqnarray}
\Psi^A=\left(
                 \begin{array}{ccc}
                 \beta & 0 & 0 \\
                  \Zero & \Zero & \Zero \\ \hline
                 \beta & 0 & 0\\
                  0 & a_1 & 0\\
                  \Zero & \Zero & \Zero \\
                 \end{array}
               \right),\qquad a_1\in\R\setminus\{0\},\qquad \beta=\frac{1}{a_1^2},\label{Ad15}\\
F^{1\,p+2}=F^{p+1\,p+2}=-\beta a_1 \tau^3=\frac{-1}{a_1}\tau^3.\label{Fd15}
\end{eqnarray}
In the sixth case, we obtain $-\beta \, 0=1$, i.e. there is no solution. In the cases (\ref{Fd14}) and (\ref{Fd15}), we have $F^2=0$.

\subsection{The case \texorpdfstring{$d_J=2$}{dJ=2}.}

Let us consider the case $d_J=2$. In this case, we have $d_A=2$.

If $d_J=2$, $x_J=1$, $y_J=0$ or $d_J=2$, $x_J=0$, $y_J=1$, then we have
$$\Sigma^J=\left(
                 \begin{array}{ccc}
                  j_1 & 0 & 0 \\
                  0 & 1 & 0 \\
                  0 & 0 & 1 \\
                  \Zero & \Zero & \Zero \\ \hline
                  0 & 1 & 0 \\
                  0 & 0 & 1 \\
                  \Zero & \Zero & \Zero \\
                 \end{array}
               \right)\qquad \mbox{or} \qquad
\left(
                 \begin{array}{ccc}
                  0 & 1 & 0 \\
                  0 & 0 & 1 \\
                  \Zero & \Zero & \Zero \\ \hline
                  j_1 & 0 & 0\\
                  0 & 1 & 0 \\
                  0 & 0 & 1 \\
                  \Zero & \Zero & \Zero \\
                 \end{array}
               \right),\qquad j_1\neq 0,
$$
respectively. We have $d_A=2$, $x_A=1$, $y_A=0$ or $d_A=2$, $x_A=0$, $y_A=1$ (respectively) and we are looking for a solution in the form
$$
\Psi^A=\left(
                 \begin{array}{ccc}
                  a_1 & 0 & 0 \\
                  0 & \beta & 0 \\
                  0 & 0 & \beta \\
                  \Zero & \Zero & \Zero \\ \hline
                  0 & \beta & 0 \\
                  0 & 0 & \beta \\
                  \Zero & \Zero & \Zero \\
                 \end{array}
               \right)\qquad \mbox{or}\qquad
             \left(
                 \begin{array}{ccc}
                  0 & \beta & 0 \\
                  0 & 0 & \beta \\
                  \Zero & \Zero & \Zero \\ \hline
                  a_1 & 0 & 0 \\
                  0 & \beta & 0 \\
                  0 & 0 & \beta \\
                  \Zero & \Zero & \Zero \\
                 \end{array}
               \right),\qquad
               a_1, \beta\neq 0\nonumber
$$
respectively. In both cases, one of the equations will be $a_1 0 =j_1\neq 0$. We have no solutions in these cases.

If $d_J=2$, $x_J=y_J=0$, then we have
$$\Sigma^J=\left(
                 \begin{array}{ccc}
                  1 & 0 & 0 \\
                  0 & 1 & 0 \\
                  \Zero & \Zero & \Zero \\ \hline
                  1 & 0 & 0 \\
                  0 & 1 & 0 \\
                  \Zero & \Zero & \Zero \\
                 \end{array}
               \right).
$$
We have 1) $d_A=2$, $x_A=1$, $y_A=0$, or 2) $d_A=2$, $x_A=0$, $y_A=1$, or 3) $d_A=2$, $x_A=y_A=0$. We are looking for a solution in the form
$$
\Psi^A=\left(
                 \begin{array}{ccc}
                  \beta & 0 & 0 \\
                  0 & \beta & 0 \\
                  0 & 0 & a_1 \\
                  \Zero & \Zero & \Zero \\ \hline
                  \beta & 0 & 0 \\
                  0 & \beta & 0 \\
                  \Zero & \Zero & \Zero \\
                 \end{array}
               \right),\quad\mbox{or}\quad
  \left(
                 \begin{array}{ccc}
                  \beta & 0 & 0 \\
                  0 & \beta & 0 \\
                  \Zero & \Zero & \Zero \\ \hline
                  \beta & 0 & 0 \\
                  0 & \beta & 0 \\
                  0 & 0 & a_1 \\
                  \Zero & \Zero & \Zero \\
                 \end{array}
               \right),\quad \mbox{or} \quad
\left(
                 \begin{array}{ccc}
                  \beta & 0 & 0 \\
                  0 & \beta & 0 \\
                  \Zero & \Zero & \Zero \\ \hline
                  \beta & 0 & 0 \\
                  0 & \beta & 0 \\
                  \Zero & \Zero & \Zero \\
                 \end{array}
               \right),
               \qquad \beta, a_1 \neq 0,
$$
respectively. In the first case, we obtain the equation $-\beta a_1^2=1$, i.e. the solution
\begin{eqnarray}
\Psi^A=\left(
                 \begin{array}{ccc}
                  \beta & 0 & 0 \\
                  0 & \beta & 0 \\
                  0 & 0 & a_1 \\
                  \Zero & \Zero & \Zero \\ \hline
                  \beta & 0 & 0 \\
                  0 & \beta & 0 \\
                  \Zero & \Zero & \Zero \\
                 \end{array}
               \right),\qquad a_1\in\R\setminus\{0\},\quad \beta=-\frac{1}{a_1^2},\label{Ad21}
\end{eqnarray}
\begin{eqnarray}
&&F^{12}=F^{1\,p+2}=F^{p+1\, 2}=F^{p+1\, p+2}=\frac{-1}{a_1^4}\tau^3,\label{Fd21}\\
&&F^{13}=F^{p+1\, 3}=\frac{-1}{a_1}\tau^2,\qquad F^{23}=F^{p+2\, 3}=\frac{1}{a_1}\tau^1.\nonumber
\end{eqnarray}
In the second case, we obtain the equation $\beta a_1^2=1$, i.e. the solution
\begin{eqnarray}
\left(
                 \begin{array}{ccc}
                  \beta & 0 & 0 \\
                  0 & \beta & 0 \\
                  \Zero & \Zero & \Zero \\ \hline
                  \beta & 0 & 0 \\
                  0 & \beta & 0 \\
                  0 & 0 & a_1 \\
                  \Zero & \Zero & \Zero \\
                 \end{array}
               \right),\qquad a_1\in\R\setminus\{0\},\quad \beta=\frac{1}{a_1^2},\label{Ad22}
\end{eqnarray}
\begin{eqnarray}
&&F^{12}=F^{1\,p+2}=F^{p+1\, 2}=F^{p+1\, p+2}=\frac{-1}{a_1^4}\tau^3,\label{Fd22}\\
&&F^{1\, p+3}=F^{p+1\, p+3}=\frac{1}{a_1}\tau^2,\qquad F^{2\, p+3}=F^{p+2\, p+3}=\frac{-1}{a_1}\tau^1.\nonumber
\end{eqnarray}
In the third case, we obtain the equation $\beta \, 0 = 1$, i.e. there is no solution.

In the first and second cases, we have $F^2=F_{\mu\nu}F^{\mu\nu}=0$.

\subsection{The case \texorpdfstring{$d_J=3$}{dJ=3}.}

Let us consider the case $d_J=3$, $x_J=y_J=0$. We have the following matrix $\Sigma^J$, and we are looking for a solution in the following form $\Psi^A$ (because $d_A=3$, $x_A=y_A=0$):
$$\Sigma^J=\left(
                 \begin{array}{ccc}
                  1 & 0 & 0 \\
                  0 & 1 & 0 \\
                  0 & 0 & 1 \\
                  \Zero & \Zero & \Zero \\ \hline
                  1 & 0 & 0 \\
                  0 & 1 & 0 \\
                  0 & 0 & 1 \\
                  \Zero & \Zero & \Zero \\
                 \end{array}
               \right),\qquad
\Psi^A=\left(
                 \begin{array}{ccc}
                  \beta & 0 & 0 \\
                  0 & \beta & 0 \\
                  0 & 0 & \beta \\
                  \Zero & \Zero & \Zero \\ \hline
                  \beta & 0 & 0 \\
                  0 & \beta & 0 \\
                  0 & 0 & \beta \\
                  \Zero & \Zero & \Zero \\
                 \end{array}
               \right),\qquad \beta\neq 0.
$$
We obtain the equation $\beta \, 0 = 1$, i.e. there is no solution.

\section{Summary for the case of arbitrary \texorpdfstring{$\R^{p,q}$}{Rpq}, \texorpdfstring{$p\geq 1$}{p>1}, \texorpdfstring{$q\geq 1$}{q>1}.}
\label{sect5}

Let we have the Yang--Mills equations (\ref{eq}) in pseudo-Euclidean space $\R^{p,q}$, $p\geq 1$, $q\geq 1$, with the known constant current $J^\mu=J^\mu_{\,\, a}\tau^a$, the unknown constant potential $A^\mu=A^\mu_{\,\, a}\tau^a$, and the corresponding unknown strength $F^{\mu\nu}$ (\ref{str}).
\begin{enumerate}
  \item For the matrix $J=(J^\mu_{\,\, a})$, we calculate three parameters $d_J$, $x_J$, $y_J$, which are uniquely determined, $d_J+x_J+y_J=\rank (J)$, $x_A$ is the number of positive eigenvalues of the matrix $J^\T \eta J$, $y_A$ is the number of negative eigenvalues of the matrix $J^\T\eta J$, $d_J=\rank (J)-\rank(J^\T \eta J)$ (see Theorem \ref{thHSVD}).
  \item We calculate the hyperbolic singular values $j_1$, $j_2$, $j_3$ of the matrix $J$, which are uniquely determined. The corresponding matrices $Q\in\OO(p,q)$, $P\in\SO(3)$ related to the HSVD are not uniquely determined.
  \item For the corresponding $p$, $q$, $d_J$, $x_J$, $y_J$, $j_1$, $j_2$, $j_3$, we get all solutions of (\ref{eq} and they are represented in Tables 2, 3, 4.
      \begin{enumerate}
      \item For each current $J$, the explicit form of these solutions in terms of the potential $A$ and the strength $F$ is given in specific coordinate system (which is determined by the matrix $Q\in\OO(p,q)$ related to the HSVD) and with specific gauge fixing (which is determined by the matrix $P\in\SO(3)$ related to the HSVD). The connection $S\in\SU(2)$, which we use in the gauge fixing, is the two-sheeted covering of the matrix $P\in\SO(3)$.
      \item We can obtain the corresponding solutions in the original system of coordinates and with the original gauge fixing (see Remarks 1 and 2), using $P$ and $Q$.
      \item We calculate the invariant $F^2$ for all constant solutions, which is important from a physical point of view. The expression $F^2$ is gauge invariant and is invariant under pseudo-orthogonal transformations of coordinates. It is present in the Lagrangian of the Yang--Mills field.
      \end{enumerate}
\end{enumerate}

We summarize the results for the case of arbitrary pseudo-Euclidean space $\R^{p,q}$ in Table 2 (the nondegenerate case with $d_J=0$ and $\rank (J)=3$), Table 3 (the degenerate case with $d_J=0$ and $\rank (J)<3$), and Table 4 (the degenerate cases with $d_J \neq 0$). We remind that we consider only the cases with positive numbers (hyperbolic singular values of $J$) $j_1$, $j_2$, $j_3$. In the columns ``$A$'', ``$F$'', and ``$F^2$'', we indicate the number of nonzero solutions in terms of $A$, $F$, and $F^2$, with references to the corresponding explicit formulas. In the column ``add.cond.'', the additional conditions on the hyperbolic singular values $j_1$, $j_2$, and $j_3$ are indicated.

In the particular case of $n$-dimensional Minkowski space $\R^{1,n-1}$, which is the most important case from a physical point of view, not all types of solutions are realized. The conditions on $p$ and $q$ for each solution are indicated in the first column of Tables 2, 3, and 4.

In Table 2, we use the notation 
\begin{eqnarray}
B^*:=\frac{(s^*)^2+2}{s^*}\approx 7.66486>2\sqrt{2},
\end{eqnarray}
where 
\begin{eqnarray}
s^*:=\sqrt{13+\sqrt{193-6^\frac{4}{3}}+\sqrt{386+6^\frac{4}{3}+\frac{5362}{193-6^\frac{4}{3}}}}\approx 7.39438.
\end{eqnarray}

\begin{table}[ht]\caption{ The case $\R^{p,q}$, $p\geq 1$, $q\geq 1$, $d_J=0$, $\rank(J)=3$.}
\begin{center}
  { \scriptsize \begin{tabular}{|c|c|c|c|c|c|c|c|c|c|c|}
  \hline
$p, q$ & $d_J$ & $x_J$ & $y_J$ & add.cond. & $d_A$ & $x_A$ & $y_A$ & $A$ & $F$ & $F^2$ \\ \hline 
$p\geq 3, q\geq 1$ & 0 & 3 & 0 & $j_1=j_2=j_3$ & 0 & 3 & 0 &  1:(\ref{Ad030}) &  1:(\ref{Fd030}) &  1:(\ref{Fqd030})  \\
$p\geq 3, q\geq 1$ & 0 & 3 & 0 & $j_1=j_2>j_3$ & 0 & 3 & 0 &  2:(\ref{Ad030n}) &  2:(\ref{Fd030n}) & 2:(\ref{F21}) \\
$p\geq 3, q\geq 1$ & 0 & 3 & 0 & $j_3>j_1=j_2$ & 0 & 3 & 0 &  2:(\ref{Ad030n}) &  2:(\ref{Fd030n}) & 1:(\ref{F22}) \\
$p\geq 3, q\geq 1$ & 0 & 3 & 0 & all different $j_1, j_2, j_3$ & 0 & 3 & 0 &  2:(\ref{Ad030n}) & 2:(\ref{Fd030n}) & 2:(\ref{F23}) \\ \hline
$p\geq 1, q\geq 3$ & 0 & 0 & 3 & $j_1=j_2=j_3$ & 0 & 0 & 3 &  1:(\ref{Ad003}) &  1:(\ref{Fd003}) &  1:(\ref{Fqd003}) \\
$p\geq 1, q\geq 3$ & 0 & 0 & 3 & $j_1=j_2>j_3$ & 0 & 0 & 3 &  2:(\ref{Ad003n}) &  2:(\ref{Fd003n}) & 2:(\ref{F21}) \\
$p\geq 1, q\geq 3$ & 0 & 0 & 3 & $j_3>j_1=j_2$ & 0 & 0 & 3 &  2:(\ref{Ad003n}) &  2:(\ref{Fd003n}) & 1:(\ref{F22}) \\
$p\geq 1, q\geq 3$ & 0 & 0 & 3 & all different $j_1, j_2, j_3$ & 0 & 0 & 3 &  2:(\ref{Ad003n}) & 2:(\ref{Fd003n}) & 2:(\ref{F23}) \\ \hline
$p\geq 2, q\geq 1$ & 0 & 2 & 1 &$j_1=j_2<\frac{j_3}{2\sqrt{2}}, \frac{j_3}{j_1}=B^*$ & 0 & 2 & 1 &  6:(\ref{Ad0212}) & 6:(\ref{Fd021})&  3:(\ref{F2ww1}) \\
$p\geq 2, q\geq 1$ & 0 & 2 & 1 &$j_1=j_2<\frac{j_3}{2\sqrt{2}}, \frac{j_3}{j_1}\neq B^*$ & 0 & 2 & 1 &  6:(\ref{Ad0212}) & 6:(\ref{Fd021})&  4:(\ref{F2ww1}) \\
$p\geq 2, q\geq 1$ & 0 & 2 & 1 & $j_1=j_2=\frac{j_3}{2\sqrt{2}}$ & 0 & 2 & 1 &  4:(\ref{Ad0212}) & 4:(\ref{Fd021})& $F^2=0$ and 2:(\ref{F2ww3})\\
$p\geq 2, q\geq 1$ & 0 & 2 & 1 & $j_1=j_2>\frac{j_3}{2\sqrt{2}}$ & 0 & 2 & 1 &  2:(\ref{Ad0212}) & 2:(\ref{Fd021})& 2:(\ref{F2ww5}) \\
$p\geq 2, q\geq 1$ & 0 & 2 & 1 & $j_1\neq j_2$, $j_3^{\frac{2}{3}}>j_2^{\frac{2}{3}}+j_1^{\frac{2}{3}}$ & 0 & 2 & 1 & 6:(\ref{Ad0212}) &  6:(\ref{Fd021})& 2-6:(\ref{F2ww6}) \\
$p\geq 2, q\geq 1$ & 0 & 2 & 1 & $j_1\neq j_2$, $j_3^{\frac{2}{3}}=j_2^{\frac{2}{3}}+j_1^{\frac{2}{3}}$ & 0 & 2 & 1 &  4:(\ref{Ad0212}) & 4:(\ref{Fd021})& $F^2=0$ and 3:(\ref{F2ww7}) \\
$p\geq 2, q\geq 1$ & 0 & 2 & 1 & $j_1\neq j_2$, $j_3^{\frac{2}{3}}<j_2^{\frac{2}{3}}+j_1^{\frac{2}{3}}$ & 0 & 2 & 1 &  2:(\ref{Ad0212}) & 2:(\ref{Fd021})& 2:(\ref{F2ww8}) \\ \hline
$p\geq 1, q\geq 2$ & 0 & 1 & 2 &$j_3=j_1<\frac{j_2}{2\sqrt{2}}, \frac{j_2}{j_3}=B^*$ & 0 & 1 & 2 &  6:(\ref{Ad0122}) & 6:(\ref{Fd012})& 3:(\ref{F2www1}) \\
$p\geq 1, q\geq 2$ & 0 & 1 & 2 &$j_3=j_1<\frac{j_2}{2\sqrt{2}}, \frac{j_2}{j_3}\neq B^*$ & 0 & 1 & 2 &  6:(\ref{Ad0122}) & 6:(\ref{Fd012})& 4:(\ref{F2www1}) \\
$p\geq 1, q\geq 2$ & 0 & 1 & 2 & $j_3=j_1=\frac{j_2}{2\sqrt{2}}$ & 0 & 1 & 2 &  4:(\ref{Ad0122}) & 4:(\ref{Fd012})& $F^2=0$ and 2:(\ref{F2www2})\\
$p\geq 1, q\geq 2$ & 0 & 1 & 2 & $j_3=j_1>\frac{j_2}{2\sqrt{2}}$ & 0 & 1 & 2 &  2:(\ref{Ad0122}) & 2:(\ref{Fd012})& 2:(\ref{F2www3}) \\
$p\geq 1, q\geq 2$ & 0 & 1 & 2 & $j_3\neq j_1$, $j_2^{\frac{2}{3}}>j_3^{\frac{2}{3}}+j_1^{\frac{2}{3}}$ & 0 & 1 & 2 & 6:(\ref{Ad0122}) &  6:(\ref{Fd012})& 2-6:(\ref{F2www4}) \\
$p\geq 1, q\geq 2$ & 0 & 1 & 2 & $j_3\neq j_1$, $j_2^{\frac{2}{3}}=j_3^{\frac{2}{3}}+j_1^{\frac{2}{3}}$ & 0 & 1 & 2 & 4:(\ref{Ad0122}) &  4:(\ref{Fd012})& $F^2=0$ and 3:(\ref{F2www5})\\
$p\geq 1, q\geq 2$ & 0 & 1 & 2 & $j_3\neq j_1$, $j_2^{\frac{2}{3}}<j_3^{\frac{2}{3}}+j_1^{\frac{2}{3}}$ & 0 & 1 & 2 & 2:(\ref{Ad0122}) &  2:(\ref{Fd012})& 2:(\ref{F2www6}) \\ \hline
\end{tabular}}
\end{center}
\end{table}

\begin{table}[ht]\caption{ The case $\R^{p,q}$, $p\geq 1$, $q\geq 1$, $d_J=0$, $\rank(J)<3$.}
\begin{center}
  \begin{tabular}{|c|c|c|c|c|c|c|c|c|c|c|c|}
  \hline
$p, q$ & $d_J$ & $x_J$ & $y_J$ & add.cond. & $d_A$ & $x_A$ & $y_A$ & $A$ & $F$ & $F^2$ \\ \hline 
$p\geq 2, q\geq 1$ & 0 & 2 & 0 & & 0 & 2 & 0 & 1:(\ref{Ad020}) & 1:(\ref{Fd020}) &1:(\ref{Fqd020}) \\ \hline
$p\geq 1, q\geq 2$ & 0 & 0 & 2 & & 0 & 0 & 2 & 1:(\ref{Ad002}) & 1:(\ref{Fd002}) & 1:(\ref{Fqd002})\\ \hline
$p\geq 1, q\geq 1$ & 0 & 1 & 1 & & 0 &1  &1  & 1:(\ref{Ad011}) & 1:(\ref{Fd011}) & 1:(\ref{Fqd011}) \\
$p\geq 2, q\geq 2$ & 0 & 1 & 1 & $j_1 = j_2$ & 0 & 2 & 2 & 1:(\ref{Ad011a})  & 1:(\ref{Fd011a})& 1:(\ref{Fqd011a})\\
$p\geq 2, q\geq 1$ & 0 & 1 & 1 & $j_2>j_1$ & 0 & 2 & 1 & 4:(\ref{Ad011b}) & 4:(\ref{Fd011b})& 2:(\ref{Fqd011b})\\
$p\geq 1, q\geq 2$ & 0 & 1 & 1 & $j_1>j_2$ & 0 & 1 & 2 & 4:(\ref{Ad011d}) & 4:(\ref{Fd011d}) & 2:(\ref{Fqd011d})\\ \hline
$p\geq 1, q=1$ & 0 & 1 & 0 & & & & & $\varnothing$ & $\varnothing$ & $\varnothing$ \\
$p\geq 1, q\geq 2$ & 0 & 1 & 0 & & 0  & 1 & 2 & 4:(\ref{Ad010}) &4:(\ref{Fd010}) & 1:(\ref{Fqd010}) \\ \hline
$p= 1, q\geq 1$ & 0 & 0 & 1 & & & & & $\varnothing$ & $\varnothing$ & $\varnothing$  \\
$p\geq 2, q\geq 1$ & 0 & 0 & 1 & & 0 & 2 & 1 & 4:(\ref{Ad001}) & 4:(\ref{Fd001})& 1:(\ref{Fqd001}) \\ \hline
$p\geq 1, q\geq 1$ & 0 & 0 & 0 &  & 0 & 0 & 0 & $A=0$ & $F=0$ & $F^2=0$  \\
$p\geq 1, q\geq 1$ & 0 & 0 & 0 &  & 0 & 1 & 0 & $\infty$:(\ref{Ad000c}) & $F=0$ & $F^2=0$  \\
$p\geq 1, q\geq 1$ & 0 & 0 & 0 &  & 0 & 0 & 1 & $\infty$:(\ref{Ad000d}) & $F=0$ & $F^2=0$  \\
$p\geq 1, q\geq 1$ & 0 & 0 & 0 &  & 1 & 0 & 0 & $1$:(\ref{Ad000}) & $F=0$ & $F^2=0$  \\
$p\geq 2, q\geq 2$ & 0 & 0 & 0 &  & 2 & 0 & 0 & $1$:(\ref{Ad000a}) & $1$:(\ref{Fd000a}) & $F^2=0$  \\
$p\geq 3, q\geq 3$ & 0 & 0 & 0 &  & 3 & 0 & 0 & $1$:(\ref{Ad000b}) & $1$:(\ref{Fd000b}) & $F^2=0$  \\  \hline
\end{tabular}
\end{center}
\end{table}

\begin{table}[ht]\caption{ The case $\R^{p,q}$, $p\geq 1$, $q\geq 1$, $d_J\neq 0$.}
\begin{center}
  \begin{tabular}{|c|c|c|c|c|c|c|c|c|c|c|}
  \hline
$p, q$ & $d_J$ & $x_J$ & $y_J$ & add.cond. & $d_A$ & $x_A$ & $y_A$ & $A$ & $F$ & $F^2$ \\ \hline
$p\geq 3, q\geq 1$ & 1 & 2 & 0 &  & 1 & 2 & 0 &  1:(\ref{Ad11}) &  1:(\ref{Fd11}) &  1:(\ref{Fqd11}) \\ \hline
$p\geq 1, q\geq 3$ & 1 & 0 & 2 &  & 1 & 0 & 2 &  1:(\ref{Ad12}) &  1:(\ref{Fd12}) &  1:(\ref{Fqd12}) \\ \hline
$p\geq 2, q\geq 2$ & 1 & 1 & 1 & $j_1=j_2$ & & &  &$\varnothing$ & $\varnothing$ & $\varnothing$ \\
$p\geq 2, q\geq 2$ & 1 & 1 & 1 & $j_1\neq j_2$ & 1 & 1 & 1 &  1:(\ref{Ad13}) &  1:(\ref{Fd13}) &  1:(\ref{Fqd13}) \\ \hline
$p\geq 2, q\geq 1$ & 1 & 1 & 0 & &  &  &  & $\varnothing$ & $\varnothing$ & $\varnothing$  \\ \hline
$p\geq 1, q\geq 2$ & 1 & 0 & 1 & &  &  &  & $\varnothing$ &$\varnothing$ &$\varnothing$ \\ \hline
$p=1, q=1$         & 1 & 0 & 0 &  &  &  &  &  $\varnothing$ & $\varnothing$ & $\varnothing$   \\
$p\geq 2, q\geq 1$ & 1 & 0 & 0 &  & 1 & 1 & 0 &  $\infty$:(\ref{Ad14}) & $\infty$:(\ref{Fd14}) & $F^2=0$   \\
$p\geq 1, q\geq 2$ & 1 & 0 & 0 &  & 1 & 0 & 1 &  $\infty$:(\ref{Ad15}) & $\infty$:(\ref{Fd15}) & $F^2=0$   \\ \hline
$p\geq 3, q\geq 2$ & 2 & 1 & 0 & &  &  &  & $\varnothing$ & $\varnothing$ & $\varnothing$ \\ \hline
$p\geq 2, q\geq 3$ & 2 & 0 & 1 & &  &  &  & $\varnothing$  &$\varnothing$ & $\varnothing$ \\ \hline
$p=2, q=2$         & 2 & 0 & 0 & &  &  &  & $\varnothing$ & $\varnothing$ & $\varnothing$  \\
$p\geq 3, q\geq 2$ & 2 & 0 & 0 & & 2 & 1 & 0 & $\infty$:(\ref{Ad21})& $\infty$:(\ref{Fd21})& $F^2=0$  \\
$p\geq 2, q\geq 3$ & 2 & 0 & 0 &  & 2 & 0 & 1 & $\infty$:(\ref{Ad22})& $\infty$:(\ref{Fd22})& $F^2=0$  \\ \hline
$p\geq 3, q\geq 3$ & 3 & 0 & 0 & &  &  &  & $\varnothing$ & $\varnothing$ & $\varnothing$  \\ \hline
\end{tabular}
\end{center}
\end{table}

\newpage
\section{Conclusions.}
\label{sect6}

The main result of this paper is the presentation of all constant solutions of the Yang--Mills equations with $\SU(2)$ gauge symmetry for an arbitrary constant current in pseudo-Euclidean space of arbitrary finite dimension and signature. All solutions in terms of the potential $A$, the strength $F$, and the invariant $F^2$ with references to the corresponding explicit formulas are presented in Tables 2, 3, and 4. Using the invariance of the Yang--Mills equations under the pseudo-orthogonal transformations of coordinates and gauge invariance, we choose a specific system of coordinates and a specific gauge fixing for each constant current and obtain all constant solutions of the Yang--Mills equations in this system of coordinates with this gauge fixing, and then in the original system of coordinates with the original gauge fixing (see Remarks 1 and 2 and the discussion in Section \ref{sect5}). We prove that the number of constant solutions of the Yang--Mills equations in terms of the strength $F$ depends on the parameters $d_J$, $x_J$, $y_J$ and the hyperbolic singular values of the matrix of current $J$. The explicit form of all solutions
and the invariant $F^2$ can always be written using the hyperbolic singular values of the matrix~$J$. Note that sometimes we have solutions with different values of the strength $F$, but the same invariant $F^2$ (compare the columns ``$F$'' and ``$F^2$'' in Tables 2, 3, and 4). These and other facts require further study and comparison of the constant solutions presented in this paper with other known solutions.


We can consider nonconstant solutions of Yang--Mills equations in the form of series of perturbation theory using all constant solutions as a zeroth approximation. For the first approximation, the problem reduces to solving the system of linear partial differential equations with constant coefficients. For the second and subsequent approximations, the problem reduces to solving the systems of linear partial differential equations with variable coefficients. These facts can help us to give a local classification of all (nonconstant) solutions of the classical $\SU(2)$ Yang--Mills equations with an arbitrary non-Abelian current. The results of this paper are new and can be used to solve some problems in particle physics, in particular, to describe physical vacuum \cite{Actor, Gr, Ja, nian}. In this paper, we discuss mathematical structures and constructions. Relating the proposed mathematical constructions to real world objects goes beyond the scope of this investigation. 
An interesting task is to generalize the results of this paper to the case of the Lie group $\SU(3)$ or, more generally, the Lie group $\SU(N)$. The same problems can be considered on curved manifolds.

\section*{Acknowledgements.}

The author is grateful to M.~O.~Katanaev, N.~G.~Marchuk, G.~Raikov, V.~Rosenhaus, D.~Vassiliev, and A.~V.~Zotov for useful comments and discussions.

This work is supported by the Russian Science Foundation (project 21-71-00043), https://rscf.ru/en/project/21-71-00043/.

\section*{Conflict of interest}
This work does not have any conflicts of interest.

\section*{Data availability statement}

Data sharing not applicable to this article as no datasets were generated or analysed during the current study.


\section*{Appendix A: Proof of Lemma \ref{thSh2}.}

The first four cases of Lemma \ref{thSh2} are easily verified.

(v) Let us consider the case $j_1=0$, $j_2\neq 0$, $j_3 \neq 0$. We have
$$b_1 ((b_2)^2-(b_3)^2)=0,\quad b_2((b_1)^2-(b_3)^2)=j_2,\quad b_3 ((b_1)^2+(b_2)^2)=j_3.$$
We obtain $b_1=0$ or $b_2=\pm b_3$. In the case $b_1=0$, we get the unique solution
\begin{eqnarray}
b_1=0,\qquad b_2=-\sqrt[3]{\frac{j_3^2}{j_2}},\qquad b_3=\sqrt[3]{\frac{j_2^2}{j_3}}.\label{tt1}
\end{eqnarray}
In the case $b_2=b_3$, we get
$$b_3((b_1)^2-(b_3)^2)=j_2,\qquad b_3((b_1)^2+(b_3)^2)=j_3.$$
Adding and subtracting both sides of these equations, we obtain
$$2 b_3 (b_1)^2=j_2+j_3,\qquad 2 (b_3)^3=j_3-j_2.$$
If $j_2=j_3$, then $b_3=0$, $j_2=-j_3$, $j_2=j_3=0$, and we obtain a contradiction. If $j_2=-j_3$, then $b_1=0$ or $b_3=0$. If $b_1=0$, then $b_2=b_3=\sqrt[3]{j_3}$. We already have this solution (\ref{tt1}). If $b_3=0$, then $j_3-j_2=0$, $j_2=j_3=0$, and we obtain a contradiction. In the case
\begin{eqnarray}
j_2\neq \pm j_3,\quad \frac{j_2+j_3}{b_3}\geq 0,\label{tt2}
\end{eqnarray}
we get the solutions
$$b_3=b_2=\sqrt[3]{\frac{j_3-j_2}{2}},\qquad b_1=\pm\sqrt[2]{\frac{j_3+j_2}{2b_3}}.$$
Note that (\ref{tt2}) is equivalent to $j_3>j_2$.

In the case $b_2=-b_3$, we get
$$-b_3((b_1)^2-(b_3)^2)=j_2,\qquad b_3((b_1)^2+(b_3)^2)=j_3.$$
Adding and subtracting both sides of these equations, we obtain
$$2 (b_3)^3=j_2+j_3,\qquad 2 b_3 (b_1)^2=j_3-j_2.$$
If $j_2=-j_3$, then $b_3=0$, $j_2=j_3$, $j_2=j_3=0$, and we obtain a contradiction. If $j_2=j_3$, then $b_1=0$ or $b_3=0$. If $b_1=0$, then $b_2=-b_3=-\sqrt[3]{j_2}$. We already have this solution (\ref{tt1}). If $b_3=0$, then $j_2+j_3=0$, $j_2=j_3=0$, and we obtain a contradiction. In the case
\begin{eqnarray}
j_2\neq \pm j_3,\quad \frac{j_3-j_2}{b_3}\geq 0,\label{tt3}
\end{eqnarray}
we get the solutions
$$b_3=-b_2=\sqrt[3]{\frac{j_3+j_2}{2}},\qquad b_1=\pm\sqrt{\frac{j_3-j_2}{2b_3}}.$$
Note that (\ref{tt3}) is equivalent to $j_3>j_2$.

We calculate $K$ for these 4 solutions:
$$K=(b_1b_2b_3)^{\frac{2}{3}}=(\sqrt{\frac{j_3 \pm j_2}{2b_3}}b_3^2)^{\frac{2}{3}}=(\sqrt{\frac{(j_3 \pm j_2)b_3^3}{2}})^{\frac{2}{3}}=\sqrt[3]{\frac{(j_3+j_2)(j_3-j_2)}{4}}=\sqrt[3]{\frac{j_3^2-j_2^2}{4}}.$$

(vi) Let us consider the case of all positive $j_1>0$, $j_2>0$, $j_3>0$. We use the following change of variables:
$$x=b_1\neq 0,\qquad y=\frac{b_2}{b_1}\neq 0,\qquad z=\frac{b_3}{b_1}\neq 0.$$
We obtain
\begin{eqnarray}
 j_1=x^3(y^2-z^2),\quad
 j_2=y x^3(1-z^2),\quad
 j_3=z x^3(1+y^2).
\end{eqnarray}
Using notation
$$A=\frac{ j_2}{ j_1}>0,\qquad B=\frac{ j_3}{ j_1}>0,$$
we get the system for two variables $y$ and $z$:
\begin{eqnarray}y(1-z^2)=A(y^2-z^2),\qquad z(1+y^2)=B(y^2-z^2).\label{2ur0}\end{eqnarray}
For the variable $x=b_1$, we have
\begin{eqnarray}
b_1=x=\sqrt[3]{\frac{ j_1}{(y^2-z^2)}}.\label{x0}
\end{eqnarray}
From the first equation (\ref{2ur0}), we obtain
\begin{eqnarray}
(A-y)z^2=y(Ay-1).\label{670}
\end{eqnarray}
Let us consider two cases: $A=1$ and $A\neq 1$.

In the case $A=1$, we can rewrite $(1-y)z^2=y(y-1)$ in the form $(1-y)(y+z^2)=0$. If $y=1$, then using the second equation (\ref{2ur0}), we get $2z=B(1-z^2)$, i.e. $z^2+\frac{2}{B}z-1=0$. We have $z=\frac{-1\pm\sqrt{1+B^2}}{B}$. If $y=-z^2$, then we substitute this expression into the second equation (\ref{2ur0}) and obtain $z(1+z^4)=B(z^4-z^2)$, $z^5-Bz^4+Bz^2+z=0$. Dividing both sides by $z^3$ and using the notation $s=z-\frac{1}{z}$, we get $s^2-Bs+2=0$, $s=\frac{B\pm\sqrt{B^2-8}}{2}$.  We have $z^2-sz-1=0$ and obtain $z=\frac{s\pm\sqrt{s^2+4}}{2}$ for each $s$. Thus in the case $B\geq 2\sqrt{2}$, we have 2 or 4 additional solutions.

The results for the case $A=1$ can be summarized as follows.

If $j_1= j_2>0$, $ j_3>0$, then we have 2 solutions
$$b_{1\pm}=b_{2\pm}=\sqrt[3]{\frac{ j_1}{(1-z_{\pm}^2)}},\quad b_{3\pm}=z_{\pm} b_{1\pm},\quad z_{\pm}=\frac{-1\pm\sqrt{1+B^2}}{B},\quad B=\frac{ j_3}{ j_1},$$
and in the case $B=2\sqrt{2}$ we have 2 additional solutions ($s^\pm=s$ in this case), and in the case $B>2\sqrt{2}$ we have 4 additional solutions
\begin{eqnarray}
b_{1\pm}^\pm=\sqrt[3]{\frac{ j_1}{(w_{\pm}^\pm)^4-(w_{\pm}^\pm)^2}},\quad b_{2\pm}^\pm=-(w_{\pm}^\pm)^2b_{1\pm}^\pm,\quad b_{3\pm}^\pm=w_{\pm}^\pm b_{1\pm}^\pm,\nonumber\\
w_{\pm}^\pm=\frac{s^{\pm}\pm\sqrt{(s^{\pm})^2+4}}{2},\quad s^\pm=\frac{B\pm\sqrt{B^2-8}}{2},\quad B=\frac{j_3}{j_1}.\nonumber
\end{eqnarray}

Using $z_+ z_-=-1$, $\frac{z_\pm}{1-z_\pm^2}=\frac{B}{2}$, $w_+^\pm w_-^\pm=-1$, $\frac{w_{\pm}^\pm}{(w_{\pm}^\pm)^2-1}=\frac{1}{s^\pm}$, we get the following conserved quantity for the first pair of solutions
$$K:=-b_{1+}b_{1-}=-b_{2+}b_{2-}=b_{3+}b_{3-}= (b_{1+}b_{2+}b_{3+})^{\frac{2}{3}}=(b_{1-}b_{2-}b_{3-})^{\frac{2}{3}}= (\frac{B j_1}{2})^\frac{2}{3},$$
the following conserved quantity for the other two pairs of solutions
$$K:=-b_{1+}^\pm b_{1-}^\pm=-b_{2+}^\pm b_{2-}^\pm=b_{3+}^\pm b_{3-}^\pm= (b_{1+}^\pm b_{2+}^\pm b_{3+}^\pm)^{\frac{2}{3}}=(b_{1-}^\pm b_{2-}^\pm b_{3-}^\pm)^{\frac{2}{3}}= (\frac{ j_1}{s^\pm})^\frac{2}{3}.$$
We can rewrite solutions in the following form
$$b_{1\pm}=b_{2\pm}=\sqrt[3]{\frac{ j_3}{2z_{\pm}}},\qquad b_{3\pm}=z_{\pm} b_{1\pm},\qquad z_{\pm}=\frac{- j_1\pm\sqrt{ j_1^2+ j_3^2}}{ j_3},$$
\begin{eqnarray}
b_{1\pm}^\pm=\frac{1}{w_{\pm}^\pm}b_{3}^\pm,\qquad b_{2\pm}^\pm=-w_{\pm}^\pm b_{3}^\pm,\qquad b_{3}^\pm=b_{\pm}^\pm=\sqrt[3]{\frac{ j_1}{s^\pm}},\nonumber\\
w_{\pm}^\pm=\frac{s^{\pm}\pm\sqrt{(s^{\pm})^2+4}}{2},\qquad s^{\pm}=\frac{ j_3\pm\sqrt{ j_3^2-8 j_1^2}}{2 j_1}.\nonumber
\end{eqnarray}
In the second case, we have also $b_{3+}^\pm=b_{3-}^\pm$, $b_{1\pm}^\pm=b_{2\mp}^\pm$.

Let us consider the case $A\neq 1$. We have $A\neq y$. Really, suppose $A=y$. Since (\ref{670}), we get $A=1$, i.e. a contradiction. Using (\ref{670}), we get
\begin{eqnarray}
z^2=\frac{y(Ay-1)}{A-y}\label{y20}
\end{eqnarray}
Since (\ref{2ur0}), we get
\begin{eqnarray}
z^2(1+y^2)^2=B^2(y^2-z^2)^2.\label{rr4}
\end{eqnarray}
We require that expressions $z$ and $y^2-z^2$ are of the same sign to obtain equivalent transformation. We have $y^2-z^2=y^2-\frac{y(Ay-1)}{A-y}=\frac{y(1-y^2)}{A-y}$, i.e. $z=\lambda \sqrt{\frac{y(Ay-1)}{A-y}}$, where $\lambda=\sign(\frac{y(1-y^2)}{A-y})$. We also need the following condition
$\frac{y(Ay-1)}{A-y}\geq 0$. This condition is satisfied automatically if we take $y$ from (\ref{345}).
We substitute (\ref{y20}) into (\ref{rr4}) and get
\begin{eqnarray}
\frac{y(Ay-1)}{A-y}(1+y^2)^2=B^2(y^2-\frac{y(Ay-1)}{A-y})^2.\label{345}
\end{eqnarray}
We have
$$\frac{y(Ay-1)(1+y^2)^2}{A-y}=\frac{B^2(y-y^3)^2}{(A-y)^2},\quad (Ay-1)(1+y^2)^2(A-y)=B^2y(1-y^2)^2,$$
$$Ay^6+(B^2-A^2-1)y^5+3Ay^4-2(A^2+B^2+1)y^3+3Ay^2+(B^2-A^2-1)y+A=0.$$
Dividing both sides by $y^3$ and using the notation $t=y+\frac{1}{y}$ ($t^2=y^2+\frac{1}{y^2}+2$, $t^3=y^3+3y+3\frac{1}{y}+\frac{1}{y^3}$), we get
\begin{eqnarray}
At^3+(B^2-A^2-1)(t^2-2)-2(A^2+B^2+1)&=&0,\nonumber\\
At^3+(B^2-A^2-1)t^2-4B^2&=&0.\nonumber
\end{eqnarray}
We obtain the following cubic equation
\begin{eqnarray}
f(t):=t^3+(\frac{B^2}{A}-A-\frac{1}{A})t^2-4\frac{B^2}{A}=0.\label{kub0}
\end{eqnarray}
Since $t=y+\frac{1}{y}$, we are only interested in the solutions $t\geq 2$ or $t\leq -2$ of this equation.

Let us find extrema of the function $f(t)$: $f'(t)=3t^2+2t(\frac{B^2}{A}-A-\frac{1}{A})=0$. The function has one extremum at the point $t=0$, which is $f(0)=-\frac{4B^2}{A}<0$, and another extrema at the point $T=-\frac{2(B^2-A^2-1)}{3A}$, which is $f(T)=\frac{4(B^2-A^2-1)^3}{27A^3}-\frac{4B^2}{A}$.

Also we have $f(-\infty)=-\infty$, $f(-2)=\frac{-4(A+1)^2}{A}<0$, $f(2)=\frac{-4(A-1)^2}{A}<0$, $f(A+\frac{1}{A})=\frac{B^2(A^2-1)^2}{A^3}>0$, $f(+\infty)=+\infty$. This means that the equation (\ref{kub0}) always has one solution $2<t_1<A+\frac{1}{A}$.

Suppose $-2<T$ and $f(T)>0$. These conditions take the form $(B^2-A^2-1)<3A$ and $(B^2-A^2-1)^3>27A^2 B^2$. Raising the first inequality in the third degree and comparing with the second one, we get $27A^2 B^2<27A^3$, i.e. $B^2<A$. On the other hand, the second inequality implies $B^2>A^2+1$. We obtain $A^2+1<A$, $A^2-A+1<0$, i.e. a contradiction.

This means that if the cubic equation has another solution except the first one $t_1>2$, then it is less than $-2$. Thus all solutions of (\ref{kub0}) are suitable for us. The number of these solutions is from $1$ to $3$ and it depends on $A$ and $B$. We have two solutions in the case $f(T)=0$, i.e. $(B^2-A^2-1)^3=27A^2 B^2$. We have one solution in the case $(B^2-A^2-1)^3<27A^2 B^2$, we have three solutions in the case $(B^2-A^2-1)^3>27A^2 B^2$. We can get the explicit form of these solutions using the Cardano formulas.

Let us prove that $(B^2-A^2-1)^3>27A^2 B^2$ is equivalent to $B^{\frac{2}{3}}>A^{\frac{2}{3}}+1$. Considering the equation $(B^2-A^2-1)^3-27A^2 B^2=0$ as a cubic equation for $B^2$, we can find one of the roots $B^2=(A^{\frac{2}{3}}+1)^3$. We get
$$
(B^2-A^2-1)^3-27A^2 B^2=(B^2-(a+1)^3)(B^4+B^2((a+1)^3-3(1+a^3))+(a^2-a+1)^3),
$$
where $a:=A^{\frac{2}{3}}$. This equation has no other roots because the discriminant for the corresponding quadratic equation is negative: $D=-27a^2(a-1)^2<0$ (we use $0<a=A^{\frac{2}{3}}\neq 1$).

Finally, in the case  $A\neq 1$, we have the solutions
\begin{eqnarray}
b_{1\pm}=\sqrt[3]{\frac{ j_1}{(y_{\pm}^2-z_{\pm}^2)}},\qquad b_{2\pm}=y_{\pm} b_{1\pm},\qquad b_{3\pm}=z_{\pm} b_{1\pm},\nonumber\\ z_{\pm}=\lambda_{\pm}\sqrt{\frac{y_\pm(Ay_\pm-1)}{A-y_\pm}},\qquad y_{\pm}=\frac{t_k\pm\sqrt{t_k^2-4}}{2},\nonumber
\end{eqnarray}
where $\lambda_\pm=\sign(\frac{y_\pm(1-y_\pm^2)}{A-y_\pm})$ and $t_k=t_k(A,B)$, $k=1, 2, 3$, are solutions of (\ref{kub0}).
Using $\frac{y_{\pm}z_{\pm}}{y_{\pm}^2-z_{\pm}^2}=\frac{B y_\pm}{1+y_\pm^2}=\frac{B}{t_k}$, $y_+ y_-=1$, $\lambda_+ \lambda_-=-1$, $z_+ z_-=-1$, we obtain the following conserved quantity
$$K:=-b_{1+}b_{1-}=-b_{2+}b_{2-}=b_{3+}b_{3-}= (b_{1+}b_{2+}b_{3+})^{\frac{2}{3}}=(b_{1-}b_{2-}b_{3-})^{\frac{2}{3}}= (\frac{ j_3}{t_k})^\frac{2}{3},$$
and can rewrite the solutions in the form
\begin{eqnarray}
b_{1\pm}=\sqrt[3]{\frac{ j_3}{t_k y_\pm z_\pm}},\qquad b_{2\pm}=y_{\pm} b_{1\pm},\qquad b_{3\pm}=z_{\pm} b_{1\pm},\nonumber\\ z_{\pm}=\lambda_\pm\sqrt{\frac{y_\pm(Ay_\pm-1)}{A-y_\pm}},\qquad y_{\pm}=\frac{t_k\pm\sqrt{t_k^2-4}}{2},\nonumber
\end{eqnarray}
where $\lambda_\pm=\sign(\frac{y_\pm(1-y_\pm^2)}{A-y_\pm})$, and $t_k=t_k(A,B)$ are solutions of (\ref{kub0}). We have from $1$ to $3$ such numbers $t_k$ and from $2$ to $6$ solutions of the considered system of equations.

We verify $y_+ y_-=1$ by direct calculation. Let us verify $\lambda_+ \lambda_-=-1$ and $z_+ z_-=-1$. We have
\begin{eqnarray}
\!\!\!\lambda_+ \lambda_-&=&\sign(\frac{y_+(1-y_+^2)}{A-y_+})\sign(\frac{y_-(1-y_-^2)}{A-y_-})=\sign(\frac{y_+y_-(1-y_+^2)(1-y_-^2)} {(A-y_+)(A-y_-)})\nonumber\\
&=&\sign(\frac{1-y_-^2-y_+^2+(y_+y_-)^2}{A^2-A(y_-+y_+)+y_+y_-})=\sign(\frac{2-(y_-^2+y_+^2)} {A^2+1-A(y_++y_-)})\nonumber\\
&=&\sign(\frac{4-t_k^2}{A^2+1-At_k})=-1,\nonumber
\end{eqnarray}
because the numerator is negative ($|t_k|>2$) and the denominator  is positive ($t_k<A+\frac{1}{A}$). We use $f(A+\frac{1}{A})=\frac{B^2(A^2-1)^2}{A^3}>0$, and this means that the largest solution of the cubic equation is between $2$ and $A+\frac{1}{A}$.

We have
$$z_+ z_-=\lambda_+ \lambda_- \sqrt{\frac{y_+y_-(Ay_+-1)(Ay_--1)}{(A-y_+)(A-y_-)}}=-\sqrt{\frac{(A^2-A(y_++y_-)+1)}{A^2-A(y_++y_-)+1}}=-1.$$

The lemma is proved.

\section*{Appendix B: Proof of Lemma \ref{lemmaF222}.}

1) For the first type of solutions in Cases (vi) - (a), (b), (c) of Lemma \ref{thSh2}, using
$$b_{1\pm}=b_{2\pm},\qquad b_{3\pm}=z_{\pm}b_{1\pm},\qquad b_{1\pm}=\sqrt[3]{\frac{j_3}{2z_{\pm}}},\qquad K=(\frac{j_3}{2})^\frac{2}{3},$$
we obtain
\begin{eqnarray}
F_{\pm}^2&=&-\frac{1}{2}((b_{1\pm}b_{2\pm})^2-(b_{2\pm}b_{3\pm})^2-(b_{3\pm}b_{1\pm})^2)I_2= -\frac{1}{2}b_{1\pm}^4(1-2z_{\pm}^2)I_2\nonumber\\
&=&-\frac{1}{2}(\frac{j_3}{2z_{\pm}})^{\frac{4}{3}}(1-2z_{\pm}^2)I_2=\frac{K^2 (2z_{\pm}^2)-1}{2 z_{\pm}^{\frac{4}{3}}}I_2.\nonumber
\end{eqnarray}
Let us prove that $F^2_+ \neq F^2_-$ in this case. Suppose we have $F^2_+=F^2_-$, i.e.
$$
\frac{1-2z_+^2}{z_+^{\frac{4}{3}}}=\frac{1-2z_-^2}{z_-^{\frac{4}{3}}},\qquad
z_-^{\frac{4}{3}}-2z_+^2z_-^{\frac{4}{3}}=z_+^{\frac{4}{3}}-2z_-^2z_+^{\frac{4}{3}}.
$$
Using $z_+ z_-=-1$, we get
$$
z_-^{\frac{4}{3}}-2z_+^{\frac{2}{3}}=z_+^{\frac{4}{3}}-2z_-^{\frac{2}{3}},\qquad (z_-^{\frac{2}{3}}+1)^2=(z_+^{\frac{2}{3}}+1)^2,\qquad
(z_-^{\frac{2}{3}}+z_+^{\frac{2}{3}}+2)(z_-^{\frac{2}{3}}-z_+^{\frac{2}{3}})=0,
$$
which is not possible, because $z_+$, $z_-$ do not equal $\pm 1$.

We see that $F^2_{\pm}=0$ if and only if $z_{\pm}=\pm\frac{1}{\sqrt{2}}$. Using $2z=B(1-z^2)$, we get $B=\frac{j_3}{j_1}=2\sqrt{2}$, i.e. $F^2$ is zero only in Case (b) for this type of solutions.

If $\frac{j_3}{j_1}=2\sqrt{2}$, then $z_+=\frac{1}{\sqrt{2}}$, $z_-=-\sqrt{2}$. For $z_+=\frac{1}{\sqrt{2}}$, we obtain $F^2_+=0$. For $z_-=-\sqrt{2}$, we obtain
\begin{eqnarray}
F^2_-=\frac{3 K^2 }{2^{\frac{2}{3}} 2}I_2=\frac{3j_3^{\frac{4}{3}}}{8}I_2.\label{rr3}
\end{eqnarray}
2) For the second type of solutions in Cases (vi) - (a), (b), (c), using
$$c^\pm_{1\pm}=\frac{1}{w^\pm_{\pm}}c_3^\pm,\qquad c^\pm_{2\pm}=-w^\pm_{\pm}c^\pm_{3},\qquad c^\pm_{3\pm}=c^\pm_{3}=\sqrt[3]{\frac{j_1}{s^\pm}},\qquad K=(\frac{j_1}{s^\pm})^{\frac{2}{3}},$$
we obtain
$$
(F^\pm_{\pm})^2=-\frac{1}{2}((c^\pm_{1\pm}c^\pm_{2\pm})^2-(c^\pm_{2\pm}c^\pm_{3\pm})^2-(c^\pm_{3\pm}c^\pm_{1\pm})^2)I_2= -\frac{1}{2}(c^\pm_{3\pm})^4(1-(w^\pm_{\pm})^2-\frac{1}{(w^\pm_{\pm})^2})I_2$$
$$=-\frac{1}{2}(\frac{j_1}{s^\pm})^{\frac{4}{3}}(1-(w^\pm_{\pm})^2-\frac{1}{(w^\pm_{\pm})^2})I_2=-\frac{K^2}{2} (1-(w^\pm_{\pm})^2-\frac{1}{(w^\pm_{\pm})^2})I_2=\frac{K^2}{2} ((s^\pm)^2+1)I_2>0.$$
In the last equality, we use $w_{\pm}-\frac{1}{w_{\pm}}=s$, i.e. $w_{\pm}^2+\frac{1}{w_{\pm}^2}=s^2+2$.

We have $(F^\pm_+)^2=(F^\pm_-)^2$, because $(F^\pm_{\pm})^2$ does not depend on $w_{\pm}$.

Let us prove that $(F_+^+)^2=(F_-^+)^2\neq(F_+^-)^2=(F_-^-)^2$. Suppose we have $(F^+)^2=(F^-)^2$, i.e.
$$
\frac{1+(s^+)^2}{(s^+)^{\frac{4}{3}}}=\frac{1+(s^-)^2}{(s^-)^{\frac{4}{3}}},\qquad
(s^-)^{\frac{4}{3}}+(s^+)^2 (s^-)^{\frac{4}{3}}=(s^+)^{\frac{4}{3}}+(s^-)^2 (s^+)^{\frac{4}{3}}.
$$
Using $s^+ s^-=2$ (because $s^2-Bs+2=0$), we get
$$
(s^-)^{\frac{4}{3}}+2^{\frac{4}{3}}(s^+)^{\frac{2}{3}}=(s^+)^{\frac{4}{3}}+2^{\frac{4}{3}}(s^-)^{\frac{2}{3}},\qquad ((s^-)^{\frac{2}{3}}-2^{\frac{1}{3}})^2=((s^+)^{\frac{2}{3}}-2^{\frac{1}{3}})^2,$$
$$
((s^+)^{\frac{2}{3}}+(s^-)^{\frac{2}{3}}-2^{\frac{4}{3}})((s^+)^{\frac{2}{3}}-(s^-)^{\frac{2}{3}})=0,
$$
which is not possible, because $(s^+)^{\frac{2}{3}}+(s^-)^{\frac{2}{3}}\geq 2\sqrt{(s^+ s^-)^\frac{2}{3}}=2^{\frac{4}{3}}$ and
if $s^+=s^-$, then $s^+=s^-=\sqrt{2}$, $B=2\sqrt{2}$, and we obtain a contradiction.

In the particular case $B=\frac{j_3}{j_1}$, we obtain $s=s^\pm=\sqrt{2}$, $w_\pm=w^\pm_\pm=\frac{\sqrt{2}\pm\sqrt{6}}{2}$, and
$$(F_{\pm})^2=\frac{3}{2}(\frac{j_1}{\sqrt{2}})^\frac{4}{3}I_2=\frac{3 j_3^{\frac{4}{3}}}{2^{\frac{11}{3}}}I_2,$$
which does not coincide with (\ref{rr3}).

3) Let us show that $F^2$ can be the same for the first type of solutions $b_k$ and for the second type of solutions $c_k$ in Case (vi) - (a). Suppose
$$
\frac{j_3^{\frac{4}{3}}(2z^2-1)}{2^{\frac{4}{3}}z^{\frac{4}{3}}}=\frac{j_1^\frac{4}{3}(1+s^2)}{s^{\frac{4}{3}}}.
$$
Using Wolfram Mathematica 11.1, we solved the system
$$(2z^2-1)^3(Bs)^4=(1+s^2)^3(2z)^4,\quad Bz^2+2z-B=0,\quad s^2-Bs+2=0,\quad B>0.$$
This system of equation has the following unique solution
$$s^*=\sqrt{13+\sqrt{193-6^\frac{4}{3}}+\sqrt{386+6^\frac{4}{3}+\frac{5362}{193-6^\frac{4}{3}}}}\approx 7.39438,$$
$$B^*=\frac{(s^*)^2+2}{s^*}\approx 7.66486,\qquad z^*=\frac{-1+\sqrt{1+(B^*)^2}}{B^*}\approx 0.878009.$$
This means that if $B=B^*$, then $(F^+_{\pm})^2$ for $(s^+)^*=\frac{B^*+\sqrt{(B^*)^2-8}}{2}$ coincides with $F_+^2$ for $z_+^*=\frac{-1+\sqrt{1+(B^*)^2}}{B^*}$.

4) In Cases (vi) - (d), (e), (f), using
$$d_{2\pm}=y_{\pm}d_{1\pm},\qquad d_{3\pm}=z_{\pm}d_{1\pm},\qquad d_{1\pm}=\sqrt[3]{\frac{j_3}{t_k y_{\pm}z_{\pm}}},\qquad K=(\frac{j_3}{t_k})^{\frac{2}{3}},$$
we obtain
$$F^2_{\pm}=-\frac{1}{2}((d_{1\pm}d_{2\pm})^2-(d_{2\pm}d_{3\pm})^2-(d_{3\pm}d_{1\pm})^2)I_2= -\frac{1}{2}d_{1\pm}^4(y_{\pm}^2-z_{\pm}^2 - y_{\pm}^2 z_{\pm}^2)I_2
$$
\begin{eqnarray}
=-\frac{1}{2}(\frac{j_3}{t_k y_{\pm}z_{\pm}})^{\frac{4}{3}}(y_{\pm}^2-z_{\pm}^2 -y_{\pm}^2 z_{\pm}^2)I_2=
\frac{K^2( y_{\pm}^2 z_{\pm}^2-y_{\pm}^2+z_{\pm}^2)}{2(y_{\pm}z_{\pm})^{\frac{4}{3}}}I_2.\label{ty2}
\end{eqnarray}
Using $z_\pm^2=\frac{y_{\pm}(Ay_{\pm}-1)}{A-y_{\pm}}$, we get
\begin{eqnarray}
F^2_{\pm}&=&\frac{K^2( y_{\pm}^2 z_{\pm}^2-y_{\pm}^2+z_{\pm}^2)}{2(y_{\pm}z_{\pm})^{\frac{4}{3}}}I_2=
\frac{K^2((1+ y_{\pm}^2)\frac{y_{\pm}(Ay_{\pm}-1)}{A-y_{\pm}}-y_{\pm}^2)}{2y_{\pm}^{\frac{4}{3}}(\frac{y_{\pm}(Ay_{\pm}-1)}{A-y_{\pm}})^{\frac{2}{3}}} I_2\nonumber\\
&=&\frac{K^2(Ay_{\pm}^3-1)}{2y_{\pm}(A-y_{\pm})^{\frac{1}{3}}(Ay_{\pm}-1)^{\frac{2}{3}}}I_2.\label{ty1}
\end{eqnarray}
Let us prove that $F^2_+ \neq F^2_-$ in this case. Suppose we have $F^2_+=F^2_-$, i.e.
$$
\frac{(1-Ay_{+}^3)^3}{y^3_{+}(A-y_{+})(1-Ay_{+})^2}=\frac{(1-Ay_{-}^3)^3}{y^3_{-}(A-y_{-})(1-Ay_{-})^2}.
$$
Using $y_-=y_+^{-1}$, we get
$$y_+^3(A-y_+)(1-Ay_+)^2(1-\frac{A}{y_+^3})^3=(1-Ay_+^3)^3\frac{1}{y_+^3}(A-\frac{1}{y_+})(1-\frac{A}{y_+})^2,$$
$$(y_+^3-A)^3(1-Ay_+)=(1-Ay_+^3)^3(y_+-A),$$
$$(A^3-A)y_+^{10}+(1-A^4)y_+^9+3(A^3-A)y_+^6+3(A-A^3)y_+^4+(A^4-1)y_++(A-A^3)=0.$$
Dividing both sides of the equation by $y_+^5 (A^2-1)\neq 0$, we obtain
$$A(y_+^5-\frac{1}{y_+^5})-(1+A^2)(y_+^4-\frac{1}{y_+^4})+3A(y_+ -\frac{1}{y_+})=0.$$
Dividing both sides of the equation by $(y_+-\frac{1}{y_+})\neq 0$, we get
$$A(y_+^4+y_+^2+1+\frac{1}{y_+^2}+\frac{1}{y_+^4})-(1+A^2)(y_+^3+y_++\frac{1}{y_+}+\frac{1}{y_+^3})+3A=0.$$
Using $t=y_+ +\frac{1}{y_+}=y_+ + y_-$, we have
$$y_+^2+\frac{1}{y_+^2}=t^2-2,\quad y_+^3+\frac{1}{y_+^3}=t^3-3t,\quad y_+^4+\frac{1}{y_+^4}=t^4-4t^2+2,$$
and obtain
$$At^4-(1+A^3)t^3-3At^2+2(1+A^2)t+4A=0.$$
Dividing by $t\neq 0$, we get
$$A(t^2+\frac{4}{t^2})-(1+A^2)(t-\frac{2}{t})-3A=0.$$
Using $d:=t-\frac{2}{t}$, we have $t^2+\frac{4}{t^2}=d^2+4$ and obtain
$$Ad^2-(1+A^2)d+A=0,\qquad \mbox{i.e. $d=A$, $d=\frac{1}{A}$.}$$
If $d=A$, then
\begin{eqnarray}
t^2-At-2=0.\label{QQ1}
\end{eqnarray}
But it is in a contradiction with
\begin{eqnarray}
At^3+(B^2-A^2-1)t^2-4B^2=0.\label{QQ2}
\end{eqnarray}
Really, multiplying both sides of (\ref{QQ1}) by $At$, we get
\begin{eqnarray}
At^3-A^2t^2-2At=0.\label{QQ3}
\end{eqnarray}
From (\ref{QQ2}) and (\ref{QQ3}), we obtain
\begin{eqnarray}
(B^2-1)t^2+2At-4B^2=0.\label{QQ4}
\end{eqnarray}
From (\ref{QQ4}) and (\ref{QQ1}), we get
$$\frac{4B^2-2At}{B^2-1}=At+2,\quad tA(1+B^2)=2(1+B^2),\quad t=\frac{2}{A}.$$
Substituting $t=\frac{2}{A}$ into (\ref{QQ1}), we get $\frac{4}{A^2}=4$, i.e. a contradiction, because $A\neq \pm 1$.

If $d=A^{-1}$, then
\begin{eqnarray}
At^2-t-2A=0.\label{Q1}
\end{eqnarray}
But it is in a contradiction with
\begin{eqnarray}
At^3+(B^2-A^2-1)t^2-4B^2=0.\label{Q2}
\end{eqnarray}
Really, multiplying both sides of (\ref{Q1}) by $t$, we get
\begin{eqnarray}
At^3-t^2-2At=0.\label{Q3}
\end{eqnarray}
From (\ref{Q2}) and (\ref{Q3}), we obtain
\begin{eqnarray}
(B^2-A^2)t^2+2At-4B^2=0.\label{Q4}
\end{eqnarray}
From (\ref{Q4}) and (\ref{Q1}), we get
$$\frac{4B^2-2At}{B^2-A^2}=\frac{t+2A}{A},\quad t(A^2+B^2)=2A(A^2+B^2),\quad t=2A.$$
Substituting $t=2A$ into (\ref{Q1}), we get $4A(A^2-1)=0$, i.e. a contradiction, because $A\neq 0$, $A\neq \pm 1$.

5) Using (\ref{ty1}), we see that we have $F^2=0$ only if
\begin{eqnarray}
y=\sqrt[3]{\frac{1}{A}}.\label{y7}
\end{eqnarray}
Using (\ref{ty2}), we see that it is equivalent to
\begin{eqnarray}
y^2-z^2=y^2z^2.\label{70}
\end{eqnarray}
Substituting this expression into the second equation (\ref{2ur0}), we get $z(1+y^2)=By^2z^2$. Comparing this equation and $z^2(1+y^2)=y^2$, we obtain
\begin{eqnarray}
z=\sqrt[3]{\frac{1}{B}}.\label{z7}
\end{eqnarray}
Finally, substituting (\ref{y7}) and (\ref{z7}) into (\ref{70}), we get
$$\sqrt[3]{\frac{1}{A^2}}-\sqrt[3]{\frac{1}{B^2}}=\sqrt[3]{\frac{1}{A^2B^2}},$$
which is equivalent to $B^{\frac{2}{3}}=1+A^{\frac{2}{3}}$, i.e. $j_3^\frac{2}{3}=j_1^\frac{2}{3}+j_2^\frac{2}{3}$.

If $B^{\frac{2}{3}}=1+A^{\frac{2}{3}}$ (which is equivalent to $(B^2-A^2-1)=3A^\frac{2}{3}B^\frac{2}{3}$), then the cubic equation $At^3+(B^2-A^2-1)t^2-4B^2=0$ takes the form
$$At^3+3A^\frac{2}{3}B^\frac{2}{3}t^2-4B^2=0.$$
Dividing both sides of the equation by $A\neq 0$ and using $\beta:=\frac{B^\frac{2}{3}}{A^\frac{1}{3}}$, we get $t^3+3\beta t^2-4\beta^3=0$, which is equivalent to
$(t-\beta)(t+2\beta)^2=0$. We obtain $t=\beta=\frac{B^\frac{2}{3}}{A^\frac{1}{3}}=A^{\frac{1}{3}}+A^{-\frac{1}{3}}$ and $t=-2\beta=-2\frac{B^\frac{2}{3}}{A^\frac{1}{3}}=-2(A^{\frac{1}{3}}+A^{-\frac{1}{3}})$.

In the first case $t_1=A^{\frac{1}{3}}+A^{-\frac{1}{3}}$, using $t=y+y^{-1}$, we conclude that $y=A^\frac{1}{3}$, $A^{-\frac{1}{3}}$. If $y=A^\frac{1}{3}$, then $$z^2=\frac{y(Ay-1)}{A-y}=\frac{A^\frac{4}{3}-1}{A^\frac{2}{3}-1}=A^\frac{2}{3}+1=B^\frac{2}{3},$$
\begin{eqnarray}
\!\!\!\!\!\!\!\!\!\!\!\!\!\!\!\!\!\!F^2=\frac{j_3^\frac{4}{3}(y^2z^2+z^2-y^2)}{2 t^\frac{4}{3}y^\frac{4}{3}z^\frac{4}{3}}I_2=\frac{j_3^\frac{4}{3}(A^\frac{4}{3}+A^\frac{2}{3}+1)}{2 B^{\frac{4}{3}}}I_2= \frac{j_1^\frac{4}{3}+j_1^\frac{2}{3}j_2^\frac{2}{3}+j_2^\frac{4}{3}}{2}I_2> 0.\label{io}\end{eqnarray}
If $y=A^{-\frac{1}{3}}$, then
$$z^2=\frac{y(Ay-1)}{A-y}=\frac{A^\frac{2}{3}-1}{A^\frac{4}{3}-1}=\frac{1}{A^\frac{2}{3}+1}=B^{-\frac{2}{3}},$$
$$y^2z^2+z^2-y^2=A^{-\frac{2}{3}}B^{-\frac{2}{3}}+B^{-\frac{2}{3}}-A^{-\frac{2}{3}}= \frac{B^\frac{2}{3}-A^\frac{2}{3}-1}{A^\frac{2}{3}B^\frac{2}{3}}=0,\qquad F^2=0.$$
In the second case $t_2=-2(A^{\frac{1}{3}}+A^{-\frac{1}{3}})$, we get $y\neq A^{-\frac{1}{3}}$ and
$F_{\pm 2}^2\neq 0.$ Using Wolfram Mathematica 11.1, we verified that each of two cumbersome expressions $F_{\pm 2}^{2}$ does not coincide with $F_{-1}^2$ (\ref{io}).


The lemma is proved.


\section*{References.}

\begin{thebibliography}{15}
\bibitem{Actor} Actor A.: {\it Classical solutions of SU(2) Yang--Mills theories}. Rev. Mod. Phys. {\bf 51}, 461--525 (1979).
\bibitem{deA}de Alfaro V., Fubini S., Furlan G.: {\it A new classical solution of the Yang--Mills field equations}. Phys. Lett. B {\bf 65}, 163 (1976).
\bibitem{ADHM} Atiyah M., Drinfeld V., Hitchin N., Manin Yu.: {\it Construction of instantons}. Physics Letters A, {\bf 65} (3), 185--187 (1978).
\bibitem{Bel} Belavin A.A., Polyakov A.M., Schwartz A.S., Tyupkin Yu.S.: {\it Pseudoparticle solutions of the Yang--Mills equations}. Phys. Lett. B {\bf 59}, 85 (1975).
\bibitem{Bojan2} Bojanczyk A.W., Onn R., Steinhardt A.O.: {\it Existence of the hyperbolic singular value~decomposition}. Linear Algebra and its Applications. {\bf 185}, 21--30 (1993).
\bibitem{C3} Chang S., Weiss N.: {\it Instability of constant Yang--Mills fields}, Physical Review D, {\bf 20}, 4 (1979).
\bibitem{Fad} Faddeev L.D., Slavnov A.A.: {\it Gauge Fields: An Introduction to Quantum Theory}, 2nd ed., CRC Press, 2018.
\bibitem{SVD1} Forsythe G.E., Malcolm M.A., Moler C.B.: {\it Computer Methods for Mathematical Computations}, Prentice Hall, Upper Saddle River, 1977.
\bibitem{SVD2} Golub G., Van Loan C.: {\it Matrix Computations} (3rd ed.), Johns Hopkins University Press Baltimore, MD, USA, 1996.
\bibitem{GN} Gorsky A., Nekrasov N.: {\it Hamiltonian systems of Calogero type and two dimensional Yang--Mills theory}, Nucl.Phys. B {\bf 414}, 213--238 (1994).
\bibitem{Gr} Greensite J.P. {\it Calculation of the Yang--Mills vacuum wave functional}, Nuclear Physics B {\bf 158}
(2)-(3), 469--496 (1979).
\bibitem{tH} 't Hooft G.: {\it Magnetic Monopoles in Unified Gauge Theories}, Nucl.Phys. B {\bf 79}, 276--284 (1974).
\bibitem{Ja} Jackiw R., Rebbi C. {\it Vacuum Periodicity in a Yang--Mills Quantum Theory}, Phys. Rev. Lett. {\bf 37},
172 (1976).
\bibitem{YMM} Marchuk N.G.: {\it On a field equation generating a new class of particular solutions to the
Yang--Mills equations}. Tr. Mat. Inst. Steklova, {\bf 285}, 207--220 (2014) [Proceedings of the Steklov Institute of Mathematics, {\bf 285}, pp. 197--210 (2014).]
\bibitem{YMP} Marchuk N.G., Shirokov D.S.: {\it Constant Solutions of Yang--Mills Equations and Generalized Proca Equations}. J. Geom. Symmetry Phys. \textbf{42}, 53--72 (2016).
\bibitem{PFE} Marchuk N.G., Shirokov D.S.: {\it General solutions of one class of field equations}. Rep. Math. Phys. \textbf{78}:3, 305--326    (2016).
\bibitem{nian} Nian J., Qian Y.: {\it A topological way of finding solutions to Yang--Mills equations}, Communications in Theoretical Physics {\bf 72}, 8 (2020).
\bibitem{C1} Martin P. P., Mecklenburg W.: {\it Constant Yang--Mills fields}, Physical Review D, {\bf 22}, 12 (1980).
\bibitem{Bojan} Onn R., Steinhardt A.O., Bojanczyk A.W.: {\it The hyperbolic singular value decomposition and applications}. IEEE Trans. Signal Proc. {\bf 39}, 1575--1588 (1991).
\bibitem{Pol} Polyakov A.M.: {\it Isomeric states of quantum fields} Sov.Phys. - JETP, {\bf 41}, 988--995 (1975).
    \bibitem{Sch} Schimming R.: {\it On constant solutions of the Yang--Mills equations}. Arch. Math. \textbf{24}:2, 65--73 (1988).
\bibitem{Sch2} Schimming R., Mundt E.: {\it Constant potential solutions of the Yang--Mills equation}. J. Math. Phys. \textbf{33}, 4250 (1992).
\bibitem{Shirokov1} Shirokov D.: {\it On constant solutions of $\SU(2)$ Yang--Mills equations with arbitrary current in Euclidean space $\R^n$}. Journal of Nonlinear Mathematical Physics, \textbf{27}:2 (2020).
\bibitem{Shirokovhsvd} Shirokov D.: {\it A note on the hyperbolic singular value decomposition without hyperexchange matrices}. Journal of Computational and Applied Mathematics, 391 (2021), 113450.
\bibitem{YMSh} Shirokov D.S.: {\it Covariantly constant solutions of the Yang--Mills equations}, Advances in Applied Clifford Algebras, \textbf{28}, 53 (2018).
\bibitem{hforms} Shirokov D.S.: {\it On solutions of the Yang--Mills equations in the algebra of h-forms}, Journal of Physics: Conference Series, 2099, IOP Publishing (2021), 012015.
\bibitem{ProcaYM} Shirokov D.S.: {\it Hyperbolic singular value decomposition in the study of Yang--Mills and Yang--Mills--Proca equations}, Computational Mathematics and Mathematical Physics, {\bf 62}, 6 (2022).
\bibitem{C4} Sikivie P.: {\it Instability of Abelian field configurations in Yang--Mills theory}, Physical Review D, {\bf 20}, 4 (1979).
\bibitem{Simonov} Simonov Yu.A., Sov. J. Nucl. Phys. {\bf 41} (1985) 835.
\bibitem{Wit} Witten E.: {\it Some Exact Multipseudoparticle Solutions of Classical Yang--Mills Theory}, Phys. Rev. Lett., \textbf{38}, 121 (1977).
\bibitem{C2} Tudon T.N.: {\it Instability of constant Yang--Mills fields generated by constant gauge potentials}, Physical Review D, \textbf{22}, 10 (1980).
\bibitem{WYa} Wu T.T., Yang C.N.: in {\it Properties of Matter Under Unusual Conditions}, edited by H. Mark and S. Fernbach, Interscience New York, 1968.
\bibitem{Zhdanov} Zhdanov R.Z., Lahno V.I.: {\it Symmetry and Exact Solutions of the Maxwell and SU(2) Yang--Mills Equations}. Adv. Chem. Phys. Modern Nonlinear Optics {\bf  119} part II, 269--352 (2001).
\end{thebibliography}
\end{document}